\documentclass[conference]{IEEEtran}
\ifCLASSINFOpdf
\else
\fi
\usepackage{amsfonts}
\usepackage{amsmath,amsthm}
\usepackage{graphicx,epstopdf,epsfig,multirow,epic,bm}
\usepackage{color}
\usepackage{multicol}
\usepackage{CJK}
\usepackage{makeidx,showidx}
\usepackage{ifthen}
\usepackage{mathrsfs}
\usepackage{colortbl}
\usepackage{multicol}
\usepackage{latexsym}
\usepackage{amssymb}
\usepackage{algorithm}
\usepackage{algorithmic}
\usepackage{longtable}
\usepackage{array}
\usepackage{enumerate}
\usepackage{paralist}
\usepackage{fancyhdr}

\theoremstyle{plain}
\newtheorem{thm}{Theorem}
\newtheorem{cor}[thm]{Corollary}
\newtheorem{lem}[thm]{Lemma}

\newtheorem{conj}[thm]{Conjecture}
\theoremstyle{definition}
\newtheorem{defn}{Definition}
\theoremstyle{plain}

\theoremstyle{definition}

\theoremstyle{definition}
\newtheorem{expr}{Experiment}

\theoremstyle{definition}

\newcommand{\qqed}{\hfill $\Box$}

\begin{document}

%

\title{Topological Graphic Passwords And Their Matchings Towards Cryptography}



\author{\IEEEauthorblockN{Bing  Yao$^{1,5}$,~Hui Sun$^{1}$,~Xiaohui Zhang$^{1}$,~Yarong Mu$^{1}$,~Yirong Sun$^{1}$\\ Hongyu Wang$^{2,\ddagger}$,~Jing Su$^{2}$, Mingjun Zhang$^{3}$, Sihua Yang$^{3}$, Chao Yang$^{4}$}
\IEEEauthorblockA{{1} College of Mathematics and Statistics,
 Northwest Normal University, Lanzhou, 730070, CHINA}
 \IEEEauthorblockA{{2} School of Electronics Engineering and Computer Science, Peking University, Beijing, 100871, CHINA}
\IEEEauthorblockA{{3} School of Information Engineering, Lanzhou University of Finance and Economics, Lanzhou, 730030, CHINA}
\IEEEauthorblockA{{4} School of Mathematics, Physics \& Statistics, Shanghai University of Engineering Science, Shanghai, 201620, CHINA}
\IEEEauthorblockA{{5} School of Electronics and Information Engineering, Lanzhou Jiaotong University, Lanzhou, 730070, CHINA\\
$^\ddagger$ Hongyu Wang is the corresponding author with email: why1988jy@163.com}

\thanks{Manuscript received June 1, 2017; revised August 26, 2017.
Corresponding author: Bing Yao, email: yybb918@163.com.}}


%


\maketitle

\begin{abstract}
Graphical passwords (GPWs) are convenient for mobile equipments with touch screen. Topological graphic passwords (Topsnut-gpws) can be saved in computer by classical matrices and run quickly than the existing GPWs. We research Topsnut-gpws by the matching of view, since they have many advantages. We discuss: configuration matching partition, coloring/labelling matching partition, set matching partition, matching chain, etc. And, we introduce new graph labellings for enriching Topsnut-matchings and show that these labellings can be realized for trees or spanning trees of networks. In theoretical works we explore Graph Labelling Analysis, and show that every graph admits our extremal labellings and set-type labellings in graph theory. Many of the graph labellings mentioned are related with problems of set  matching partitions to number theory, and yield new objects and new problems to graph theory.\\[4pt]
\end{abstract}
\textbf{\emph{Keywords---Cryptography; graphical password; matching; partition; labelling.}}


%
\IEEEpeerreviewmaketitle


\pagestyle{fancy}
\pagestyle{plain}

\section{Introduction}

As known, public key and private key play important roles in cryptography nowadays. How to realize the authentication of public keys and private keys by ciphers with easy to use and high level security? Graphical passwords (GPWs) emerged for alternative to  text-based passwords. GPWs have been researched and applied in the real life for a long time, for example, QR code is popular in electronic commerce, open screen graphic cipher for smart mobiles, and so on (\cite{Suo-Zhu-Owen-2005, Biddle-Chiasson-van-Oorschot-2009, Gao-Jia-Ye-Ma-2013}). The existing GPWs are lack of pictures frequently changed and occupy huge spaces of computer, and need users to learn more and have good memory, and do not support more individual ideas and personal making GPWs.

For overcoming weak limits of the existing GPWs, Wang \emph{et al.} (\cite{Wang-Xu-Yao-2016, Wang-Xu-Yao-Key-models-Lock-models-2016}) have designed Topological graphic passwords (Topsnut-gpws) by an idea of ``topological structure pulsing number theory''. Clearly, Topsnut-gpw is a mathematical expression of  nature, and can be storage into computer by canonical matrices, and be operated quickly in computer. Topsnut-gpws have such advantages: (i) no general polynomial algorithms for finding topological structures and colorings/labellings in graph theory, which are two basic components for producing Topsnut-gpws; (ii) easily yield text-based passwords, and such procedure is irreversible; (iii) easily operating like gesture passwords used in mobiles with touch screen; (iv) allow personal knowledge into making Topsnut-gpws for long time remembering in mind; (v) huge spaces \cite{Yao-Zhang-Sun-Mu-Wang-Zhang-Yang-2018}, for instant, there are $t_{23}\approx 2^{179}$ and $t_{24}\approx 2^{197}$, where $t_p$ is the number of graphs of $p$ vertices, and over 200 existing labellings \cite{Gallian2016}, and so on. Thereby, Topsnut-gpws have provable security, computational security and unbreakable in nowadays' computer. We will study Topsnut-gpws by the matching of view in this article.

\subsection{Examples}

Matching phenomena are popular and exist almost every where of the world, such as black and white, more and less, men and women, rich and poor, public and private, and even mathematics, also, is the matching of space structure and quantity. Matching is not a connection between two different things, but also connections of one-more things and more-more things. Matching is not a simple combination of two or more things, but a combination with restrictive conditions. Here, our matchings belong to  mathematics and cryptograph.

In cryptography we can consider that a public key and a private key form an \emph{authentication matching}. Sometimes, people want one public key vs two or more private keys. In \cite{Wang-Xu-Yao-2016} and \cite{Wang-Xu-Yao-Key-models-Lock-models-2016}, the authors have listed many advantages of Topsnut-gpws. One advantage of Topsnut-gpws is that Topsnut-gpws can produce easily text-based passwords with longer bytes as longer as desired. However, we cannot reconstruct the origin Topsnut-gpws from the text-based passwords made by them. This irreversibility also appears in Hash algorithm that is a one-way encryption system, that is, only encryption procedure, no decryption procedure.

We start our discussion with the following examples for showing Topsnut-gpws worked best in generating text-based passwords. In Fig.\ref{fig:K4-2-decomposition}, we can see a Topsnut-gpw $K_4$ having: Any two small circles (called vertices hereafter) are connected by an edge. Furthermore, we identify the vertices of the Topsnut-gpws $T_1,T_2,T_3$ pictured in Fig.\ref{fig:K4-2-decomposition} that have the same label into one vertex, the resulting graph is just $K_4$, so we write this fact as $K_4=\odot\langle T_i\rangle^3_1$ for briefness. A phenomenon is that any two caterpillars $T_i$ and $T_j$ have no the edge with the same label, we denote this phenomenon by $E(T_i)\cap E(T_j)=\emptyset $. However, the Topsnut-gpws $H_1,H_2,H_3$ shown in Fig.\ref{fig:K4-2-decomposition} have common edges with the same labels, we denote this fact as $E(H_i)\cap E(H_j)\neq \emptyset $; and identifying the vertices of these Topsnut-gpws $H_1,H_2,H_3$ results in $K_4$ by deleting multiple edges, we write this case as $K_4=\ominus(H_i)^3_1$. We denote the number of vertices of $T_i$ as $|V(T_i)|$ for $i=1,2,3$, where $V(G)$ is the set of vertices of a graph $G$. It is not hard to observe that $|V(T_i)|\neq |V(T_j)|$ as $i\neq j$. But, $|V(H_i)|=|V(H_j)|$ for $i\neq j$ in Fig.\ref{fig:K4-2-decomposition}, these $H_1,H_2,H_3$ are the spanning trees of $K_4$. Spanning trees, such as spanning algorithms, Spanning Tree Protocol (STP), minimum spanning trees and spanning tree searches, are useful in today's scientific areas, especially, network security.

\begin{figure}[h]
\centering
\includegraphics[height=4cm]{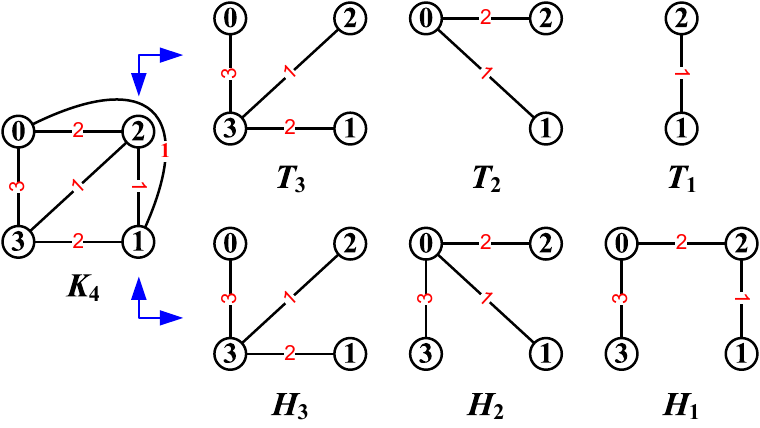}\\
\caption{\label{fig:K4-2-decomposition} {\small Two decompositions of $K_4$. The vertices  and edges of each $T_i$ and $H_i$ are labelled with $0,1,2,3$.}}
\end{figure}

We use three Topsnut-gpws $T_1,T_2,T_3$ shown in Fig.\ref{fig:K4-2-decomposition} to generate three text-based passwords $D(T_1)=112$, $D(T_2)=110022$, $D(T_3)=033312321$ for the purpose of encrypting electronic files. Moreover, we produce a little bit of complex text-based passwords below:
\begin{center}
$D_{123}=D(T_1)+D(T_2)+D(T_3)=112110022033312321$

$D_{321}=D(T_3)+D(T_2)+D(T_1)=033312321110022112$

$D_{132}=D(T_1)+D(T_3)+D(T_2)=112033312321110022$.
\end{center}
Clearly, $D_{123}$, $D_{321}$ and $D_{132}$ differ from to each other (notice that we have other three text-based passwords $D_{213}$, $D_{231}$ and $D_{312}$). In fact, we can construct more complex text-based passwords by
\begin{equation}\label{eqa:long-byte}
D_T=\sum ^m_{i=1}D(T_{j_i})
\end{equation}
with $1\leq j_i\leq 3$ and $m\geq 2$, such that $D_T$ has longer bytes as we desired. Also, we can get text-based passwords $D(H_i)$ from three Topsnut-gpws $H_1,H_2,H_3$ shown in Fig.\ref{fig:K4-2-decomposition}, and moreover
\begin{equation}\label{eqa:long-byte-1}
D_H=\sum ^m_{i=1}D(H_{j_i})
\end{equation}
with $1\leq j_i\leq 3$ an $m\geq 2$.

The second example for showing an important property of Topsnut-gpws. In Fig.\ref{fig:twi-odd-elegant}, we can walk along a path $1\rightarrow 10\rightarrow 21\rightarrow 6\rightarrow 13\rightarrow 2\rightarrow 9\rightarrow 14$ to make two text-based passwords
\begin{equation}\label{eqa:long-byte-elegant00}
{
\begin{split}
T'_{vv}=&1'1012141'10'1152110'21'681021'6'13\\
&17216'13'2613'2'913192'9'14'135914'
\end{split}}
\end{equation}
and
\begin{equation}\label{eqa:long-byte-elegant11}
{
\begin{split}
T'_{vev}=&1'1110131215141'10'11131592110'21'567\\
&891021'6'19131175216'13'15219613'2'119\\
&151321192'9'2414189'14'1511731951914'
\end{split}}
\end{equation}
or eliminating ``$'$'' from $T'_{vv}$ and $T'_{vev}$ for complex reason yields
\begin{equation}\label{eqa:long-byte-elegant22}
{
\begin{split}
T_{vv}=&1101214110115211021681021613\\
&172161326132913192914135914
\end{split}}
\end{equation}
and
\begin{equation}\label{eqa:long-byte-elegant33}
{
\begin{split}
T_{vev}=&11110131215141101113159211021567\\
&89102161913117521613152196132119\\
&15132119292414189141511731951914
\end{split}}
\end{equation}

\begin{figure}[h]
\centering
\includegraphics[height=5.6cm]{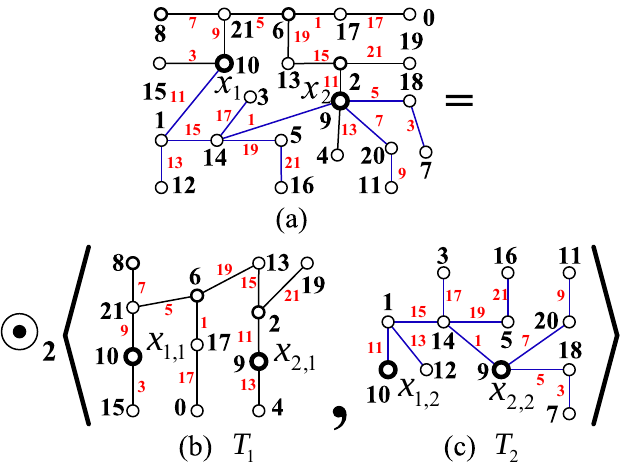}\\
\caption{\label{fig:twi-odd-elegant} {\small (a) is a twin odd-elegant graph $\odot_2 \langle T_1, T_2\rangle$ cited from \cite{Wang-Xu-Yao-2017}.}}
\end{figure}

It is not easy to reconstruct $\odot_2 \langle T_1, T_2\rangle$ shown in Fig.\ref{fig:twi-odd-elegant} by $T_{vv}$ in (\ref{eqa:long-byte-elegant22}) or $T_{vev}$ in (\ref{eqa:long-byte-elegant33}), thus large scale Topsnut-gpws are \emph{provable security}, since reconstructing graph problems are related with some mathematical conjectures, such as Kelly-Ulam's Reconstruction Conjecture proposed in 1942. So we can claim that the procedure of getting text-based passwords from Topsnut-gpws is \emph{irreversible}. On the other hand, large scale Topsnut-gpws made by various graph labellings are \emph{computational security, or computationally unbreakable}, since \emph{no polynomial algorithm} for finding all possible graph labellings for a given graph, also \emph{no polynomial algorithm} for constructing all non-isomorphic graphs. We have \emph{no polynomial algorithm} for listing all possible text-based passwords in a Topsnut-gpw, although it may be interesting and important.

We have to face the following problems:

\begin{asparaenum}[Prob-1.  ]
\item In general, for $n\geq 4$, how many ways are there to form a Topsnut-gpw $K_n=\odot\langle T_i\rangle^m_1$ with $E(T_i)\cap E(T_j)=\emptyset $ as $i\neq j$, or $K_n=\ominus(H_i)^m_1$ with $V(K_n)=V(H_i)$ for $i\in [1,m]$ and $|E(H_i)|=|E(H_j)|$?
\item Are the text-based passwords $D_T$ in (\ref{eqa:long-byte}) and $D_H$ in (\ref{eqa:long-byte-1}), $T_{vv}$ in (\ref{eqa:long-byte-elegant22}) and $T_{vev}$ in (\ref{eqa:long-byte-elegant33}) computationally unbreakable?
\end{asparaenum}

\vskip 0.4cm

Fig.\ref{fig:H3-T3-decomposition} tells us: Each Topsnut-gpw $T_i$ ($H_i$) consists of one configuration $T_{i,1}$ ($H_{i,1}$) and one labelling $T_{i,2}$ ($H_{i,2}$). We need to know:
\begin{asparaenum}[$-$ ]
\item How many configurations $T_{i}$ or $H_{i}$ are there for producing $K_n$ with $n\geq 4$?
\item How many type of label-functions (also, called labellings hereafter) do $T_{i}$ and $H_{i}$ admit?
\item How to label the vertices or edges of $T_{i}$ and $H_{i}$ with the labellings admitted by $T_{i}$ and $H_{i}$, such that identifying the vertices of $T_{i}$ or $H_{i}$ that have the same labels into one just results $K_n=\odot\langle T_i\rangle^m_1$ or $K_n=\ominus(H_i)^m_1$?
\end{asparaenum}

\begin{figure}[h]
\centering
\includegraphics[height=4.4cm]{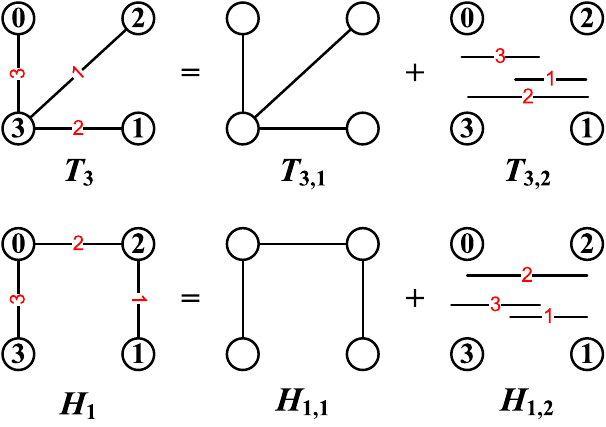}\\
\caption{\label{fig:H3-T3-decomposition} {\small Each Topsnut-gpw $T_i$ ($H_i$) consists of a topological structure $T_{i,1}$ ($H_{i,1}$) and label-function $T_{i,2}$ ($H_{i,2}$).}}
\end{figure}

A public key and a private key make an authentication true in network communication. Sometimes, an authentication needs one public key and two or more private keys, and vice versa. In other words, we can consider that ``public key vs private key'' forms some \emph{matching partition} (authentication can be seen as a \emph{matching entirety} that can be partitioned into several parts). Here, we will design matching type of Topsnut-gpws (\emph{Topsnut-matchings}) for the requirement of protecting people's information and property in networks.

The topic of matching partition contains: configuration matching partition, coloring/labelling matching partition, set matching partition, matching chain, etc. In the number of matching partitions, we have one-vs-one, one-vs-more and more-vs-more styles of matching partitions. Each matching mentioned here will be obtained by one of configuration-vs-configuration, configuration-vs-labelling, labelling-vs-labelling, and (configuration, labelling)-vs-(configuration, labelling).

A Topsnut-gpw was made by a \emph{topological structure} (also, configuration, called \emph{graph} in graph theory, which is a branch of mathematics) with a \emph{label-function} (also, called graph labelling, or labelling for short) on vertices, or edges, or vertices and edges (see Fig.\ref{fig:H3-T3-decomposition}). So, we are reasonable to consider any labeled graph of graph theory as a Topsnut-gpw here.  Notice that Topsnut-gpws can be defined by many labellings shown in \cite{Gallian2016}.

\subsection{Preliminary}

Before exploring solutions of Prob-1 and Prob-2, we need terminology, notations and particular graphs (=configurations) in our discussion, standard notations and definitions of graph theory can be found in \cite{Bondy-2008}. A $(p,q)$-graph $G$ has $p$ vertices and $q$ edges. We will use a notation $[a, b]=\{a, a+1,\dots, b\}$, where $m, n$ are integers with $ 0\leq m<n$, and employ another notation $[\alpha, \beta]^o=\{\alpha, \alpha+2, \dots, \beta\}$ with odd integers $\alpha, \beta$ holding $ 1\leq \alpha<\beta$ true.

A \emph{tree} is a graph in which any pair of two vertices $x,y$ is connected by a unique path $P(x,y)=xu_1u_2\cdots u_my$; a \emph{leaf} is a vertex of degree one; a \emph{caterpillar} is a tree such that the deletion of all leaves of the tree results in a path; a \emph{lobster} is a tree such that the deletion of all leaves of the tree produces just a caterpillar.

A \emph{labelling} $h$ of a graph $G$ is a mapping $h:S\subseteq V(G)\cup E(G)\rightarrow [a,b]$ such that $h(x)\neq h(y)$ for any pair of elements $x,y$ of $S$, and write the label set $h(S)=\{h(x): x\in S\}$. A \emph{dual labelling} $h'$ of a labelling $h$ is defined as: $h'(z)=\max h(S)+\min h(S)-h(z)$ for $z\in S$. Moreover, $h(S)$ is called  \emph{vertex label set} if $S=V(G)$, $h(S)$ \emph{edge label set} if $S=E(G)$, and $h(S)$ a \emph{universal  label set} if $S=V(G)\cup E(G)$.

We, in the following discussion, need four pairs of graph operations on four basic elements of vertex, edge, path and cycle as follows: In Fig.\ref{fig:two-operations}, a \emph{vertex-split operation} from (a) to (b); a \emph{vertex-identifying operation} from (b) to (a); an \emph{edge-split operation} from (c) to (d); and an \emph{edge-identifying operation} from (d) to (c). Let $N(x)$ be the set of all neighbors of a vertex $x$, very often, $N(x)$  is called a \emph{neighbor set}. In Fig.\ref{fig:two-operations} (b) and (d), after a series of vertex/edge-split operations, we have to emphasize that the neighbor sets hold $N(y')\cap N(y'')=\emptyset $, $N(u')\cap N(u'')=\emptyset $, $N(u')\cap N(v'')=\emptyset $, $N(v')\cap N(u'')=\emptyset $ and $N(v')\cap N(v'')=\emptyset $ true. The \emph{path/cycle-split operation} and the \emph{path/cycle-identifying operation} are shown in Fig.\ref{fig:path-split}, also, it stresses that the neighbor sets $N(u'_i)\cap N(u''_j)=\emptyset $ with $i,j\in [1,n]$.

\begin{figure}[h]
\centering
\includegraphics[height=6.4cm]{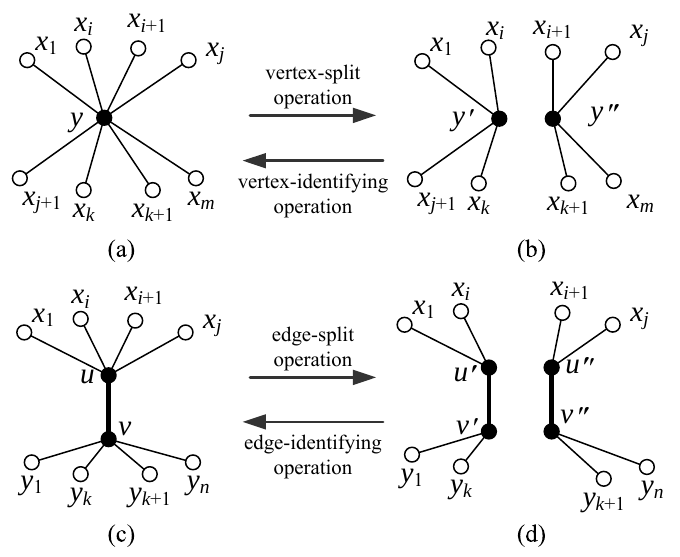}\\
\caption{\label{fig:two-operations} {\small A scheme for illustrating four graph operations: vertex-split operation; vertex-identifying operation; edge-split operation; edge-identifying operation.}}
\end{figure}

\begin{figure}[h]
\centering
\includegraphics[height=7.4cm]{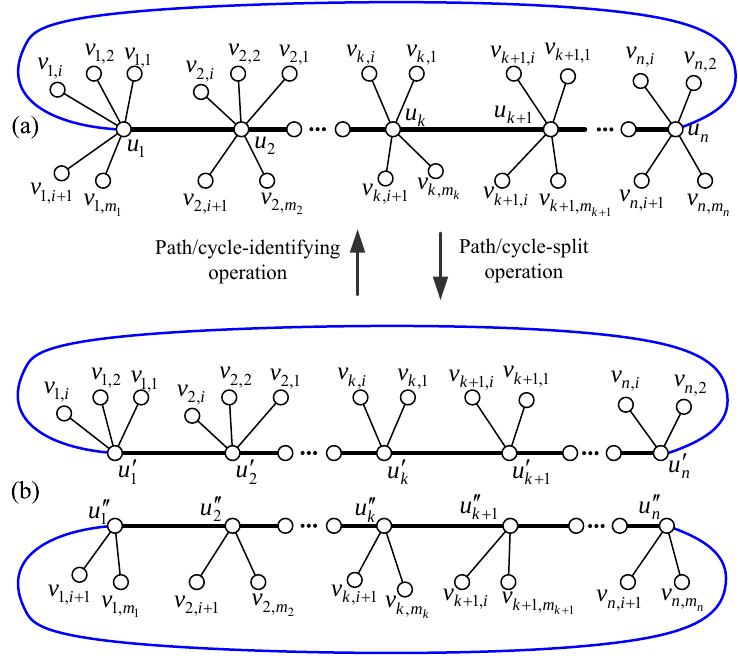}\\
\caption{\label{fig:path-split} {\small A path-split operation from (a) to (b), and a path-identifying operation from (b) to (a). Also, a cycle-split operation from (a) to (b), and a cycle-identifying operation from (b) to (a).}}
\end{figure}

\section{Matching diversity}

\subsection{Configuration matching partition}

Let $W$ be a \emph{universal graph} admitting a labelling $f$. If $W$ contains \emph{edge disjoint subgraph} $G_i$ admitting a labelling $f_i$ induced by $f$ with $i\in [1,m]$ such that $V(W)=\bigcup^m_{i=1} V(G_i)$ and $E(W)=\bigcup^m_{i=1} E(G_i)$ after identifying the vertices of $G_1, G_2,\dots ,G_m$ having the same labels respectively (see an example $K_4=\odot\langle T_i\rangle ^3_1$ shown in Fig.\ref{fig:K4-2-decomposition}), we write this case as $W=\odot\langle G_i\rangle ^m_1$, called an \emph{edge-disjoint matching partition}. On the other hand, each $G_i$ matches with the set $Z_i=\{G_1, G_2,\dots ,G_m\}\setminus \{G_i\}$ on the edge-disjoint matching partition $W$, then we denote simply $W$ as $\odot\langle G_i, Z_i\rangle$. Also, we allow the case $V(G_i)=V(W)$ with $i\in [1,m]$ in an edge-disjoint matching partition $W=\odot\langle G_i\rangle ^m_1$.

In encryption of networks, each $G_i$ with $i\in [1,m]$ can be considered as a \emph{key}, so $W=\odot\langle G_i\rangle ^m_1$ is just an \emph{edge-disjoint authentication}. Furthermore, if $G_i$ is a \emph{public key}, and the set $Z_i$ is a group of \emph{private keys}, and $W=\odot\langle G_i, Z_i\rangle$ is just a \emph{one-vs-more authentication}.

Similarly, we have a \emph{multiple-edge matching partition} $W=\ominus(H_i) ^m_1$ as each subgraph $H_i$ admitting a labelling $f_i$ induced by $f$ with $i\in [1,m]$ holds $V(H_i)=V(W)$ and $E(H_i)\cap E(G_j)\neq \emptyset$ true for some $i\neq j$ by identifying the vertices of $\bigcup ^m_{i=1}V(H_i)$ having the same labels together and eliminating multiple edges. A \emph{mixed matching partition} $W=\bigcup^m_{i=1}G_i$ has some $G_j$ holding $V(G_j)\neq V(W)$ and $E(G_s)\cap E(G_t)\neq \emptyset$ true for some $s\neq t$.

\begin{expr}\label{expr:labelling-vs-dual-labelling}
Naturally, a labelling $h$ and its dual labelling $h'$ of a graph $G$ are matching with each other, call $(h,h')$ a \emph{labelling matching} of $G$. Let $G_1,G_2$ be two copies of $G$, and let $G_1$ admit the labelling $h$, $G_2$ admit the dual labelling $h'$ of $h$, so we have a graph $\odot \langle G_1,G_2\rangle$ obtained by identifying those vertices of $V(G_1)\cup V(G_2)$ having the same labels together.
\end{expr}

\begin{expr}\label{expr:complete-matching}
As $W=K_n$, $V(K_n)=V(G_i)$ with $i=1,2$ and $E(G_1)\cap E(G_2)=\emptyset$ holding $E(K_n)=E(G_1)\cup E(G_2)$ true, we say $G_1$ and $G_2$ matching to each other and $G_1$ and $G_2$ are \emph{complementary} to each other, moreover $K_n=\odot\langle G_1,G_2\rangle$. Conversely, by doing a vertex-split operation to each vertex of $K_n$, so we split $K_n$ into two subgraphs $G_1$ and $G_2$.
\end{expr}

\begin{expr}\label{expr:Euler-matching}
As $W$ is an Euler graph, the edge-disjoint matching partition $W=\odot\langle G_i\rangle ^m_1$ has its vertex degree $\textrm{deg}_W(u)=\sum ^m_{i=1}\textrm{deg}_{G_i}(u)$ to be even for each vertex $u\in V(W)$. Here, $W$ admits some v-set e-proper labelling $(F,g)$ defined in Definition \ref{defn:set-labelling}. For $m=2$,  $G_1,G_2$ are not Euler graphs, but $\odot\langle G_1,G_2\rangle$ is an Euler graph, then we say both $G_1,G_2$ are \emph{Euler matching} to each other.
\end{expr}

\begin{expr}\label{expr:complete-conjecture}
As $W$ is a complete graph $K_n$, we have the following longstanding conjectures in graph theory, which show that the edge-disjoint matching partition $K_n=\odot\langle G_i\rangle ^m_1$ may be \emph{computationally unbreakable}:

(i) Anton Kotzig (1964) proposed the \emph{Perfect 1-Factorization Conjecture}: For any $n \geq 2$,
$K_{2n}$ can be decomposed into $2n-1$ perfect matchings such that the union of any two matchings
forms a hamiltonian cycle of $K_{2n}$.

(ii) If each tree admits a graceful labelling, then this
will settle down a well-known \emph{Ringel-Kotzig Decomposition Conjecture} (Gerhard Ringel and Anton Kotzig, 1963; Alexander Rosa, 1967): A complete graph $K_{2n+1}$ can be decomposed into $2n + 1$ edge-disjoint subgraphs that are all isomorphic to any given tree having $n$ edges.

(iii) \emph{K-T conjecture} (Gy\'{a}r\'{a}s and Lehel, 1978; B\'{e}la Bollob\'{a}s, 1995): For integer
$n \geq 3$, given $n$ disjoint trees $T_k$ having $k$ vertices with respect to $k\in [1,n]$. Then the complete graph
$K_n$ can be decomposed into the union of $n$ edge-disjoint trees $H_k$, such that $T_k\cong H_k$ with $k\in [1,n]$. Very often, we write this case as $\langle T_1, T_2, \dots, T_m | K_n\rangle$.

Thereby, the above three conjectures can help us to design more complex Topsnut-gpws to be computationally unbreakable.
\end{expr}

\subsection{Coloring/labelling matchings}

Graph coloring/labellings are powerful and essential for designing Topsnut-gpws, let us see an example as follows:

\begin{expr}\label{expr:planar-matching}
As $W$ is a \emph{maximal planar graph}, each $G_i$ with $i\in [1,m]$ is a \emph{semi-maximal planar graph} (\cite{Jin-Xu-(1)-2016, Jin-Xu-(2)-2016, Jin-Xu-(3)-2016, Jin-Xu-(4)-2016}) having a common even-cycle $C$, such that $E(G_i)\cap E(G_j)=E(C)$ for $i\neq j$. If each $G_i$ is $4$-colorable such that the even-cycle $C$ is colored with $1,2$ only, then we do a series of cycle-identifying operations on $m$ common even-cycles in $G_1,G_2,\dots ,G_m$ in to one. The resulting Topsnut-gpw, like a ``book'', is denoted as $W=\ominus_C(G_i) ^m_1$, and $G_i\cup G_j$ for $i\neq j$ is a maximal planar graph, we say $W=\ominus_C(G_i) ^m_1$ to be a \emph{maximal planar $C$-matching partition}. Each Topsnut-gpw $G_i$ is as a ``page'' of the ``book'' $W$, and the common even-cycle $C$ is as the ``bone'' of the ``book'' $W$. This ``book'' can be considered as an authentication too.  It is not hard to see that there are two or more (infinite) ``books'' $W=\ominus_C(G_i) ^m_1$.

Conversely, select a cycle $L$ of a maximal planar graph $G$, and do an edge-split operation to each edge of the cycle $L$, so we split $G$ into two semi-maximal planar graphs $G^L$ and $\overline{G}^L$, and call $G^L$ and $\overline{G}^L$ to be matching to each other (see Fig.\ref{fig:sami-planar}). So, $G=\ominus_L\langle G^L, \overline{G}^L\rangle$. Determining particular \emph{semi-maximal planar matchings} $(G^L,\overline{G}^L)$ can provide more complex models for authentication of public keys and private keys, such as both graphs $G^L-V(L)$ and $\overline{G}^L-V(L)$ are trees.
\end{expr}

\begin{figure}[h]
\centering
\includegraphics[height=3cm]{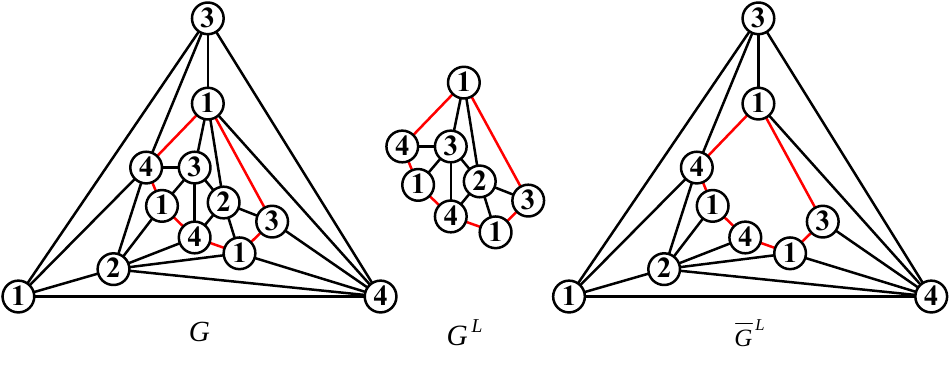}\\
\caption{\label{fig:sami-planar} {\small A maximal planar graph $G=\ominus_L(G^L,\overline{G}^L)$ with its two semi-maximal planar graphs $G^L$ and $\overline{G}^L$. As known, every planar graph is $4$-colorable, no mathematical proof for this fact up to now.}}
\end{figure}

\begin{defn}\label{defn:set-ordered-odd-graceful-labelling}
(\cite{Gallian2016, Bing-Yao-Cheng-Yao-Zhao2009, Zhou-Yao-Chen-Tao2012}) Suppose that a bipartite $(p,q)$-graph $G$ with partition $(X,Y)$  admits a vertex labelling $f:V(G) \rightarrow [0,2q-1]$, such that every edge $uv$ is labeled as $f(uv)=|f(u)-f(v)|$ holding $f(E(G))=[1, 2q-1]^o$ true, we call $f$ an \emph{odd-graceful labelling} of $G$ (called an \emph{odd-graceful graph}). Furthermore, if $G$ holds $\max\{f(x):~x\in X\}< \min\{f(y):~y\in Y\}$ ($f_{\max}(X)<f_{\min}(Y)$ for short) true, then $f$ is called a \emph{set-ordered odd-graceful labelling}.\qqed
\end{defn}

In \cite{Yao-Mu-Sun-Zhang-Wang-Su-Ma-2018}, we have expanded the odd-graceful labelling as: Let $G$ be a $(p,q)$-graph, we have:

(i) If $G$ admits a vertex labelling $f:V(G) \rightarrow [0,2q]$ (it is allowed $f(E(G))\subset [1,2q]$), such that every edge $uv$ is labeled as $f(uv)=|f(u)-f(v)|$ or $f(uv)=2q-1-|f(u)-f(v)|$ and $f(E(G))=[1, 2q-1]^o$, then we call $f$ a \emph{pan-odd-graceful labelling}.

(ii) A \emph{$k$-matching odd-graceful labelling} $g$ of an odd-graceful labelling $f$ of a $(p,q)$-Topsnut-gpw $G$ is a vertex labelling defined on another graph $H$ admitting $g:V(H) \rightarrow [0,2q-1]$, every edge $uv\in E(H)$ has its label $g(uv)=|g(u)-g(v)|$ holding $g(E(H))=[1, 2q-1]^o$ true, such that $f(V(G))\cup g(V(H))=[0,2q]$, or $[0,2q-1]$ and $|f(V(G))\cap g(V(H))|=k$. We call $H$ with a $k$-matching odd-graceful labelling as an \emph{odd-graceful Topsnut-matching} of $G$, denoted as $\odot\langle G,H \rangle$. (see  Fig. \ref{fig:odd-graceful-matching} and Fig.\ref{fig:odd-graceful-matching-connected})

(iii) A \emph{$k$-sequential odd-graceful labelling} $h: V(G)\rightarrow [k,2q-1+k]$ such that the induced edge labelling $h(uv)=|h(u)-h(v)|$ for $uv \in E(G)$ holding $g(E(G))=[1,2q-1]^o$ true.

\begin{defn}\label{defn:odd-elegant-labelling}
\cite{Zhou-Yao-Chen-2013} A $(p,q)$-graph $G$ admits a labelling $f:V(G)
\rightarrow [0,2q-1]$ such that each edge $uv\in E(G)$ is labeled as $ f(uv)=f(u)+f(v)~(\bmod~2q)$. If the set $f(E(G))$ of all edge labels is equal to $[1, 2q-1]^{o}$, we call $f$ an \emph{odd-elegant labelling} of $G$ (called an \emph{odd-elegant graph}).\qqed
\end{defn}

Finding all odd-graceful (odd-elegant) labellings of a Topsnut-gpw $G$ admitting an odd-graceful (odd-elegant) labelling seems to be very difficult, and no way is for determining conditions for graphs admitting set-ordered odd-graceful (odd-elegant) labellings up to now. In Fig.\ref{fig:odd-graceful-matching}, we have six odd-graceful Topsnut-matchings $G_i=\odot_0\langle G,H_i \rangle$ with $i\in [1,6]$, since $G$ admits an odd-graceful labelling. Here, $V(G_i)=V(G)\cup V(H_i)$ and $E(G_i)=E(G)\cup E(H_i)$ with $i\in [1,6]$. Here, we present an algorithm for finding  odd-graceful Topsnut-matchings.

\vskip 0.4cm

\textbf{ODD-GRACEFUL-GRAPH Algorithm:}

\textbf{Input:} A connected $(p,q)$-graph $G$ admitting an odd-graceful labelling $f$.

\textbf{Output:} A connected odd-graceful Topsnut-matching $H$ admitting an odd-graceful labelling.

\emph{Step 1.} Build up an integer set $R(f,G)=[0,2q-1]\setminus f(V(G))$ (or $R'(f,G)=[0,2q]\setminus f(V(G))$), and make a candidate edge set $C_0=\{x^ox^e:x^o,x^e\in R(f,G), x^o\textrm{ is odd, }x^e\textrm{ is even}\}$, and a graph $H_0$ is constructed by identifying the ends of edges of $C_0$ into one vertex, these ends  have the same labels.

\emph{Step 2.} If the graph $H_k$ contains no two edges $x^o_ix^e_i$ and $x^o_jx^e_j$ holding $|x^o_i-x^e_i|=|x^o_j-x^e_j|$ true, go to Step 4; otherwise,  go to Step 3.

\emph{Step 3.} Do $H_{k+1}:=x^o_ix^e_i+H_k-E(k+1)$ with $E(k+1)=\{x^o_jx^e_j: |x^o_j-x^e_j|=k+1\}\subset C_0$ and $x^o_ix^e_i$ with $|x^o_i-x^e_i|=k+1$, such that $H_{k+1}$ is connected,  go to Step 2.

\emph{Step 4.} Return a connected Topsnut-matching $H$ admitting an odd-graceful labelling.

\vskip 0.4cm

\begin{figure}[h]
\centering
\includegraphics[height=7.2cm]{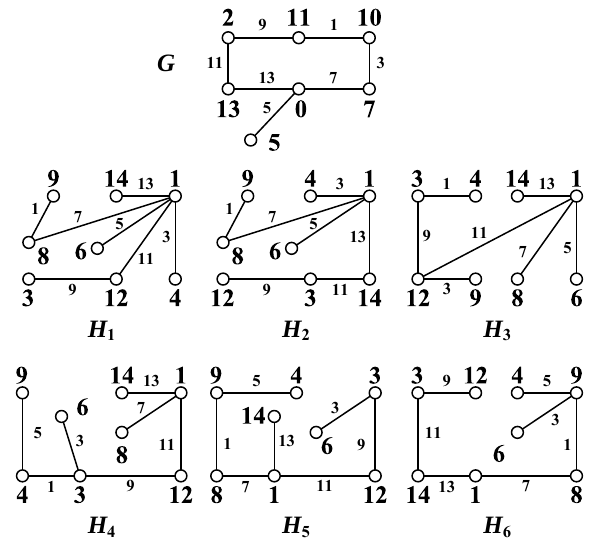}\\
\caption{\label{fig:odd-graceful-matching} {\small A $(7,7)$-graph $G$ admitting an odd-graceful labelling has six odd-graceful matchings $H_i$ with a matching odd-graceful labelling $g_i$ for $i\in [1,6]$, such that each odd-graceful Topsnut-matching $\odot\langle G,H_i \rangle$ holds $f(V(G))\cup g_i(V(H_i))= [0,14]$ and $f(V(G))\cap g_i(V(H_i))=\emptyset$ true with $i\in [1,6]$.}}
\end{figure}

\begin{figure}[h]
\centering
\includegraphics[height=2.4cm]{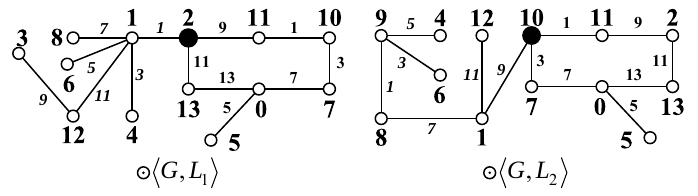}\\
\caption{\label{fig:odd-graceful-matching-connected} {\small Two odd-graceful Topsnut-matchings $\odot_1\langle G,L_1 \rangle$ and $\odot_1\langle G,L_2 \rangle$ of a $(7,7)$-graph $G$ shown in Fig.\ref{fig:odd-graceful-matching} hold $f(V(G))\cap g_1(V(L_1))=\{2\}$, $f(V(G))\cap g_2(V(L_2))=\{10\}$ and $f(V(G))\cup g(V(L_i))\neq [0,14]$ true with $i=1,2$.}}
\end{figure}

Let $H$ be an (a pan-)odd-graceful graph with vertices $v_1,v_2,\dots, v_p$. If each vertex $v_i$ matches with an (a pan-)odd-graceful graph $T_i$ with $i\in [1,p]$ such that there exists an odd-graceful Topsnut-matching $\odot_1\langle H,T_i\rangle ^p_1$ obtained by identifying the vertex $v_i$ with some vertex of $T_i$ into one, these two vertices have been labeled with the same labels. We say $\odot_1\langle H,T_i\rangle ^p_1$ a  \emph{(pan-)odd-graceful Topsnut-matching team} (see an example shown in  Fig.\ref{fig:odd-every-vertex-matching}). Moreover, $\odot_1\langle H,T_i\rangle ^p_1$ is called a \emph{uniformly (pan-)odd-gracefully Topsnut-matching team} if $T_i\cong T_j$ for $i\neq j$.

\begin{figure}[h]
\centering
\includegraphics[height=5.6cm]{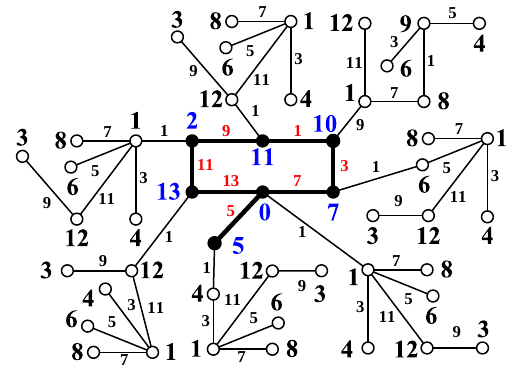}\\
\caption{\label{fig:odd-every-vertex-matching} {\small An odd-graceful Topsnut-matching team.}}
\end{figure}

\begin{thm} \label{thm:Topsnut-matching-team}
Each caterpillar $H$ of $p$ vertices has an (a pan-)odd-graceful Topsnut-matching team $\odot_1\langle H,T_i\rangle ^p_1$.
\end{thm}
\begin{proof} There is  a path $P=u_1u_2\cdots u_n$ in a caterpillar $H$, such that each $u_i$ has its own leaf set $L(u_i)=\{v_{i,j}:j\in [1,m_i]\}$ with $m_i\geq 0$ and $i\in [1,n]$. And
$$V(H)=V(P)\cup L(u_1)\cup L(u_2)\cup \cdots \cup L(u_n),$$ see a caterpillar shown in Fig.\ref{fig:caterpillar-general}.

We define an odd-graceful labelling $h$ of $H$ by setting

$h(u_1)=0$, $h(v_{2,j})=h(u_1)+2j$ with $j\in [1,m_2]$;

$h(u_3)=h(v_{2,m_2})+2=2(m_2+1)$,

$h(v_{4,j})=h(u_3)+2j=2(m_2+1+j)$ with $j\in [1,m_4]$;

$h(u_5)=h(v_{4,m_4})+2=2(m_2+m_4+1)+2=2(m_2+m_4+2)$,

$h(v_{6,j})=h(u_5)+2j=2(m_2+m_4+2+j)$ with $j\in [1,m_6]$.

For $i\geq 2$, we let
$$h(u_{2i-1})=h(v_{2i-2,m_{2i-2}})+2=2\left (i-1+\sum^{i-1}_{k=1}m_{2k}\right )$$
and for $j\in [1,m_{2i}]$,
$$h(v_{2i,j})=h(u_{2i-1})+2j=2\left (i-1+j+\sum^{i-1}_{k=1}m_{2k}\right ).$$

Clearly, $h(v_{2i-2,m_{2i-2}})<h(u_{2i-1})$ and $h(v_{2i,j})<h(v_{2i,j+1})$ with $j\in [1,m_{2i}-1]$ and $i\geq 2$.

\emph{Case 1.}  $n$ is even.

We set $h(u_{n})=h(v_{n,m_{n}})+1$, an odd integer, and $h(v_{n-1,j})=h(u_{n})+2j$  with $j\in [1,m_{n-1}]$. Notice that $h(u_{n})-h(v_{n,m_{n}})=1$. Furthermore, we have $h(u_{n-2})=h(v_{n-1,m_{n-1}})+2$, and $h(v_{n-3,j})=h(u_{n-2})+2j$  with $j\in [1,m_{n-3}]$.

For $i\geq 1$, $h(u_{n-2i})=h(v_{n-2i+1,m_{n-2i+1}})+2$, and $h(v_{n-2i-1,j})=h(u_{n-2i})+2j$  with $j\in [1,m_{n-2i-1}]$. As  $2i=n-2$, $h(u_{2})=h(v_{3,m_{3}})+2$, and $h(v_{1,j})=h(u_{2})+2j$  with $j\in [1,m_{1}]$.

\emph{Case 2.}  $n$ is odd.

We set $h(u_{n})=h(v_{n-1,m_{n-1}})+2$ to be an even integer, so we label $h(v_{n,j})=h(u_{n})+j$  with $j\in [1,m_{n}]$, such that each $h(v_{n,j})$ is  an odd integer. Then,  $h(v_{n,1})-h(u_{n})=1$. Next, let $h(u_{n-1})=h(v_{n,m_{n}})+2$, and $h(v_{n-2,j})=h(u_{n-1})+2j$  with $j\in [1,m_{n-2}]$. In general, we have $h(u_{n-2i+1})=h(v_{n-2i+2,m_{n-2i+2}})+2$, and $h(v_{n-2i,j})=h(u_{n-2i+1})+2j$  with $j\in [1,m_{n-2i}]$ for $i\geq 1$. As $2i=n-1$, $h(u_{2})=h(v_{1,m_{n-2i-2}})+2$, and $h(v_{1,j})=h(u_{2})+2j$  with $j\in [1,m_{1}]$.

Notice that $h(v_{1,m_{1}})=2p-3$ and $h(v_{1,m_{1}-1})=2p-5$. So, we can use the induction to show $h$ is a set-ordered   odd-graceful labelling of the  caterpillar $T$.

Now, we write a copy of the caterpillar $T$ with an odd-graceful labelling $h$ by $H$ and define a set-ordered pan-odd-graceful labelling $h^*$ of  $H$ as:  $h^*(x)=h(x)+1$ for each $x\in V(T)=V(H)$. Clearly,
$$h^*(V(H))\cup h(V(T))=[0,2(P-1)]
$$
and $h^*(V(H))\cup h(V(T))=\emptyset $. Now, we make another caterpillar $T-v_{n,1}$ with a labelling $f$ defined by $f(y)=h(y)$ for $y\in V(T)\setminus \{v_{n,1}\}$. So, $f(E(T-v_{n,1}))=[3,2p-3]^o$. We add a new vertex $w_i$ to $T-v_{n,1}$, and join $w_i$ with some vertex $w'_i$ of $T-v_{n,1}$ by an edge $w_iw'_i$, the resulting tree is denoted as $T_i=T-v_{n,1}+w_iw'_i$, and define a labelling $f_i$ of $T_i$ as $f_i(x)=f(x)$ for $x\in V(T_i)\setminus \{w_i\}$, and $f_i(w_i)\in h^*(V(H))$ such that $|f_i(w_i)-f_i(w'_i)|=1$, which means $f_i(V(T_i))=[1,2p-3]^o$. Thereby, $$|f_i(V(T_i))\cap h^*(V(H))|=1$$ and the graph $\odot_1\langle H,T_i\rangle$ is a pan-odd-graceful Topsnut-matching, and so we can claim that $\odot_1\langle H,T_i\rangle^p_1$ is an (a pan-)odd-graceful Topsnut-matching team. The proof of this theorem is complete.
\end{proof}

\begin{figure}[h]
\centering
\includegraphics[height=2.4cm]{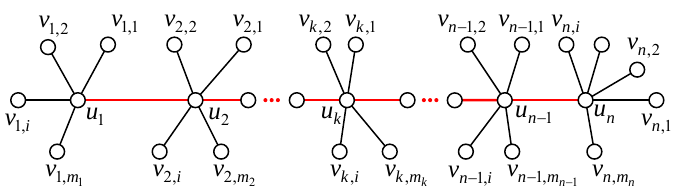}\\
\caption{\label{fig:caterpillar-general}{\footnotesize A caterpillar.}}
\end{figure}

In the proof of Theorem  \ref{thm:Topsnut-matching-team}, we can see
$$
H-v_{n,1}\cong T-v_{n,1}\cong T_i-w_i
$$ with $i\in [1,p]$, so $\odot_1\langle H,T_i\rangle^p_1$ is approximately a uniformly (pan-)odd-gracefully Topsnut-matching team. Thereby, we can get a result: ``\emph{If a tree $T$ of $p$ vertices admits an odd-graceful labelling $f$ such that $f(uv)=1$ for an edge $uv \in E(T)$, where $v$ is a leaf of $T$, that is, $deg_T(v)=1$. Then we have an (a pan-)odd-graceful Topsnut-matching team $\odot_1\langle H,T_i\rangle^p_1$, where $T_i$ is obtained by deleting the leaf $v$ of $T$ and add a new vertex $w$ to the remainder $T-v$, and join $w$ with some vertex of $T-v$; $H$ is a copy of $T$ and admits a pan-odd-graceful labelling $g$ defined by $g(x)=f(x)+1$ for $x\in V(T)=V(H)$}.''

For designing complex Topsnut-gpws, Wang \emph{et al.} (\cite{Wang-Xu-Yao-2017-Twin, Wang-Xu-Yao-2017}) have defined firstly the twin-type of labellings by means of the matching of view (also, \emph{key-vs-lock}).

\begin{defn}\label{defn:odd-elegant-labelling}
\cite{Wang-Xu-Yao-2017-Twin} For two connected $(p_i,q)$-graphs $G_i$ with $i=1,2$, and $p=p_1+p_2-2$, if a $(p,q)$-graph $G=\odot \langle G_1, G_2\rangle$ admits a vertex labelling $f$: $V(G)\rightarrow [0, q]$ such that

(i) $f$ is just an odd-graceful labelling of of $G_1$;

(ii) $f(E(G_2))=\{f(uv)=|f(u)-f(v)|: uv\in E(G_2)\}=[1, q-1]^o$;

(iii) $|f(V(G_1))\cap f(V(G_2))|=k$ and $f(V(G_1))\cup f(V(G_2))\subseteq [0, q-1]$.

Then $f$ is called a \emph{twin odd-graceful labelling} (Tog-labelling) of $G$, and $G$ a \emph{Tog-matching partition}.\qqed
\end{defn}

\begin{defn}\label{defn:twin-odd-elegant-graph}
\cite{Wang-Xu-Yao-2017} For two connected $(p_i,q)$-graphs $G_i$ with $i=1,2$, and $p=p_1+p_2-2$, if a $(p,q)$-graph $G=\odot \langle G_1, G_2\rangle$ admits a vertex labelling $f$: $V(G)\rightarrow [0, q-1]$ such that

(i) $f$ is just an odd-elegant labelling of $G_i$ with $i=1,2$;

(ii) $|f(V(G_1))\cap f(V(G_2))|=k$ and $f(V(G_1))\cup f(V(G_2))\subseteq [0, q-1]$.

Then $f$ is called a \emph{twin odd-elegant labelling} (Toe-labelling) of $G$ (called a \emph{Toe-graph}), and $\odot \langle G_1, G_2\rangle$ is called a \emph{Toe-matching partition}, where $G_1$ is the \emph{Toe-source} and $G_2$ is a \emph{Toe-association}.\qqed
\end{defn}

Fig.\ref{fig:twi-odd-elegant} shows an example for understanding Definition \ref{defn:twin-odd-elegant-graph}. If each $G_i$ with $i=1,2$ is a connected graph in Definition \ref{defn:twin-odd-elegant-graph}, and $G_1$ is a bipartite connected graph with bipartition $(X_1,Y_1)$ holding $f_{\max}(X_1)<f_{\min}(Y_1)$ true, then $f$ is called a \emph{set-ordered twin odd-elegant labelling} (Sotoe-labelling) of $G$ (called a \emph{Sotoe-graph}). Wang \emph{et al.} have shown the algorithms for producing Toe-graphs $G=\odot \langle G_1, G_2\rangle$, such as

\begin{thm} \label{thm:set-ordered-and-twin-odd-elegant}
\emph{\cite{Wang-Xu-Yao-2017}} Every non-star tree $T$ admitting a set-ordered odd-elegant labelling matches with at least two trees $T_1,T_2$ such that each $\odot_2 \langle T, T_i\rangle$ with $i=1,2$ admits a set-ordered twin odd-elegant labelling.
\end{thm}

Wang \emph{et al.} have constructed large scale of Toe-graphs $\odot_2 \langle H_{1}, H_{2}\rangle$ by smaller Toe-graphs $\odot_2 \langle T_i, T'_i\rangle$ with $i\in [1,m]$ for exploring Topsnut-chains. Moreover, they have mixed odd-graceful labelling with odd-elegant labelling together in Definition \ref{defn:2odd2-matching-labelling} below.

\begin{defn}\label{defn:2odd2-matching-labelling}
\cite{Wang-Xu-Yao-2017} For two connected $(p_i,q)$-graphs $G_i$ with $i=1,2$, and $p=p_1+p_2-2$, if a $(p,q)$-graph $G=\odot \langle G_1, G_2\rangle$ admits a vertex labelling $f$: $V(G)\rightarrow [0, q]$ such that

(i) $f$ is an odd-graceful labelling of $G_1$;

(ii) $f:V(G_2)\rightarrow [0,2|E(G_2)|]$ holding
$${
\begin{split}
f(E(G_2))&=\{f(uv)=f(u)+f(v) (\bmod~2|E(G_2)|):\\
&\qquad uv\in E(G_2)\}\\
&=[1, 2|E(G_2)|-1]^{o}
\end{split}}
$$ true.

Then $f$ is called a \emph{2-odd graceful-elegant labelling} (a 2odd2-labelling) of $G$ (called a \emph{2odd2-graph}), and $\odot \langle G_1, G_2\rangle$ is called a \emph{2odd2-matching partition}.\qqed
\end{defn}

In Definition \ref{defn:2odd2-matching-labelling}, if $f$ is a set-ordered odd-graceful labelling of $G_1$, and $G_2$ has its bipartition $(X, Y)$ holding $f_{\max}(X)<f_{\min}(Y)$ true, then we call $f$ a \emph{set-ordered 2odd2-labelling} of $G$. The results on the 2odd2-matching partition can be found in \cite{Wang-Xu-Yao-2017}. The authors of two articles \cite{Wang-Xu-Yao-2017-Twin} and \cite{Wang-Xu-Yao-2017}  propose several conjectures on the twin type of odd-graceful/odd-elegant labellings which mean that Topsnut-gpws made by such labellings are computational security (\cite{Wang-Xu-Yao-2016, Wang-Xu-Yao-Key-models-Lock-models-2016, Wang-Xu-Yao-2017, Wang-Xu-Yao-2017-Twin}).

\begin{figure}[h]
\centering
\includegraphics[height=3cm]{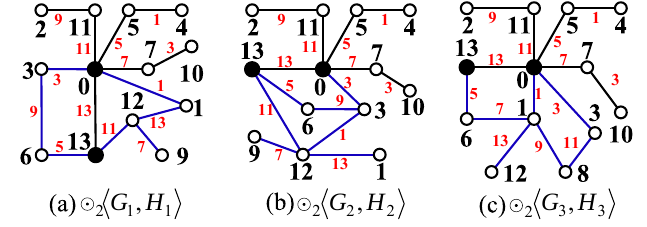}\\
\caption{\label{fig:2-odd-labelling} {\small Three 2odd2-graphs $\odot_2 \langle G_i, H_i\rangle$ for $i\in [1,3]$ cited from \cite{Wang-Xu-Yao-2017}, where $G_i$ has black edges, and $H_i$ has blue edges.}}
\end{figure}

\subsection{Matching partitions based on set-type of labellings}

We use a notation $S^2$ to denote the set of all subsets of a set $S$. For instance, $S=\{a,b,c\}$, so $S^2$ has its own elements: $\{a\}$, $\{b\}$, $\{c\}$, $\{a,b\}$, $\{a,c\}$, $\{b,c\}$ and $\{a,b,c\}$. It is clear that $\emptyset \not\in S^2$ in the paper.

\begin{defn}\label{defn:set-labelling}
$^*$ Let $G$ be a $(p,q)$-graph $G$, and $[0,M]^2$ be a set of all subsets of $[0,M]$ with $M\geq p+q$, and $[a,b]$ be an integer set. We have:

(i) A set mapping $F: V(G)\cup E(G)\rightarrow [0, M]^2$ is called a \emph{total set-labelling} of $G$ if $F(x)\neq F(y)$ for distinct elements $x,y\in V(G)\cup E(G)$.

(ii) A vertex set mapping $F: V(G) \rightarrow [0, M]^2$ is called a \emph{vertex set-labelling} of $G$ if $F(x)\neq F(y)$ for distinct vertices $x,y\in V(G)$.

(iii) An edge set mapping $F: E(G) \rightarrow [0, M]^2$ is called an \emph{edge set-labelling} of $G$ if $F(uv)\neq F(xy)$ for distinct edges $uv, xy\in E(G)$.

(iv) A vertex set mapping $F: V(G) \rightarrow [0, M]^2$ and a proper edge mapping $g: E(G) \rightarrow [a, b]$ are called a \emph{v-set e-proper labelling $(F,g)$} of $G$ if $F(x)\neq F(y)$ for distinct vertices $x,y\in V(G)$ and two edge labels $g(uv)\neq g(wz)$ for distinct edges $uv, wz\in E(G)$.

(v) An edge set mapping $F: E(G) \rightarrow [0, M]^2$ and a proper vertex mapping $f: V(G) \rightarrow [a,b]$ are called an \emph{e-set v-proper labelling $(F,f)$} of $G$ if $F(uv)\neq F(wz)$ for distinct edges $uv, wz\in E(G)$ and two vertex labels $f(x)\neq f(y)$ for distinct vertices $x,y\in V(G)$.\qqed
\end{defn}

\begin{expr}\label{expr:v-set-e-proper}
Fig.\ref{fig:v-set-e-proper}(a) shows a $(6,8)$-graph $G$ admitting a v-set e-proper graceful labelling $(F,g)$ defined by $F:V(G)\rightarrow [0,8]^2$ and $g(E(G))=[1,8]$. And Fig.\ref{fig:v-set-e-proper}(b) shows a $(6,8)$-graph $H$ admitting a v-set e-proper odd-graceful labelling $(F,f)$ defined by $F:V(H)\rightarrow [0,15]^2$ and $f(E(H))$ $=[1,15]^o$.
\end{expr}

\begin{figure}[h]
\centering
\includegraphics[height=3cm]{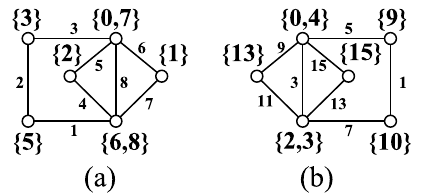}\\
\caption{\label{fig:v-set-e-proper} {\small (a) A v-set e-proper graceful labelling; (b) a v-set e-proper odd-graceful labelling.}}
\end{figure}

\begin{figure}[h]
\centering
\includegraphics[height=10cm]{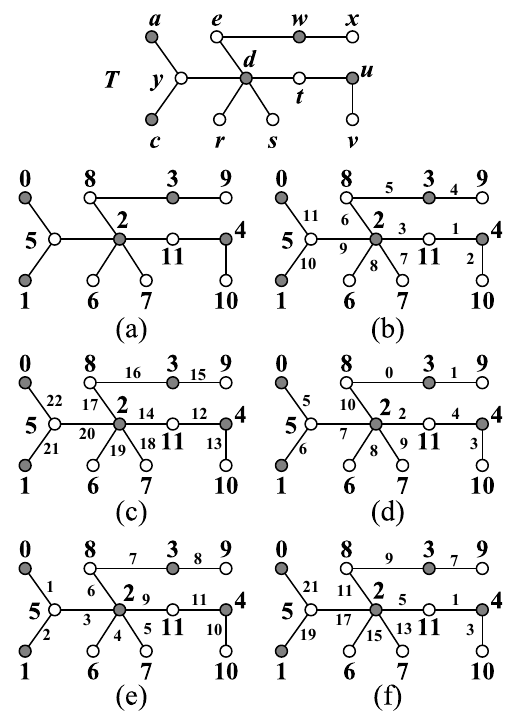}\\
\caption{\label{fig:1-labelling-more-meanning} {\small A lobster $T$ admits: (a) a labelling with null edge-labels: (b) a pan-edge-magic total labelling; (c) a pan-edge-magic total labelling; (d) a felicitous labelling; (e) an edge-magic graceful labelling; (f) an edge-odd-graceful labelling.}}
\end{figure}

\begin{expr}\label{expr:multiple-tree-matching-partition2}
A tree $T$ shown in Fig.\ref{fig:1-labelling-more-meanning} admits an e-set v-proper labelling $(F,f)$ defined by $f:V(T)\rightarrow [0,11]$, and let ``$\bullet$''=``null'', each edge of $T$ has its own label set as follows:

\noindent $F(ay)=\{\bullet,1,5,11,21,22\}$, $F(cy)=\{\bullet,2,6,10,19,21\}$,

\noindent $F(dy)=\{\bullet,3,7,9,17,20\}$, $F(de)=\{\bullet,6,10,11,17\}$,

\noindent $F(dr)=\{\bullet,4,8,15,19\}$, $F(ds)=\{\bullet,5,7,9,13,18\}$,

\noindent $F(dt)=\{\bullet,2,3,5,9,14\}$, $F(ew)=\{\bullet,0,5,7,9,16\}$,

\noindent $F(xw)=\{\bullet,1,4,7,8,15\}$, $F(ut)=\{\bullet,1,4,11,12\}$,

\noindent $F(uv)=\{\bullet,2,3,10,13\}$.

Thereby, this Topsnut-gpw $T$ can produce more complex text-based passwords.
\end{expr}

\begin{thm}\label{thm:total-set-labelling}
If a connected graph $T$ admits mutually distinct graceful labellings $f_1,f_2, \dots, f_m$ and odd-graceful labellings $g_1,g_2, \dots, g_n$ with $m,n\geq 1$, then $T$ admits a total set-labelling.
\end{thm}

There are many trees support Theorem \ref{thm:total-set-labelling}, such as caterpillars, and so do trees admitting set-ordered graceful labellings.

\begin{thm}\label{thm:e-set-v-proper-labellings}
If a tree $T$ admits a set-ordered graceful labelling, then $T$ admits an e-set v-proper labelling $(F,f)$ such that $\max\{|F(uv)|:uv\in E(T)\}\geq 5$.
\end{thm}
\begin{proof} Let a tree $T$ of $p$ vertices has its own vertex bipartition $(X,Y)$ with $X=\{x_{i}:i\in [1,s]\}$ and $Y=\{y_{j}:j\in [1,t]\}$ holding $s+t=|V(T)|=p$ true. By the hypothesis of this theorem, $T$ admits a set-ordered graceful labelling $f$ defined by $f(x_i)=i-1$ for $i\in [1,s]$, $f(y_j)=s+j-1$ for $j\in [1,t]$ and $f(x_iy_j)=f(y_j)-f(x_i)=s+j-i$ for each edge $x_iy_j\in E(T)$. We define another labelling $f^*$ as: $f^*(x_i)=f(x_i)$ for $i\in [1,s]$, $f^*(y_j)=f(y_{t-j+1})$ for $j\in [1,t]$ and $f^*(x_iy_j)=f(x_iy_j)$ for each edge $x_iy_j\in E(T)$.

We define the following labellings:

\begin{asparaenum}[(L-1) ]
\item \label{e-null} $h_1(w)=f^*(w)$ for $w\in V(T)$, $h_1(x_iy_j)=null$ for each edge $x_iy_j\in E(T)$.
\item \label{e-pan-magic-total} We construct a labelling $h_2$ by setting $h_2(w)=f^*(w)$ for $w\in V(T)$, and $h_2(x_iy_j)=f^*(x_iy_j)$ for each edge $x_iy_j\in E(T)$. We verify
$${
\begin{split}
&\quad h_2(x_i)+h_2(x_iy_j)+h_2(y_j)\\
&=f^*(x_i)+f^*(x_iy_j)+f^*(y_j)\\
&=f(x_i)+f(x_iy_j)+f(y_{t-j+1})\\
&=i-1+(s+j-i+1)+s+(t-j+1)-1\\
&=2s+t\\
&=p+s.
\end{split}}$$
Thereby, $h_2$ is really a \emph{pan-edge-magic total labelling} of $T$.

\item \label{e-pan-magic-total-1} Suppose that $h_3(w)=f^*(w)$ for $w\in V(T)$, $h_3(x_iy_j)=p-1+f^*(x_iy_j)$ for each edge $x_iy_j\in E(T)$. Hence, $h_3$ is another \emph{pan-edge-magic total labelling} of $T$.

\item \label{e-set-felicitous} Define a labelling $h_4$ by setting $h_4(w)=f^*(w)$ for $w\in V(T)$ and $h_4(x_iy_j)=h_4(x_i)+h_4(y_j)~(\bmod ~p-1)$ for each edge $x_iy_j\in E(T)$. We compute
$${
\begin{split}
h_4(x_iy_j)&=h_4(x_i)+h_4(y_j)\\
&=f^*(x_i)+f^*(y_j)\\
&=f(x_i)+f(y_{t-j+1})\\
&=f(y_{t-j+1})+f(y_{j})-[f(y_{j})-f(x_i)]\\
&=s+(t-j+1)+s+j-f(x_iy_j)\\
&=p+s-f(x_iy_j),
 \end{split}}$$
which distributes a set
 $\{p+s-1,p+s-2,\dots ,p,p-1,p-2,\dots ,s+1\}~(\bmod ~p-1)=\{s,s-1,\dots ,1,0,p-2,\dots ,s+1\}=[0,p-2]$. So we claim that $h_4$ is a \emph{felicitous labelling} of $T$.

\item \label{e-set-e-magic-graceful} We define a labelling $h_5$ as: $h_5(w)=f^*(w)$ for $w\in V(T)$ and $h_5(x_iy_j)=p-f^*(x_iy_j)=p-f(x_iy_j)$ for each edge $x_iy_j\in E(T)$. Because

$${
\begin{split}
&\quad |h_5(x_i)+h_5(y_j)-h_5(x_iy_j)|\\
&=|f^*(x_i)+f^*(y_j)+f^*(x_iy_j)-p|\\
&=|f(x_i)+f(y_{t-j+1})+f(x_iy_j)-p|\\
&=|i-1+s+(t-j+1)-1+s+j-i+1-p|\\
&=s-1,
\end{split}}$$

then we claim that $h_5$ is an \emph{edge-magic graceful labelling} of $T$ according to  Definition \ref{defn:edge-magic-graceful-labelling}.
\item \label{e-odd-graceful} Let $h_6$ be defined by $h_6(w)=f^*(w)$ for $w\in V(T)$ and $h_6(x_iy_j)=2f^*(x_iy_j)-1=2f(x_iy_j)-1$ for each edge $x_iy_j\in E(T)$. We get

$${
\begin{split}
&\quad h_6(x_i)+h_6(x_iy_j)+h_6(y_j)\\
&=f(x_i)+f(y_{t-j+1})+2f(x_iy_j)-1\\
&=f(x_i)+f(y_{t-j+1})+2[f(y_j)-f(x_i)]-1\\
&=i-1+s+(t-j+1)-1+s+j-i+f(x_iy_j)-1\\
&=p+s-2+f(x_iy_j),
\end{split}}$$
which induces a set $\{p+s-1,p+s+1,\dots, 2p+s-3\}=[p+s,2p+s-3]$.
\end{asparaenum}

We can see relationships between the above labellings: $h_2(x_iy_j)+p-1=h_3(x_iy_j)$, $${
\begin{split}
h_2(x_iy_j)+h_4(x_iy_j)&=f(x_iy_j)+p+s-f(x_iy_j)\\
&\equiv s-1 ~(\bmod~p-1)
\end{split}}
$$
and
$$h_2(x_iy_j)+h_5(x_iy_j)=f(x_iy_j)+p-f(x_iy_j)=p.$$

Now, we define the desired e-set v-proper labelling $F$ in the way: $F(x_iy_j)=\{h_k(x_iy_j):~k\in [1,6]\}$ for each edge $x_iy_j\in E(T)$, so $T$ admits an e-set v-proper labelling $(F,f^*)$ with $\max\{|F(uv)|:uv\in E(T)\}\geq 6$. This theorem is covered.
\end{proof}

\begin{defn}\label{defn:multiple-trees-labelling}
$^*$ If a $(p,q)$-graph $G$ admits an e-set v-proper labelling $(F,f)$ by $f: V(G)\rightarrow [0, a(p,q)]$ and $F: E(G)\rightarrow [0, b(p,q)]^2$, where $a(p,q)$ and $b(p,q)$ are linear functions of $p$ and $q$, such that $G$ can be decomposed into spanning trees $T_1,T_2,\dots ,T_m$ with $m\geq 2$ and $E(G)=\bigcup^m_{i=1}E(T_i)$ (allow $E(T_i)\cap E(T_j)\neq \emptyset $ for some $i\neq j$), and each spanning tree $T_i$ admits a proper labelling $f_i$ induced by $(F,f)$. We call $G$ a \emph{multiple-tree matching partition}, denoted as $G=\oplus_F\langle T_i\rangle ^m_1$.\qqed
\end{defn}

\begin{figure}[h]
\centering
\includegraphics[height=6.2cm]{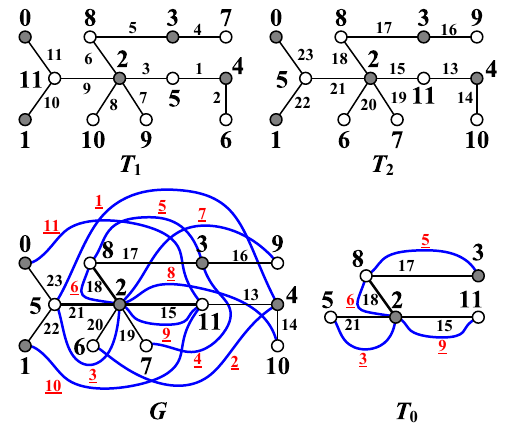}\\
\caption{\label{fig:twin-trees} {\small A multiple-tree matching partition $G$ for illustrating Definition \ref{defn:multiple-trees-labelling}.}}
\end{figure}

\begin{expr}\label{expr:multiple-tree-matching-partition1}
A multiple-tree matching partition $G$ shown in Fig.\ref{fig:twin-trees} has two spanning trees $T_1,T_2$ and an e-set v-proper labelling $F$, such that $T_1$ admits a set-ordered graceful labelling $f_1$ induced by $F$, and $T_2$ admits a super pan-edge-magic total labelling $f_2$ induced by $F$. In detail, $f_1(u)\in f(V(G))$ for $x\in V(T_1)=V(G)$, each edge $uv\in E(T_1)$ is labeled as $f_1(uv)=|f_1(u)-f_1(v)|\in F(uv)$; and $f_2(u)\in f(V(G))$ for $y\in V(T_2)=V(G)$, $f_2(uv)\in F(uv)$ for each edge $uv\in E(T_2)$, such that $f_2(u)+f_2(uv)+f_2(v)=28$ for $uv\in E(T_2)$. We can see $T_0=T_1\cap T_2$, called the \emph{common body} of $G$, each edge of $T_0$ is labeled with a set.
\end{expr}

Similarly with Definition \ref{defn:multiple-trees-labelling}, we propose:
\begin{defn}\label{defn:multiple-graph-matching}
$^*$ If a $(p,q)$-graph $G$ admits a vertex labelling $f: V(G)\rightarrow [0, p-1]$, such that $G$ can be decomposed into (spanning) graphs $G_1,G_2,\dots ,G_m$ with $m\geq 2$ and $E(G)=\bigcup^m_{i=1}E(G_i)$ and $E(G_i)\cap E(G_j)=\emptyset $ for  $i\neq j$, and each graph $G_i$ admits a proper labelling $f_i$ induced by $f$. We call $G$ a \emph{multiple-graph matching partition}, denoted as $G=\odot_f\langle G_i\rangle ^m_1$.\qqed
\end{defn}

In Fig.\ref{fig:generalized-multiple-tree}, we can see that $E(G_i)\cap E(G_j)=\emptyset $ for  $i\neq j$, and $G_1$ admits a graceful labelling $f_1$ induced by $f$, each $G_j$ admits a felicitous labelling $f_j$ induced by $f$ with $j\in [2,4]$. $H_j=\odot _{k_j}\langle G_1,G_j\rangle $, with $j\in [2,4]$, admits the labelling $f$ such that $f(E(G_1))=\{|f(u)-f(v)|:~uv\in E(G_1)\}=[1,11]$, $f(E(G_j))=\{f(x)+f(y)~(\bmod ~11):~xy\in E(G_j)\}=[1,11]$ (where we set $0\equiv 11~(\bmod ~11)$), so each $H_j=\odot _{k_j}\langle G_1,G_j\rangle $ with $j\in [2,4]$ is called a \emph{graceful-felicitous matching partition}.

\begin{thm}\label{thm:set-ordered-matchings-10-labellings}
If a tree $T$ admits a set-ordered graceful labelling, then $T$ matches with a multiple-tree matching partition $\oplus_F\langle T_i\rangle ^m_1$ with $m\geq 10$.
\end{thm}
\begin{proof} Let $(X,Y)$ be the bipartition of vertex set of the tree $T$ admitting a set-ordered graceful labelling $f$, where $X=\{x_{i}:i\in [1,s]\}$ and $Y=\{y_{j}:j\in [1,t]\}$ with $s+t=|V(T)|=p$. By the definition of a set-ordered graceful labelling, we have $f(x_i)=i-1$ for $i\in [1,s]$, $f(y_j)=s+j-1$ for $j\in [1,t]$ and $f(x_iy_j)=f(y_j)-f(x_i)=s+j-i$ for each edge $x_iy_j\in E(T)$. Now, let $T=T_1$ and $f=f_1$ for the purpose of statement. Clearly, $\max f_1(X)<\min f_1(Y)$.

We construct the following trees.

\begin{asparaenum}[(Equ-1) ]
\item \label{edge-magic-y} The tree $T_2$ is isomorphic to $T$, and admits an edge-magic total labelling $f_2$ defined as: $f_2(x_i)=f(x_i)+1$ for $i\in [1,s]$, $f_2(y_j)=f(y_{t-j+1})+1$ for $j\in [1,t]$, and $f_2(x_iy_j)=f(x_iy_j)+p$ for each edge $x_iy_j\in E(T)$. Notice that $f_2(V(T_2))=[1,p]$. Furthermore,
\begin{equation}\label{eqa:edge-magic-y}
{
\begin{split}
&\quad f_2(x_i)+f_2(x_iy_j)+f_2(y_j)\\
&=f(x_i)+f(x_iy_j)+p+f(y_{t-j+1})\\
&=s+2p+1.
\end{split}}
\end{equation}
So, $f_2$ is a super edge-magic total labelling with
the magic constant $s+2p+1$ and $\max f_2(X)<\min f_2(Y)$.

\item \label{edge-magic-x} We take a copy of $T$, denoted as $T_3$, and define an edge-magic total labelling $f_3$ of $T_3$ in the way: $f_3(x_i)=f(x_{s-i+1})+1$ for $i\in [1,s]$, $f_3(y_j)=f(y_j)+1$ for $j\in [1,t]$, and $f_2(x_iy_j)=f(x_iy_j)+p$ for each edge $x_iy_j\in E(T)$. So $f_3(V(T_3))=[1,p]$. By the same way used in Equ-\ref{edge-magic-y}, $f_3$ is a super edge-magic total labelling with
the magic constant $t+2p+1$ and $\max f_3(X)<\min f_3(Y)$.

\item \label{felicitous-y} Let $T_4$ be a tree being isomorphic to $T$. And $T_4$ admits a super felicitous
labelling $f_4$ made by: $f_4(x_i)=f(x_i)$ for $i\in [1,s]$, $f_4(y_j)=f(y_{t-j+1})$ for
$j\in [1,t]$. Moreover, we have
\begin{equation}\label{eqa:felicitous-y}
{
\begin{split}
f_4(x_i)+f_4(y_j)&=f(x_i)+f(y_{t-j+1})\\
&=f(x_i)+s+p-1-f(y_j)\\
&=s+p-1-\big [f(y_j)-f(x_i)\big ]\\
&=s+p-1-f(x_iy_j),
\end{split}}
\end{equation}
The above form (\ref{eqa:felicitous-y}) induces two sets 
$${
\begin{split}
S_1=&\{s+p-1-1,s+p-1-2,\dots
,\\
&s+p-1-(s-1),s+p-1-s\}
\end{split}}
$$ and $S_2=\{p-2,p-3,\dots ,s\}$. Under modulo
$(p-1)$, $S_1\equiv [0,s-1]$. Thereby,
$f_4(E(T_4))=[0,p-2]$. We claim that $f_4$ is a super felicitous
labelling of $T_4$ with $\max f_4(X)<\min f_4(Y)$.

\item \label{felicitous-x} $T_5$ is isomorphic to $T$, and admits a super felicitous
labelling $f_5$ with $\max f_5(X)<\min f_5(Y)$. made by: $f_5(x_i)=f(x_{s-i+1})$ for $i\in [1,s]$, $f_5(y_j)=f(y_{j})$ for
$j\in [1,t]$. The remainder proof is as the same as that in Equ-\ref{felicitous-y}.

\item \label{antimagic-y} Let $T_6\cong T$. We define a labelling $f_6$ of $T_6$ in the way:
$f_6(x_i)=f(x_i)+1$ for $i\in [1,s]$, $f_6(y_j)=f(y_{t-j+1})+1=s+p-f(y_{j})$ for $j\in [1,t]$, and
$f_6(x_iy_j)=2p-f(x_iy_j)$ for each edge $x_iy_j\in E(T_6)$. It is not hard to see $f_6(V(T_6))=[1,p]$. We have
$$f_6(x_i)+f_6(x_iy_j)+f_6(y_j)=s+3p+1-2f(x_iy_j)$$ which produces a
set $\{p+s+3,p+s+3+2,p+s+3+4, \dots ,p+s+3+2(p-2)\}$. We can confirm that $f_6$
is a super edge antimagic total labelling with $\max f_6(X)<\min f_6(Y)$.

\item \label{antimagic-x} Take a tree $T_7\cong T$, and define a super edge antimagic total labelling $f_7$ of $T_7$ with $\max f_7(X)<\min f_7(Y)$ as follows: $f_7(x_i)=f(x_{s-i+1})+1$ for $i\in [1,s]$, $f_7(y_j)=f(y_{j})+1=s+p-f(y_{j})$ for $j\in [1,t]$, and $f_7(x_iy_j)=2p-f(x_iy_j)$ for each edge $x_iy_j\in E(T_7)$. The remainder proof is very similar to that in Equ-\ref{antimagic-y}.

\item \label{harmonious-y} Suppose that $T_8$ is a copy of $T$, we define a labelling $f_8$ of $T_8$ in the way that
$f_8(x_i)=f(x_i)$ for $i\in [1,s]$, $f_8(y_j)=f(y_{t-j+1})$ for
$j\in [1,t-1]$, and $f_8(y_t)=0$. For each edge $x_ky_t\in E(T_8)$, we have
\begin{equation}\label{eqa:harmonious-y}
{
\begin{split}
&\quad f_8(x_k)+f_8(y_t)=f_8(x_k)+0\\
&=f_8(x_k)+(p-1)~(\bmod~p-1)\\
&=f(x_k)+f(y_{t})~(\bmod~p-1)\\
&=f(x_k)+s+p-1-f(y_1)~(\bmod~p-1)\\
&=s+p-1-\big [f(y_1)-f(x_k)\big ]~(\bmod~p-1)\\
&=s-f(x_ky_1).
\end{split}}
\end{equation} Under modulo $(p-1)$,
$${
\begin{split}
&\quad f_8(E(T_8))=\{f_8(x_iy_j)\\
&=f_8(x_i)+f_8(y_j)~ (\bmod ~p-1): x_iy_j\in
E(T_8)\}\\
&=[0,p-2].
\end{split}}$$ Therefore, $f_8$ is a harmonious labelling of $T_8$ with $\max f_8(X)<\min f_8(Y)$.

\item \label{harmonious-x} Let $T_9\cong T$. We define a harmonious labelling $f_9$ of $T_9$ with $\max f_9(X)<\min f_9(Y)$ as follows: $f_9(x_i)=f(x_{s-i+1})$ for $i\in [1,s]$, $f_9(y_j)=f(y_{j})$ for $j\in [1,t-1]$, and $f_9(y_t)=0$. By the same way used in Equ-\ref{harmonious-y}, we can show $f_9$ is a harmonious labelling of $T_9$ and $\max f_9(X)<\min f_9(Y)$.

\item \label{harmonious-x} By Definition \ref{defn:Dgemm-labelling} we define a Dgemm-labelling $f_{10}$ of $T_{10}$ that holds $T_{10}\cong T$ true as following: $f_{10}(x_i)=f(x_{i})$ for $i\in [1,s]$, $f_{10}(y_j)=f(y_{j})$ for $j\in [1,t]$, and $f_{10}(x_iy_j)=p-f(x_iy_j)$ for each edge $x_iy_j\in E(T_{10})$. We verify:

\quad (i) Each edge $x_iy_j$ corresponds an edge $x'_iy'_j$ such that
$${
\begin{split}
f_{10}(x_iy_j)&=p-f(x_iy_j)=p-|f(x_i)-f(y_j)|\\
&=|f(x'_i)-f(y'_j)|.
\end{split}}$$

\quad (ii) For each edge $x_iy_j\in E(T_{10})$,  $T_{10}$ has $p-1$ edges, and
$${
\begin{split}
s(x_iy_j)&=|f_{10}(x_i)-f_{10}(y_j)|-f_{10}(x_iy_j)\\
&=|f(x_i)-f(y_j)|-p+f(x_iy_j)\\
&=2f(x_iy_j)-p
\end{split}}$$ distributes a set $\{2-p,4-p,\dots, 0, 2,4, \dots, p-2\}$ when $p$ is even. Clearly, each edge $x_iy_j$ matches with another edge $x'_iy'_j\in E(T_{10})$ holding $s(x_iy_j)+s(x'_iy'_j)=0$ true. If $p$ is odd, $s(x_iy_j)$ induces another set $\{2-p,4-p,\dots, -1, 1,3, \dots, p-2\}$, so $s(x_iy_j)+s(x'_iy'_j)=0$ is true.

\quad (iii) For each edge $x_iy_j\in E(T_{10})$, we have
$${
\begin{split}
&\quad f_{10}(x_iy_j)+|f_{10}(x_i)-f_{10}(y_j)|\\
&=p-f(x_iy_j)+|f_{10}(x_i)-f_{10}(y_j)|\\
&=p.
\end{split}}$$

\quad (iv) Since $f(V(T))=[0,p-1]$ and $f(E(T))=[1,p-1]$, thus, $f(E(T))=[1,p-1]=f(V(T))\setminus \{0\}$, we have $f(x_iy_j)=f(w)$ for each edge $x_iy_j$ matching with a vertex $w$, which implies
$$f_{10}(x_iy_j)+f_{10}(w)=p-f(x_iy_j)+f(w)=p;$$ conversely, each vertex $w$ corresponds an edge $x_iy_j$ holding $$f_{10}(w)+f_{10}(x_iy_j)=p-f(x_iy_j)+f(w)=p$$  true, except the \emph{singularity} $f(x_0)=\lfloor p\rfloor $.
\end{asparaenum}

\begin{figure}[h]
\centering
\includegraphics[height=12cm]{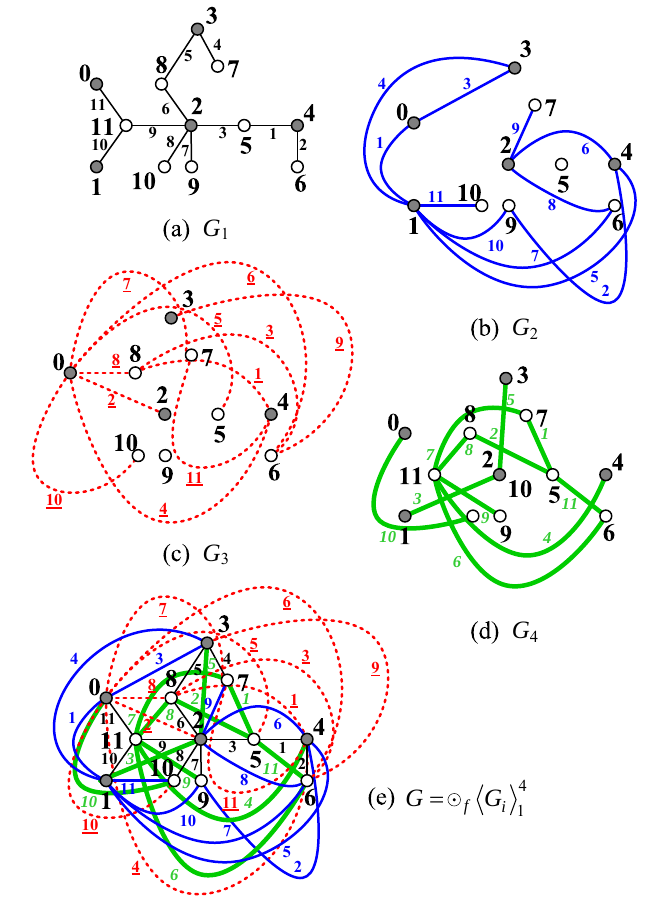}\\
\caption{\label{fig:generalized-multiple-tree} {\small An example for illustrating Definition \ref{defn:multiple-graph-matching}.}}
\end{figure}

We identify the vertices of $T_1,T_2,\dots ,T_{10}$ with the same labels into one, and delete the multiple edges, the resulting graph is just the desired multiple-tree matching partition $G=\oplus_F\langle T_i\rangle ^{10}_1$ with the v-proper labelling $f$ holding $f(V(G))=[0,p-1]$ true and the e-set labelling $F$ satisfying $F(x_iy_j)=\{f_{k}(x_iy_j):~k\in [1,10]\}$ for each edge $x_iy_j\in E(G)$.

The labellings $f_2,f_3,\dots ,f_{10}$ shown in the above proof can deduce the set-ordered graceful labelling $f$, we omit the proof since the proof methods are similar to that in \cite{Yao-Liu-Yao-2017}.
\end{proof}

\begin{thm}\label{thm:total-set-labellings}
If a bipartite $(p,q)$-graph $G$ admits a set-ordered graceful labelling, then $G$ admits a set-ordered graceful/odd-graceful total set-labelling.
\end{thm}
\begin{proof} Suppose that $G$ has its own vertex bipartition $(X,Y)$ with $X=\{x_{i}:i\in [1,s]\}$ and $Y=\{y_{j}:j\in [1,t]\}$ holding $s+t=|V(G)|=p$ true. By the hypothesis of this theorem, $G$ admits a set-ordered graceful labelling $f$ defined by $f(x_i)=i-1$ for $i\in [1,s]$, $f(y_j)=s+j-1$ for $j\in [1,t]$ and
$$f(x_iy_j)=f(y_j)-f(x_i)=s+j-i
$$ for each edge $x_iy_j\in E(G)$, such that $f(E(G))=[1,q]$.

We define a total  set-labelling $F: V(G)\cup E(G)\rightarrow [0, p+q]^2$  as follows: $F(x_i)=[0,i-1]$ with $i\in [1,s]$, $F(y_j)=[0,s+j-1]$ with $j\in [1,t]$ and $F(x_iy_j)=F(y_j)\setminus F(x_i)=[i,s+j-1]$ for each edge $x_iy_j\in E(G)$. Clearly,  $F(u)\neq F(w)$ for distinct elements $u,w\in V(G)\cup E(G)$. Moreover, we have
\begin{equation}\label{eqa:c3xxxxx}
{
\begin{split}
&\quad \{|F(x_iy_j)|:i\in [1,s],j\in [1,t]\}\\
&=\{s+j-i:i\in [1,s],j\in [1,t]\}\\
&=\{f(x_iy_j):x_iy_j\in E(G)\}\\
&=f(E(G))=[1,q],
\end{split}}
\end{equation}
and $F(x_i)\subset F(y_j)$ with $i=|F(x_i)|<|F(y_j)|=s+j$. Thereby, $F$ is really  a set-ordered graceful total set-labelling of $G$.

For proving that $G$ admits a set-ordered odd-graceful total set-labelling, we set a total  set-labelling $F': V(G)\cup E(G)\rightarrow [0, p+q]^2$  in the following way: $F'(x_i)=[0,2(i-1)]$ with $i\in [1,s]$, $F'(y_j)=[0,2(s+j)-3]$ with $j\in [1,t]$ and
$$
F'(x_iy_j)=F'(y_j)\setminus F'(x_i)=[2(i-1)+1,2(s+j)-3]
$$ for each edge $x_iy_j\in E(G)$. Furthermore, $F'(x_i)\subset F'(y_j)$ with
$$2i-1=|F'(x_i)|<|F'(y_j)|=2(s+j)-2$$
and
$${
\begin{split}
\{|F'(x_iy_j)|:i\in [1,s],j\in [1,t]\}&=f(E(G))\\
&=[1,2q-1]^o.
\end{split}}
$$

We are done for the proof of the theorem.
\end{proof}

By the results appeared in \cite{Yao-Liu-Yao-2017}, we can show other type of  total set-labelling on bipartite graphs admitting set-ordered graceful labellings.

The graceful graph $G$ shown in Fig.\ref{fig:K5-v-set-e-proper}(b) is a \emph{graceful matching} of $K_5$ admitting a v-set e-proper graceful labelling, and an odd-graceful graph $H$ shown in Fig.\ref{fig:K5-v-set-e-proper}(d) is an \emph{odd-graceful matching} of $K_5$ having a v-set e-proper odd-graceful labelling. The graceful graph $G$ and the odd-graceful graph $H$ can be obtained from $K_5$ by the vertex-split operation introduced in Section I. Two graphs shown in Fig.\ref{fig:K5-v-set-e-proper}(b) and (d) can be shrunk back to $K_5$. Identifying two non-adjacent vertices $u,v$ of a graph $H$ into one $w=u\circ v$ if $N(u)\cap N(v)=\emptyset $ until, any pair of vertices $x,y$ of the last graph $H^*$ hold $|N(u)\cap N(v)|\geq 1$ true. Clearly, $|E(H)|=|E(H^*)|$. We call $H^*$ a non-contracted graph, $H$ has a \emph{non-contracted $H^*$-kernel}. Two graphs shown in Fig.\ref{fig:K5-v-set-e-proper}(b) and (d) both have a non-contracted $K_5$-kernel.

\begin{lem}\label{thm:non-contracted-H-kernel-0}
Each complete graph $K_n$ admits a v-set e-proper (odd-)graceful labelling.
\end{lem}
\begin{proof} We have known that $K_5$ admits a v-set e-proper graceful labelling. Assume that $K_n$ admits a v-set e-proper graceful labelling $(F,f)$ such that $F:V(K_n)\rightarrow [0,M_n]$ with $F(x_i)\cap F(x_j)=\emptyset$ for distinct $x_i,x_j\in V(K_n)$ and $f(E(K_n))=[1,M_n]$, where $M_n=\frac{1}{2}n(n-1)$. We add a new vertex $x_{n+1}$ to $K_n$ by joining $x_{n+1}$ with each vertex of $K_n$, and label it with the number $M_n+n$. If $x_i\in V(K_n)$ such that $M_n+n-a_{i,k_i}=M_n+k_i$ for $a_{i,k_i}\in F(x_i)=\{a_{i,1},a_{i,2},\dots ,a_{i,m_i}\}$ with $1\leq m_i\leq M_n$, we put $x_i$ into a set $S$ such that each $F(x_j)$ with $x_j\in \overline{S}=V(K_n)\setminus S$ has no element of $F(x_j)=\{a_{j,1},a_{j,2},\dots ,a_{j,m_j}\}$ with $1\leq m_j\leq M_n$ holds $M_n+n-a_{j,l}=M_n+l$ true. We add an number $b_j\in [M_n+1,M_n+n]\setminus \{n-k_i:x_i\in S\}$ to $F(x_j)$ for $x_j\in \overline{S}$, one-vs-one, thus, we get a v-set e-proper graceful labelling $(F',f')$ of $K_{n+1}$ holding $F':V(K_{n+1})\rightarrow [0,\frac{1}{2}n(n+1)]$ true with $F'(x)\cap F'(y)=\emptyset$ for distinct $x,y\in V(K_{n+1})$ and $f'(E(K_{n+1}))=[1,\frac{1}{2}n(n+1)]$.

By the induction of hypothesis, we claim that each complete graph admits a v-set e-proper graceful labelling, and furthermore this proof way can be used to show each complete graph admits a v-set e-proper odd-graceful labelling.
\end{proof}

Lemma \ref{thm:non-contracted-H-kernel-0} enables us to obtain the following result:

\begin{thm}\label{thm:non-contracted-H-kernel}
A $(p,q)$-graph $G$ with a non-contracted $H$-kernel admits a proper $\varepsilon$-labelling if and only if $H$ admits a v-set e-proper $\varepsilon$-labelling.
\end{thm}

\begin{figure}[h]
\centering
\includegraphics[height=7.6cm]{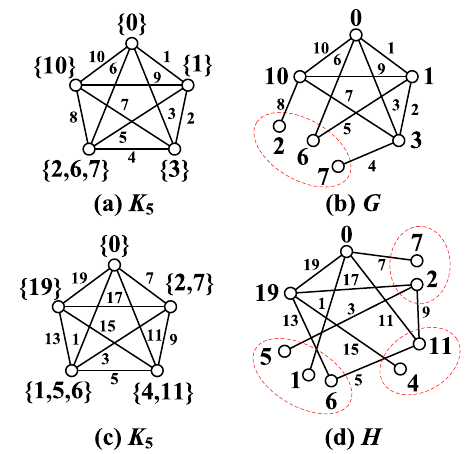}\\
\caption{\label{fig:K5-v-set-e-proper} {\small (a) $K_5$ admits a v-set e-proper graceful labelling; (b) $G$ admits a proper graceful labelling; (c) $K_5$ admits a v-set e-proper odd-graceful labelling; (d) $H$ admits a proper odd-graceful labelling.}}
\end{figure}

\begin{figure}[h]
\centering
\includegraphics[height=7.8cm]{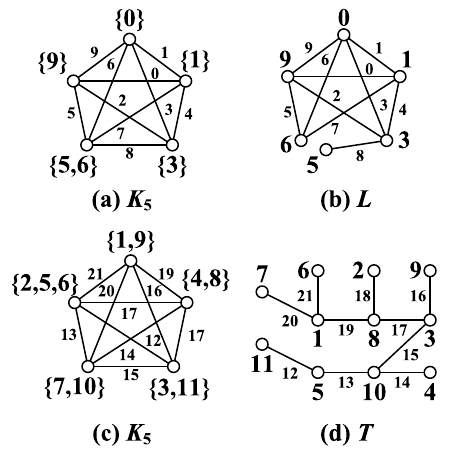}\\
\caption{\label{fig:K5-v-set-felicitous-magic} {\small (a) $K_5$ admits a v-set e-proper felicitous labelling; (b) $L$ admits a felicitous labelling $f$ holding $f(uv)=f(u)+f(v)~(\bmod ~10)$ true; (c) $K_5$ admits a v-set e-proper edge-magic total labelling; (d) $T$ admits an edge-magic total labelling.}}
\end{figure}

\begin{lem}\label{thm:graph-split-tree}
Any connected $(p,q)$-graph can be split into a tree of $q+1$ vertices.
\end{lem}
\begin{proof} Our proof is based on induction of vertex number. As $p=2$ and $q=1$, the lemma is obvious. Assume that a connected $(p,q)$-graph can be split into a tree of $q+1$ vertices. We consider any connected $(p+1,q')$-graph $G$. There exists a spanning tree $T$ in $G$, since $G$ is connected. We take a leaf $x$ of $T$, so the graph $G-x$ is a connected $(p,q'-m)$-graph, where $m=|N(x)|$, and the neighbor set $N(x)$ collects all neighbors $x_1,x_2,\dots x_m$ of the vertex $x$. By the hypothesis of induction, the connected $(p,q'-m)$-graph $G-x$ can be split into a tree $H$ of $q'-m$ vertices. Suppose that each vertex $x_i\in N(x)$ was split into $x_{i,1},x_{i,2}, \dots ,x_{i,m_i}$ in $H$. We add new vertices $x'_i$ to $x_{i,1}$ by an edge $x'_ix_{i,1}$ with $i\in [1,m]$, the result graph is just a tree $H'$ of $q'-m+m$ vertices. Thereby, $H'$ is the desired tree split from $G$, and the leaves $x'_1,x'_2,\dots ,x'_m$ of $H'$ are the result of splitting the vertex $x$ of $G$.
\end{proof}

By Lemma \ref{thm:graph-split-tree} we can see: ``If every tree is (odd-)graceful, then any connected $(p,q)$-graph admits a v-set e-proper (odd-)graceful labelling. Conversely, if a connected $(p,q)$-graph $G$ admits a v-set e-proper (odd-)graceful labelling $(F,f)$ holds $f(E(G))=[1,q]$ (resp. $f(E(G))=[1,2q-1]$) true and $\bigcup_{u\in V(G)}F(u)=[0,q]$ (resp. $\bigcup_{u\in V(G)}F(u)\subset [0,2q-1]$), then it can be split into a (an odd-)graceful tree''. For a connected $(p,q)$-graph $G$ holding $q\equiv 0,3~(\bmod~4)$ true, we can say $G$ admits a v-set e-proper (odd-)graceful labelling. Thereby, we conjecture:
\begin{conj}\label{conj:v-set-e-proper}
Each connected graph with no multiple edges and self-loops admits a v-set e-proper (odd-)graceful labelling.
\end{conj}

In \cite{Hui-Sun-Bing-Yao2018}, the authors show some Euler graphs admit v-set e-proper $X$-labellings with $X \in \{$graceful, odd-graceful, harmonious, $k$-graceful, odd sequential, elegant, odd-elegant, felicitous, odd-harmonious, edge-magic total$\}$. Thereby, we can generalize Conjecture \ref{conj:v-set-e-proper} to other labellings. Zhou \emph{et al.} in\cite{Zhou-Yao-Chen-Tao2012} have proven: Lobsters admit odd-graceful labellings, then we have

\begin{thm}\label{thm:caterpillar-lobster-v-set-e-proper}
If  a connected graph can be split into a tree admitting a v-set e-proper $X$-labelling, where $X$ is a graph labelling, then it admits a v-set e-proper $X$-labelling.
\end{thm}

Notice that a  connected $(p,q)$-graph $G$ admitting a graceful labelling $f$ can be split into a tree $T$ of $q$ edges, then $T$ admits a \emph{splitting  graceful labelling} $g$ induced by $f$, here, there are at least two vertices $u,v$ holding $g(u)=g(v)$ true under $g(V(T))\subseteq [0,q]$, but $g(E(T))=[1,q]$.

\begin{conj}\label{conj:splitting-graceful-trees}
Each tree with diameter less than three admits a splitting graceful labelling.
\end{conj}

\subsection{Magic type of matching labellings}

\begin{defn}\label{defn:relaxed-Emt-labelling}
$^*$ Let $f:V(G)\cup E(G)\rightarrow [1,p+q]$ be a total labelling of a $(p,q)$-graph $G$. If there is a constant $k$ such that $$f(u)+f(uv)+f(v)=k,$$ and each edge $uv$ corresponds another edge $xy$ holding $f(uv)=|f(x)-f(y)|$ true, then we call $f$ a \emph{relaxed edge-magic total labelling} (relaxed Emt-labelling) of $G$ (called a \emph{relaxed Emt-graph}) (see Fig.\ref{fig:some-examples}(a)).\qqed
\end{defn}

\begin{defn}\label{defn:Oemt-labelling}
$^*$ Suppose that a $(p,q)$-graph $G$ admits a vertex labelling $f:V(G) \rightarrow [0,2q-1]$ and an edge labelling $g:E(G)\rightarrow [1,2p-1]^o$. If there is a constant $k$ such that $$f(u)+g(uv)+f(v)=k$$ for each edge $uv\in E(G)$, and $g(E(G))=[1,2p-1]^o$, then we call $(f,g)$ an \emph{odd-edge-magic matching labelling} (Oemm-labelling) of $G$ (called an \emph{Oemm-graph}). See Fig.\ref{fig:odd-magic-2-labellings}(c),Fig.\ref{fig:odd-graceful-2-magic}(a) and Fig.\ref{fig:odd-graceful-2-magic} (b).\qqed
\end{defn}

\begin{defn}\label{defn:relaxed-Oemt-labelling}
$^*$ Suppose that a $(p,q)$-graph $G$ admits a vertex labelling $f:V(G)\rightarrow [0,2q-1]$ and an edge labelling $g:E(G)\rightarrow [1,2q-1]^o$, and let $s(uv)=|f(u)-f(v)|-g(uv)$ for $uv\in E(G)$. If

(i) each edge $uv$ corresponds an edge $u'v'$ such that $g(uv)=|f(u')-f(v')|$;

(ii) and there exists a constant $k'$ such that each edge $xy$ has a matching edge $x'y'$ holding $s(xy)+s(x'y')=k'$ true;

(iii) there exists a constant $k$ such that $f(uv)+|f(u)-f(v)|=k$ for each edge $uv\in E(G)$.

Then we call $(f,g)$ an \emph{ee-difference odd-edge-magic matching labelling} (Eedoemm-labelling) of $G$ (called a \emph{Eedoemm-graph}). (see Fig.\ref{fig:odd-graceful-2-magic}(a) and (b))\qqed
\end{defn}

\begin{figure}[h]
\centering
\includegraphics[height=6cm]{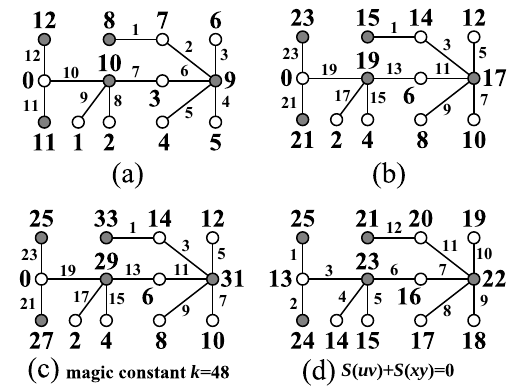}\\
\caption{\label{fig:odd-magic-2-labellings} {\small A caterpillar admits: (a) a set-ordered graceful labellings; (b) a set-ordered odd-graceful labellings; (c) an odd-edge-magic matching labelling $(f,g)$; (d) an ee-difference odd-magic matching labelling $(f,g)$.}}
\end{figure}

We, again, define a new labelling with more restrictive conditions as follows:

\begin{defn}\label{defn:6C-labelling}
$^*$ A total labelling $f:V(G)\cup E(G)\rightarrow [1,p+q]$ for a bipartite $(p,q)$-graph $G$ is a bijection and holds:

(i) (e-magic) $f(uv)+|f(u)-f(v)|=k$ for a constant $k$ and each edge $uv\in E(G)$;

(ii) (ee-difference) each edge $uv\in E(G)$ matches with another edge $xy\in E(G)$ holding $f(uv)=|f(x)-f(y)|$ (or $f(uv)=2(p+q)-|f(x)-f(y)|$);

(iii) (ee-balanced) let $s(uv)=|f(u)-f(v)|-f(uv)$ for $uv\in E(G)$, then there exists a constant $k'$ such that each edge $uv$ matches with another edge $u'v'$ holding $s(uv)+s(u'v')=k'$ (or $2(p+q)+s(uv)+s(u'v')=k'$) true;

(iv) (EV-ordered) $f_{\min}(V(G))>f_{\max}(E(G))$ (or $f_{\max}(V(G))<f_{\min}(E(G))$, or $f(V(G))\subseteq f(E(G))$, or $f(E(G))\subseteq f(V(G))$, or $f(V(G))$ is an odd-set and $f(E(G))$ is an even-set);

(v) (ve-matching) there exists a constant $k''$ such that each edge $uv$ matches with one vertex $w$ such that $f(uv)+f(w)=k''$, and each vertex $z$ matches with one edge $xy$ such that $f(z)+f(xy)=k''$, except the \emph{singularity} $f(x_0)=\lfloor \frac{p+q+1}{2}\rfloor $;

(vi) (set-ordered) $f_{\max}(X)<f_{\min}(Y)$ (or $f_{\min}(X)>f_{\max}(Y)$) for the bipartition $(X,Y)$ of $V(G)$.

We then call $f$ a \emph{6C-labelling}.\qqed
\end{defn}

For a given $(p,q)$-tree $G$ admitting a 6C-labelling $f$, if another $(p,q)$-tree $H$ admits a 6C-labelling $g$ such that
$$
f(V(G))\setminus X^*=g(E(H)),~f(E(G))=g(V(H))\setminus X^*
$$ and $f(V(G))\cap g(V(H))=X^*$, where $X^*=\{\lfloor \frac{p+q+1}{2}\rfloor \}$, we identify the vertex $x_0$ of $G$ having $f(x_0)=\lfloor \frac{p+q+1}{2}\rfloor $ with the vertex $w_0$ of $H$ having $g(w_0)=\lfloor \frac{p+q+1}{2}\rfloor $ into one to form a graph $\odot_1\langle G,H \rangle $, called a \emph{6C-complementary matching}. See examples shown in Fig.\ref{fig:ev-set-exchanged} and Fig.\ref{fig:odd-even-separable}.

\begin{figure}[h]
\centering
\includegraphics[height=5.6cm]{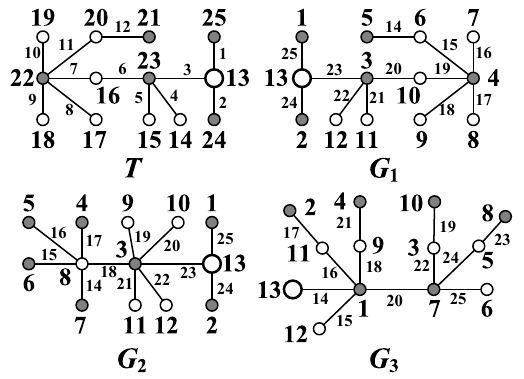}\\
\caption{\label{fig:ev-set-exchanged} {\small A caterpillar $T$ admitting a 6C-labelling has three 6C-complementary matchings $\odot_1\langle T,G_i\rangle$ with $i=1,2,3$.}}
\end{figure}

\begin{thm}\label{thm:6C-complementary-matching}
If a tree $T$ admits a 6C-labelling, then there exists another tree $H$ admitting a 6C-labelling such that $\odot_1\langle T,H \rangle $ is a 6C-complementary matching.
\end{thm}
\begin{proof}
Suppose that a $(p,q)$-tree $T$ admits a 6C-labelling $f$, we define another labelling $h$ of $T$ by $h(x)=2(p+q+1)-f(x)$ for $x\in V(T)$, so $h(uv)=(p+q+1)-f(uv)$ for $uv\in E(T)$. Next, we set $E(T)=\{u_iv_i:i\in [1,q]\}$ holding $h(u_jv_j)<u_{j+1}v_{j+1}$ true with $j\in [1,q-1]$. Now, we can defined a labelling $g$ of a copy $T'$ of $T$ in the way: $g(x)=h(x)$ for $x\in V(T')=V(T)$, and $g(u_iv_i)=2p+q+1-h(u_iv_i)$ for $u_iv_i\in E(T')=E(T)$. Notice that
$${
\begin{split}
g(u_iv_i)&=2p+q+1-h(u_iv_i)\\
&=2p+q+1-[(p+q+1)-f(u_iv_i)]\\
&=f(u_iv_i)+p
\end{split}}$$ for $u_iv_i\in E(T')=E(T)$. We claim that $g$ is a 6C-labelling of  $T'$ too, which means that $\odot_1\langle T,T' \rangle $ is a 6C-complementary matching.
\end{proof}

\begin{thm}\label{thm:set-ordered-vs-6C-labelling}
A tree admits a set-ordered graceful labelling if and only if it admits a 6C-labelling.
\end{thm}
\begin{proof} Suppose that $(X,Y)$ is the bipartition of vertex set of a tree $T$, where $X=\{x_{i}:i\in [1,s]\}$ and $Y=\{y_{j}:j\in [1,t]\}$ with vertex number $|V(T)|=p=s+t$ and edge number $|E(T)|=p-1$.

\emph{The proof of ``if''}. Notice that $T$ admitting a set-ordered graceful labelling $g$, so each vertex is labeled as $g(x_i)=i-1$ for $i\in [1,s]$ and $g(y_j)=s+j-1$ for $j\in [1,t]$, and furthermore each edge $x_iy_j\in E(T)$ has its label $g(x_iy_j)=g(y_j)-g(x_i)=s+j-i$.

We define another labelling $f$ for the tree $T$ in the way: $f(w)=p+g(w)$ for $w\in V(T)$, and $$f(x_iy_j)=p-g(x_iy_j)=p-|g(x_i)-g(y_j)|$$ for $x_iy_j\in E(T)$. Clearly,
\begin{equation}\label{eqa:f(V(T))-f(E(T))}
f(V(T))=[p,2p-1],\quad f(E(T))=[1,p-1].
\end{equation}

(i) (e-magic) Each edge $x_iy_j\in E(T)$ holds $f(x_iy_j)+|f(x_i)-f(y_j)|=p-g(x_iy_j)+g(x_iy_j)=p$ true.

(ii) (ee-difference) Each edge $x_iy_j\in E(T)$ matches with another edge $x'_iy'_j\in E(T)$ holding $p-g(x_iy_j)=g(x'_iy'_j)$ such that

$${
\begin{split}
f(x_iy_j)&=p-g(x_iy_j)=g(x'_iy'_j)\\
&=|g(x'_i)-g(y'_j)|\\
&=|p+g(x'_i)-[p+g(y'_j)]|\\
&=|f(x'_i)-f(y'_j)|.
\end{split}}
$$

(iii) (ee-balanced) Let $s(x_iy_j)=|f(x_i)-f(y_j)|-f(x_iy_j)$ for $x_iy_j\in E(T)$, so
$${
\begin{split}
s(x_iy_j)&=|f(x_i)-f(y_j)|-f(x_iy_j)\\
&=|g(x'_i)-g(y'_j)|-p+g(x_iy_j)\\
&=2g(x_iy_j)-p,
\end{split}}
$$ which distributes $\{2-p,4-p, \dots ,-2,0,2,4,\dots ,p-2\}$ if $p$ is even, or $\{2-p,4-p, \dots ,-3,-1,1,3,\dots ,p-2\}$ if $p$ is odd. Thereby, each edge $x_iy_j\in E(T)$ matches with another edge $x''_iy''_j\in E(T)$ such that $s(x_iy_j)+s(x''_iy''_j)=0$, except that edge $e$ golding $s(e)=0$ as $p$ is even.

(iv) (EV-ordered) $f_{\max}(E(T))<f_{\min}(V(T))$ from (\ref{eqa:f(V(T))-f(E(T))}).

(v) (ve-matching) The form (\ref{eqa:f(V(T))-f(E(T))}) tells us: Each edge $uv$ matches with one vertex $w$ such that $f(uv)+f(w)=2p$, and each vertex $z$ matches with one edge $xy$ such that $f(z)+f(xy)=2p$, except the \emph{singularity} $f(w')=p$.

(vi) (set-ordered) $f_{\max}(X)<f_{\min}(Y)$ for the bipartition $(X,Y)$ of $V(G)$ according to (\ref{eqa:f(V(T))-f(E(T))}).

Hence, we claim that the labelling $f$ admits really a 6C-labelling defined in Definition \ref{defn:6C-labelling}.

\emph{The proof of ``only if''}. Suppose that $T$ admits a 6C-labelling $h$. By the property (iv) and $h(V(T)\cup E(T))=[1,2p-1]$, we get $h(E(T))=[1,p-1]$ and $h(V(T))=[p,2p-1]$. We define a labelling $h^*$ as: $h^*(w)=h(w)-p$ for $w\in V(T)$, which gives $h^*(V(T))=[0,p-1]$; and $h^*(x_iy_j)=p-h(x_iy_j)$ for each edge $x_iy_j\in E(T)$, so $h^*(E(T))=[1,p-1]$. The property (i) enables us to compute

\begin{equation}\label{eqa:c3xxxxx}
{
\begin{split}
h^*(x_iy_j)&=p-h(x_iy_j)\\
&=p-[p-|h(x_i)-h(y_j)|]\\
&=|[h(x_i)-p]-[h(y_j)-p]|\\
&=|h^*(x_i)-h^*(y_j)|,
\end{split}}
\end{equation}
that is, $h^*$ is graceful. The property (vi) means that $h^*$ is set-ordered.
\end{proof}

\begin{figure}[h]
\centering
\includegraphics[height=5.8cm]{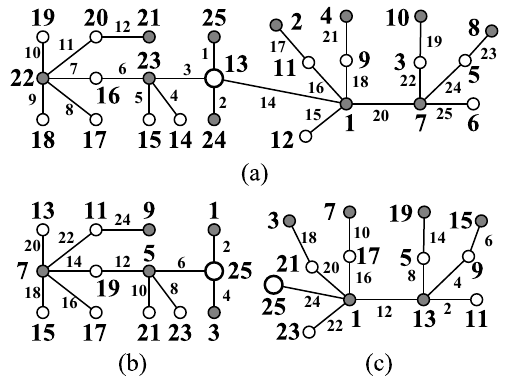}\\
\caption{\label{fig:odd-even-separable} {\small (a) $G=\odot_1\langle T,G_3\rangle $ admits a 6C-labelling defined in Definition \ref{defn:6C-labelling}, where $T$ and $G_3$ are shown in Fig.\ref{fig:ev-set-exchanged}; (b) an \emph{odd-even separable 6C-labelling}; (c) an odd-even separable 6C-labelling.}}
\end{figure}

In Fig.\ref{fig:odd-even-separable}(a), $G=\odot_1\langle T,G_3\rangle $ is obtained by identifying two singularities $13$ of $T$ and $G_3$ shown in Fig.\ref{fig:ev-set-exchanged} into one, where the 6C-labelling $f_{G_3}$ of $G_3$ is the \emph{reciprocal-inverse labelling} of the 6C-labelling $f_T$ of $T$, so we say $f_{G_3}$ and $f_T$ are matching to each other, and $(f_{G_3},f_T)$  is a 6C-complementary matching. Observe the 6C-labelling $\theta$ of $G=\odot_1\langle T,G_3\rangle $, we can see such properties: $\theta(E(G))\subset \theta(V(G))$; $13~(=p)$ is the common singularity of two trees $T$ and $G_3$; and $\theta(uv)+|\theta(u)-\theta(v)|=13~(=p)$ for each edge $uv\in E(T)$, $\theta(xy)-|\theta(x)-\theta(y)|=13~(=p)$ for each edge $xy\in E(G_3)$. The particular properties of the 6C-labelling $\theta$ of $G=\odot_1\langle T,G_3\rangle $ enables us to define a new labelling. Fig.\ref{fig:odd-even-separable}(b) and Fig.\ref{fig:odd-even-separable}(c) show two \emph{odd-even separable 6C-labellings}. Thereby, we can have the following results (the proofs of these two results are similar to that in the proof of Theorem \ref{thm:set-ordered-vs-6C-labelling}):

\begin{cor}\label{thm:odd-even-separated-labelling}
A tree admits a set-ordered graceful labelling if and only if it admits an odd-even separable 6C-labelling defined in Definition \ref{defn:6C-labelling}.
\end{cor}

\begin{cor}\label{thm:reciprocal-inverse-in-one}
Suppose that two trees $T$ and $H$ of $p$ vertices admit set-ordered graceful labellings. Then $G=\odot_1\langle T,H\rangle $ admits a 6C-labelling $\theta$ with
$$
\theta(uv)+|\theta(u)-\theta(v)|=p
$$
for each edge $uv\in E(T)$,
$$
\theta(xy)-|\theta(x)-\theta(y)|=p
$$ for each edge $xy\in E(H)$.
\end{cor}

Similarly with Definitions \ref{defn:Oemt-labelling} and \ref{defn:relaxed-Oemt-labelling}, we can define a graceful-magic matching labelling (e.g. the edge-magic total labelling) and an ee-difference graceful-magic matching labelling $(f,g)$ (see Fig.\ref{fig:odd-graceful-2-magic}(c) and (d)).

\begin{defn}\label{defn:Dgemm-labelling}
$^*$ Suppose that a $(p,q)$-graph $G$ admits a vertex labelling $f:V(G)\rightarrow [0,p-1]$ and an edge labelling $g:E(G)\rightarrow [1,q]$, and let $s(uv)=|f(u)-f(v)|-g(uv)$ for $uv\in E(G)$. If

(i) each edge $uv$ corresponds an edge $u'v'$ such that $g(uv)=|f(u')-f(v')|$ (or $g(uv)=p-|f(u')-f(v')|$);

(ii) and there exists a constant $k''$ such that each edge $xy$ has a matching edge $x'y'$ holding $s(xy)+s(x'y')=k''$ true;

(iii) there exists a constant $k$ such that $|f(u)-f(v)|+f(uv)=k$ for each edge $uv\in E(G)$;

(iv) there exists a constant $k'$ such that each edge $uv$ matches with one vertex $w$ such that $f(uv)+f(w)=k'$, and each vertex $z$ matches with one edge $xy$ such that $f(z)+f(xy)=k'$, except the \emph{singularity} $f(x_0)=0$.

Then we call $(f,g)$ an \emph{ee-difference graceful-magic matching labelling} (Dgemm-labelling) of $G$ (called a \emph{Dgemm-graph}). (see examples shown in Fig.\ref{fig:odd-graceful-2-magic}(c) and (d))\qqed
\end{defn}

\begin{figure}[h]
\centering
\includegraphics[height=6.2cm]{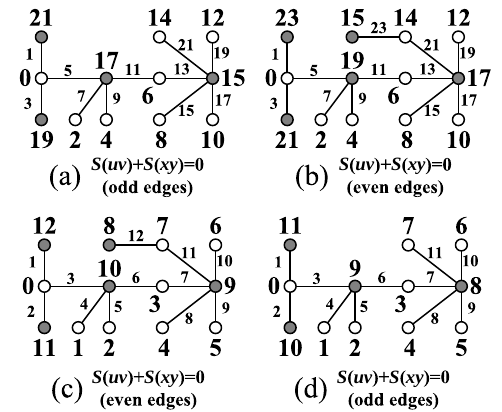}\\
\caption{\label{fig:odd-graceful-2-magic} {\small (a) An ee-difference odd-magic matching labelling; (b) an ee-difference odd-magic matching labelling; (c) an ee-difference graceful-magic matching labelling; (d) an ee-difference graceful-magic matching labelling.}}
\end{figure}

\subsection{Inverse matchings}

\begin{defn}\label{defn:edge-magic-graceful-labelling}
\cite{Marumuthu-G-2015} If there exists a constant $k\geq 0$, such that a $(p, q)$-graph $G$ admits a total labelling $f:V(G)\cup E(G)\rightarrow [1, p+q]$, each edge $uv\in E(G)$ holds
$$|f(u)+f(v)-f(uv)|=k$$ and $f(V(G)\cup E(G))=[1, p+q]$ true, we call $f$ an \emph{edge-magic graceful labelling} of $G$, and $k$ a \emph{magic constant}. Moreover, $f$ is called a \emph{super edge-magic graceful labelling} if $f(V(G))=[1, p]$.\qqed
\end{defn}

\begin{defn}\label{defn:ve-exchanged-labelling}
\cite{Yao-Mu-Sun-Zhang-Wang-Su-Ma-2018} A \emph{ve-exchanged matching labelling} $h$ of an edge-magic graceful labelling $f$ of a $(p, q)$-graph $G$ is defined as: $h:V(G)\cup E(G)\rightarrow [1, p+q]$, each edge $uv\in E(G)$ holds $h(uv)=|h(u)-h(v)|$ true (or $h(uv)+|h(u)-h(v)|=k$, or $h(uv)=h(u)+h(v)~(\bmod~q)$, or $|h(u)+h(v)-h(uv)|=k$, or $h(u)+h(uv)+h(v)=k$), such that $h(V(G)\cup E(G))=[1, p+q]$, $h(V(G))\setminus \{a_0\}=f(E(G))$ and $h(E(G))=f(V(G))\setminus \{a_0\}$, where $a_0=\lfloor \frac{p+q+1}{2}\rfloor $ is the singularity of two labellings $f$ and $h$. (see Fig.\ref{fig:some-examples}(b) and (c)).\qqed
\end{defn}

\begin{figure}[h]
\centering
\includegraphics[height=2.4cm]{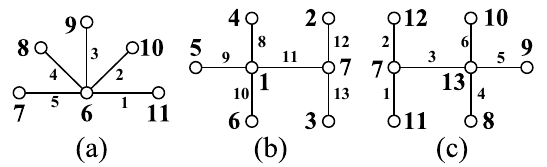}\\
\caption{\label{fig:some-examples} {\small (a) A relaxed Emt-graph: (b) an edge-magic graceful labelling $f$; (c) a ve-exchanged matching labelling of $f$ with the singularity 7.}}
\end{figure}

By Definition \ref{defn:ve-exchanged-labelling}, we propose the concept of ``\emph{reciprocal-inverse matching labelling}'': Suppose that a $(p,q)$-graph $G$ admits an edge-magic graceful labelling $f$, and a $(q,p)$-graph $H$ admits an edge-magic graceful labelling $g$. If
$$
f(E(G))=g(V(H))\setminus X^*,~f(V(G))\setminus X^*=g(E(H))
$$ for $X^*=f(V(G))\cap g(V(H))$, we say that $f$ and $g$ are \emph{reciprocal-inverse} (or \emph{reciprocal complementary}), moreover $H$ (or $G$) is an \emph{inverse matching} of $G$ (or $H$).

Observe Fig.\ref{fig:some-examples}(b) and (c), we have the \emph{total-magic matching labelling} $(f,f')$ of a $(p,q)$-graph $G$ defined as: Two total labellings $f:V(G)\cup E(G)\rightarrow [1,p+q]$ and  $f':V(G)\cup E(G)\rightarrow [1,p+q]$ such that $$f(x)+f'(x)=f(uv)+f'(uv)$$ for any vertex $x\in V(G)$ and edge $uv\in E(G)$. We can ask for $f$ (or $f'$) being an edge-magic total labelling, or an edge-magic graceful labelling, or other labellings defined on $V(G)\cup E(G)$ and $[1,p+q]$.

\subsection{Self-matchings}

For a partition $K_n=G\cup \overline{G}$ with the same vertex set $V(K_n)=V(G)=V(\overline{G})$ and edge-disjoint sets $E(G)\cap E(\overline{G})=\emptyset $, we say that $G$ and $\overline{G}$ are \emph{complementary} to each other, and we say $G$ is \emph{self-complementary} if $G$ is isomorphic to $\overline{G}$. So, we can consider this case as a \emph{self-matching}. Motivated from complete graph $K_n=G\cup \overline{G}$, we propose
\begin{defn}\label{defn:self-matchings}
$^*$ Let $W$ be a universal graph, and two graphs $G,H$ of $W$ hold $E(G)\cap E(H)=\emptyset $ and $V(W)=V(G)\cup V(H)$ true. If $E(W)=E(G)\cup E(H)$, we say $G$ and $H$ to be \emph{$W$-complementary} to each other, and moreover we call $G$ to be \emph{self-matching} (also, self-complementary) if $G$ is isomorphic to $H$, that is, $G\cong H$.\qqed
\end{defn}

For example, if $T$ is isomorphic to $T'$ in a universal graph $\odot\langle T,T'\rangle $, then we say $T$ a \emph{self-matching}.

If a connected $(p,q)$-graph $G$ admits an edge-magic total labelling $f$, then there exists a connected graph $H$ admitting a ve-matching labelling $g$ such that $f(E(G))=g(V(H))$ and $f(V(G))=g(E(H))$ for $G$ and $H$ being not trees. We prove this proposition as: We take $H$ as a copy of $G$, and define the \emph{dual labelling} of $f$ for $H$ as: $g(x)=\max f(V(G)\cup E(G))+\min f(V(G)\cup E(G))-f(x)$ for $x\in V(G)\cup E(G)$. So, $G$ is self-matching with $H$ under the ve-matching labelling.

If a connected $(p,q)$-graph $G$ admits an odd-graceful labelling $f:V(G)\rightarrow [0,2q-1]$, then there exists another connected graph $H$ admitting a pan-odd-graceful labelling $g:V(H)\rightarrow [1,2q]$ such that $\odot\langle G,H\rangle $ admits a \emph{twin odd-graceful labelling}. For showing this claim, we let $H\cong G$, and define $g(x)=f(x)+1$ for $x\in V(H)=V(G)$, and then identify the vertices of $G$ and $H$ with the same labels into one for obtaining $\odot\langle G,H\rangle $. Here, $G$ is a self-matching.

If a connected $(p,q)$-graph $G$ admits a 6C-labelling $f$, $G$ has its own reciprocal complementary $H$ admitting a 6C-labelling $g$, and $f$ and $g$ are pairwise reciprocal-inverse labellings, so $\odot\langle G,H\rangle $ is a \emph{self-matching} when $G\cong H$ (see for examples $T$ and $G_1$ shown in Fig.\ref{fig:ev-set-exchanged}).

\begin{cor}\label{thm:self-matching-6C-labelling}
If a tree $T$ admits a set-ordered graceful labelling, then we have a self-matching $\odot_1\langle T,T\rangle $ admitting a 6C-labelling.
\end{cor}

\begin{cor}\label{thm:self-matching-twin-odd-graceful}
If a tree $T$ admits a set-ordered odd-graceful/odd-elegant labelling, then there exists a self-matching $\odot_1\langle T,T\rangle $ admitting a twin odd-graceful/odd-elegant labelling.
\end{cor}
\begin{proof} By the hypothesis of the corollary, a tree $T$ has its own vertex bipartition $(X,Y)$ with $X=\{x_{i}:i\in [1,s]\}$ and $Y=\{y_{j}:j\in [1,t]\}$ with vertex number $|V(T)|=p=s+t$ and edge number $|E(T)|=p-1$. Since $T$ admits a set-ordered graceful labelling $f$, so we get $f(x_i)=i-1$ for $i\in [1,s]$ and $f(y_j)=s+j-1$ for $j\in [1,t]$, and $f(x_iy_j)=f(y_j)-f(x_i)=s+j-i$ for each edge $x_iy_j\in E(T)$. Clearly, $f(X)<f(Y)$.

(1) We define a labelling $g$ of a copy $T'$ of $T$ with $(X',Y')=(X,Y)$ as: $g(x'_i)=2f(x_i)=2(i-1)$ for $i\in [1,s]$ and $g(y'_j)=2f(y_j)-1=2(s+j)-3$ for $j\in [1,t]$, immediately,

$${
\begin{split}
g(x'_iy'_j)&=|g(y'_j)-g(x'_i)|\\
&=|2(s+j)-3-2(i-1)|\\
&=2s+2(j-i)-1\\
&=2(s+j-i)-1\\
&=2f(x_iy_j)-1.
\end{split}}$$

So, $g$ is an odd-graceful labelling of $T'$, since $g(X')=\{0,2,\dots, 2(s-1)\}$ is an even-set, $g(Y')=\{2s-1,2s+1,\dots, 2p-3\}$ is an odd-set, and $g(E(T))=[1,2p-3]^o$ is an odd-set too. Now, we take another copy $T''$ of $T$ with $(X'',Y'')=(X,Y)$, and make a \emph{complementary labelling} $g''$ of the odd-graceful labelling $g'$ by setting $g''(w)=g'(w)+1$ for $w\in V(T)$, clearly, $g''(E(T))=g'(E(T))$. Moreover, $g''(X)=[1, 2s-1]^o$, $g''(Y)=\{2s,2s+2,\dots, 2p-2\}$, we can see
$$g''(V(T))\cap g'(V(T))=\{2s-1\}$$
and
$$g''(V(T))\cup g'(V(T))=[0,2p-2].$$ $T''$ is the \emph{complementary matching} of $T'$. Thereby, $\odot_1\langle T',T''\rangle $ admits a twin odd-graceful labelling and it is a self-matching.

(2) The proof of $\odot_1\langle T',T''\rangle $ admitting a twin odd-elegant labelling is very similar to that of the above (1), here, it takes ``$\bmod~2p-2$''.

The proof of this corollary is finished.
\end{proof}

\subsection{Set-ordered matchings}

Suppose that a $(p,q)$-graph $G$ admits an $\varepsilon$-labelling $f: V(G)\rightarrow [0,p-1]$~(or $[1,p+q]$, or $[0,2q-1]$, or $[1,2q]$), $G$ is bipartite with its own bipartition $(X,Y)$. The symbol $f_{\max}(X)<f_{\min}(Y)$ is defined by $$\max\{f(x):~x\in X\}< \min\{f(y):~y\in Y\}$$ and we call $f$ a \emph{set-ordered $\varepsilon$-labelling} of $G$. As known, many set-ordered $\varepsilon$-labellings have good properties, and have been connected with other labellins equivalently (\cite{SH-ZXH-YB-2017-2, SH-ZXH-YB-2017-3,SH-ZXH-YB-2017-4, SH-ZXH-YB-2018-5, Yao-Liu-Yao-2017}). However, determining a graph whether admits a set-ordered $\varepsilon$-labelling seems to be not easy.

\begin{thm}\label{thm:rational-set-ordered}
Suppose that a bipartite $(p,q)$-graph $G$ admits an $\varepsilon$-labelling, then there exists another bipartite $(p,q)$-graph $H$ such that a bipartite graph $G\ominus H$ obtained by using an edge to join a vertex of $G$ with a vertex of $H$ admits a set-ordered $\varepsilon$-labelling.
\end{thm}
\begin{proof} Let $(X,Y)$ be the bipartition of vertices of $G$ such that each edge $uv\in E(G)$ satisfies $u\in X$ and $v\in Y$. We take a copy of $G$, denoted as $H$ with its bipartition $(X',Y')=(X,Y)$. Suppose that $G$ admits an $\varepsilon$-labelling $f$, so $H$ admits an $\varepsilon$-labelling $f'$ which is a copy of $f$. Now, we use an edge to join any vertex $u$ of $G$ with its isomorphic vertex $u'$ of $H$ for producing the desired graph $G\ominus H$. Clearly, $G\ominus H$ is a bipartite $(2p,2q+1)$-graph with bipartition $(X\cup Y',X'\cup Y)$. We define a labelling $g$ as: $g(x)=f(x)$ for $x\in X$, $g(y')=f(y')$ for $y'\in Y'$, $g(w)=f(w)+p$ for $w\in X'$, and $g(z)=f(z)+p$ for $z\in Y$. Obviously,
$$g_{\max}(X\cup Y')<g_{\min}(X'\cup Y),$$ so $g$ is a set-ordered $\varepsilon$-labelling of $G\ominus H$.
\end{proof}

If the graph $G\ominus H$ based on two disjoint graphs $G$ and $H$ admits a set-ordered $\varepsilon$-labelling, we say $G$ to be a \emph{set-ordered matching} of $H$, and vice versa. It may be interesting to look for $G\not\cong H$ in the set-ordered matching $G\ominus H$.

\subsection{Labellings with extremal conditions}

\begin{defn}\label{defn:difference-sum-labelling}
$^*$ Let $f:V(G)\rightarrow [0,q]$ be a vertex labelling of a $(p,q)$-graph $G$, and let
$$S_{um}(G,f)=\sum_{uv\in E(G)}|f(u)-f(v)|,$$ we call $f$ a \emph{difference-sum labelling}. Find two extremum $\max_f S_{um}(G,f)$ (profit) and $\min_f S_{um}(G,f)$ (cost) over all difference-sum labellings of $G$.\qqed
\end{defn}

A tree $T$ shown in Fig.\ref{fig:extremum-1} has $\max_f S_{um}(T,f)=S_{um}(T,h_1)=84$ and $\min_f S_{um}(T,f)=S_{um}(T,h_4)=20$. We will show some properties of difference-sum labellings of graphs with necessary proofs.

\begin{asparaenum}[(\textbf{Extr}-1) ]
\item Each complete graph $K_n$ holds
$${
\begin{split}
&\quad \min_g S_{um}(K_n,g)=\max_f S_{um}(K_n,f)\\
&=\sum^{n-1}_{i=1}\sum^{n-i}_{k=1}k=\frac{1}{2}\sum^{n-1}_{i=1}(n-i+1)(n-i).
\end{split}}$$

\item Let $G$ be a caterpillar, then

\quad (a) Adding a leaf to $G$ produces another caterpillar $G+e$, we have
$${
\begin{split}
\min_f S_{um}(G,f)&\leq \min_gS_{um}(G+e,g)\\
&\leq \min_f S_{um}(G,f)+3.
\end{split}}$$

\quad (b) Suppose $H$ is another caterpillar with $|V(G)|=|V(H)|$. If two diameters $D(H)>D(G)$, then
$$\min_g S_{um}(H,g)\leq \min_f S_{um}(G,f).$$

\item If a disconnected graph $G$ has its components $G_1,\dots ,G_m$, then $$\max_fS_{um}(G,f)=\sum^m_{i=1}\max_gS_{um}(G_i,g),$$
 $$\min_fS_{um}(G,f)=\sum^m_{i=1}\min_gS_{um}(G_i,g).$$
\item Adding an edge $uv$ to join two non-adjacent vertices $u,v$ of $G$ produces a new graph $G+uv$, then $\max_f S_{um}(G,f)\leq \max_h S_{um}(G+uv,h)$.
\item If $f$ and $f^*$ are dual difference-sum labellings to each other, then $S_{um}(G,f)=S_{um}(G,f^*)$.

\item For a tree $T$ of $p$ vertices, a \emph{path} $P_{p}$ of $p$ vertices and a \emph{star} $K_{1,p-1}$, we have
\begin{equation}\label{eqa:c3xxxxx}
{
\begin{split}
p-1=&\min_f S_{um}(P_{p},f)\leq \min_f S_{um}(T,f)\\
&\leq \min_f S_{um}(K_{1,p-1},f)=\frac{p^2}{4},\frac{p^2-1}{4}.
\end{split}}
\end{equation}

\item A difference-sum labelling $h$ of a tree $T$ holds $\max_f S_{um}(T,f)=S_{um}(T,h)$ true if and only if $ h_{\max}(X)< h_{\min}(Y)$ with the partition $(X,Y)$ of $T$.

\begin{proof} Let $V(T)=X\cup Y$ with $X\cap Y=\emptyset$, where $X=\{x_1,x_2,\dots ,x_a\}$, $Y=\{y_1,y_2,\dots ,y_b\}$. Suppose $G$ admits a vertex labelling $f:V(T)\rightarrow [0,p-1]$ such that
\begin{equation}\label{eqa:set-ordered}
{
\begin{split}
0=&f(x_1)<f(x_2)<\cdots f(x_i)<f(x_{i+1})<\cdots \\
&<f(x_a)<f(y_1)<f(y_2)<\cdots f(y_j)<\\
&f(y_{j+1})<\cdots <f(y_b)=p-1
\end{split}}
\end{equation}
We exchange two labels $f(x_i)$ and $f(y_j)$ for some $i\neq j$. In other word, we define another vertex labelling $g:V(T)\rightarrow [0,p-1]$ such that $g(x_i)=f(y_j)$, $g(y_j)=f(x_i)$ and $g(w)=f(w)$ for $w\in V(T)\setminus \{x_i,y_j\}$. Clearly, there is no $ g_{\max}(X)< g_{\min}(Y)$. We claim

\begin{equation}\label{eqa:assertion}
S_{um}(T,f)>S_{um}(T,g).
\end{equation}

Let the set of neighbors of the vertex $x_i$ is denoted as $N(x_i)=\{y_{i_1},y_{i_2},\dots y_{i_m}\}$, the set of neighbors of the vertex $y_j$ is written as $N(y_j)=\{x_{j_1},x_{j_2},\dots x_{j_n}\}$. We compute
\begin{equation}\label{eqa:compute1}
{
\begin{split}
&|g(x_i)-g(y_{i_s})|=|f(y_j)-f(y_{i_s})|\leq f(y_{i_s})-\\
&f(y_1)<f(y_{i_s})-f(x_a)\leq f(y_{i_s})-f(x_i),
\end{split}}
\end{equation}
and
\begin{equation}\label{eqa:compute2}
{
\begin{split}
&|g(y_j)-g(x_{j_t})|=|f(x_i)-f(x_{j_t})|\leq f(x_a)-\\
&f(x_{j_t})<f(y_{1})-f(x_{j_t})\leq f(y_{1})-f(x_{j_t}).
\end{split}}
\end{equation}
Thereby, our assertion (\ref{eqa:assertion}) holds true.
\end{proof}

\item \textbf{Max-min-sum Algorithm} for computing $\max_f S_{um}(T,f)$.

\quad \textbf{Initiation.} Take a labelling $f_0:V(T)\rightarrow [0,p-1]$.

\quad \textbf{Iteration.} For an optimal labelling $f_k:V(T)\rightarrow [0,p-1]$, find a pair of vertices $x,y$, and check whether
\begin{equation}\label{eqa:max-min-1}
\sum_{x_i\in N(x)}|f_k(y)-f(x_i)|\geq (\leq)\sum_{x_i\in N(x)}|f_k(x)-f(x_i)|
\end{equation}
and
\begin{equation}\label{eqa:max-min-2}
\sum_{y_j\in N(y)}|f_k(x)-f(y_j)|\geq (\leq)\sum_{y_j\in N(y)}|f_k(y)-f(y_j)|.
\end{equation}
If it is so, we define a new labelling $f_{k+1}$ as: $f_{k+1}(x):=f_k(y)$, $f_{k+1}(y):=f_k(x)$, and $f_{k+1}(w):=f_k(w)$ for $w\in V(T)\setminus \{x,y\}$.

\quad \textbf{Termination.} If no two vertices $x,y$ hold the forms (\ref{eqa:max-min-1}) and (\ref{eqa:max-min-2}) true, output the labelling $f_k$ with $S_{um}(T,f_k)=\max_fS_{um}(T,f)$.

\quad Similarly, we can use $(\leq )$ in (\ref{eqa:max-min-1}) and (\ref{eqa:max-min-2}) to deal with the case $\min_fS_{um}(T,f)$.

\item If $T$ is a caterpillar, then we can compute the exact value of $\min_f S_{um}(T,f)$.
\begin{proof} We show an algorithmic proof here. A caterpillar $T$ shown in Fig.\ref{fig:caterpillar-general} contains a path $P=u_1u_2\cdots u_n$, and each set of leaves $v_{i,j}$ adjacent to a vertex $u_i$ is denoted as $L(u_i)=\{v_{i,j}:j\in [1,m_i]\}$ with $m_i\geq 0$ and $i\in [1,n]$. So, each vertex $u_i$ has its own degree $\textrm{deg}(u_1)=|L(u_1)|+1=m_1+1$, $\textrm{deg}(u_j)=|L(u_j)|+2=m_j+2$ with $j\in [2,n-1]$, $\textrm{deg}(u_n)=|L(u_n)|+1=m_n+1$. We define a labelling $f$ of $T$ as follows. Let $N(x)$ be the set of neighbors of a vertex $x$.

\quad \textbf{Step 1}. For the vertices of $N(u_1)=L(u_1)\cup \{u_2\}$, we set $f(v_{1,j})=j-1$ with $j\in [1,m_1]$, $f(u_1)=m_1$, $f(u_2)=m_1+1$. Compute the sum
\begin{equation}\label{eqa:sum-1}
{
\begin{split}
&\quad \sum ^{m_1}_{j=1}|f(v_{1,j})-f(u_1)|=\sum ^{m_1}_{j=1}[m_1-f(v_{1,j})]\\
&=1+2+\cdots +m_1=\frac{1}{2}(m_1+1)m_1,
\end{split}}
\end{equation}
and $|f(u_1)-f(u_2)|=m_1+1-m_1=1$.

\quad \textbf{Step 2}. Notices that $u_i\in N(u_{i-1})$ with $i\in [2,n]$, so $f(u_i)$ has been defined well. For the vertices of $N(u_i)=L(u_i)\cup \{u_{i-1},u_{i+1}\}$ with $i\in [2,n-1]$, we set $f(v_{i,j})=f(u_i)+j$ with $j\in [1,m_i]$, $f(u_{i+1})=f(u_i)+m_i+1$. Thereby, we have
\begin{equation}\label{eqa:sum-2}
{
\begin{split}
&\quad \sum ^{m_i}_{j=1}|f(v_{i,j})-f(u_i)|=\sum ^{m_i}_{j=1}[f(u_i)+j-f(u_i)]\\
&=1+2+\cdots +m_i=\frac{1}{2}(m_i+1)m_i,
\end{split}}
\end{equation}
and
$$|f(u_i)-f(u_{i+1})|=f(u_i)+m_i+1-f(u_i)=m_i+1.$$

\quad \textbf{Step 3}. For the vertices of $N(u_n)=L(u_n)\cup \{u_{n-1}\}$, we set $f(v_{n,j})=f(u_n)+j$ with $j\in [1,m_n]$. We get
\begin{equation}\label{eqa:sum-3}
{
\begin{split}
&\quad \sum ^{m_n}_{j=1}|f(v_{n,j})-f(u_n)|=\sum ^{m_n}_{j=1}[f(u_n)+j-f(u_n)]\\
&=1+2+\cdots +m_n=\frac{1}{2}(m_n+1)m_n.
\end{split}}
\end{equation}
Therefore, we summarize the above sub-sums as
\begin{equation}\label{eqa:sum-4}
{
\begin{split}
S_{um}(T,f)&=\frac{1}{2}(m_1+1)m_1+1+\frac{1}{2}(m_n+1)m_n\\
&\quad \sum^{n-1}_{k=2}\frac{1}{2}(m_k+1)m_k+\sum^{n}_{k=2}(m_k+1)\\
&=n-m_1+\frac{1}{2}\sum^{n}_{k=1}m_k(m_k+3)
\end{split}}
\end{equation}

We, now, optimize the sum $S_{um}(T,f)$. According to the definition of the labelling $f$, $|f(u_{i})-f(u_{i+1})|=f(u_{i+1})-f(u_{i})=m_{i}+1$, so $|f(u_{i})-f(u_{i-1})|=f(u_{i})-f(u_{i-1})=m_{i-1}+1$. We select a vertex $v_{i,j}$ for some $j\in [1,m_i]$, and define a new labelling $g$ as: $g(u_i)=f(v_{i,j})=f(u_i)+j$, $g(v_{i,j})=f(u_i)$, and $g(w)=f(w)$ for $w\in V(T)\setminus \{v_{i,j},u_i\}$. Now, we inspect the sum $S_{um}(T,g)$. From

$${
\begin{split}
&\quad |g(u_{i})-g(u_{i-1})|=f(v_{i,j})-f(u_{i-1})\\
&=f(u_i)+j-f(u_{i-1})\\
&=m_{i-1}+1+j
\end{split}}
$$
and
$${
\begin{split}
&\quad |g(u_{i})-g(u_{i+1})|=f(u_{i+1})-f(u_i)\\
&=f(u_{i+1})-f(v_{i,j})\\
&=f(u_i)+1+m_i-[f(u_i)+j]\\
&=m_i+1-j.
\end{split}}
$$

we have
$$
{
\begin{split}
&\quad |f(u_{i})-f(u_{i-1})|+|f(u_{i})-f(u_{i+1})|\\
&=|g(u_{i})-g(u_{i-1})|+|g(u_{i})-g(u_{i+1})|.
\end{split}}$$
and
\begin{equation}\label{eqa:new-vs-old}
{
\begin{split}
&\quad \sum ^{m_i}_{k=1}|g(u_{i})-g(v_{i,k})|\\
&=|f(v_{i,j})-f(u_{i})|+\sum ^{m_i}_{k=1, k\neq j}|f(v_{i,j})-f(v_{i,k})|\\
&<\sum ^{m_i}_{k=1}|f(u_{i})-f(v_{i,k})|
\end{split}}
\end{equation}
We select $a_{i,0}=\lfloor \frac{m_i+1}{2}\rfloor $ with $i\in [1,n]$, thus, each sum
$${
\begin{split}
\sum ^{m_i}_{k=1}|g(u_{i})-g(v_{i,k})|&=\sum ^{m_i}_{k=1, k\neq a_{i,0}}|f(v_{i,a_{i,0}})-f(v_{i,k})|\\
&=2(1+2+\cdots +a_{i,0})\\
&=a_{i,0}(1+a_{i,0})
\end{split}}$$
Clearly, $$a_{i,0}(1+a_{i,0})<\sum ^{m_i}_{j=1}|f(v_{i,j})-f(u_i)|$$ from in (\ref{eqa:sum-2}). Furthermore, using (\ref{eqa:max-min-1}) and (\ref{eqa:max-min-2}) check the sum $S_{um}(T,g)$, we can see $S_{um}(T,g)=\min_fS_{um}(T,f)$ holds true. This algorithm is correct and has the complex of polynomial time.
\end{proof}
\end{asparaenum}

\begin{figure}[h]
\centering
\includegraphics[height=5.2cm]{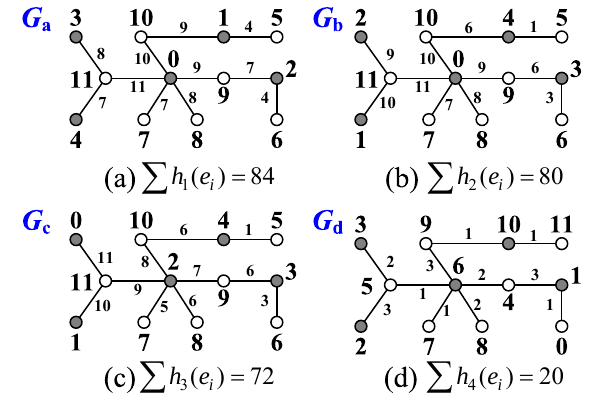}\\
\caption{\label{fig:extremum-1} {\small A lobster admits four difference-sum labellings with the maximum sum 84, sum 80, sum 72 and the minimum sum 20.}}
\end{figure}

When considering network passwords, we have theoretical guarantee for using Topsnut-gpws made by caterpillars. A spider with three legs of length 2, also, is called an aster, denoted as $A_{2,2,2}$. Or, the deletion of all leaves of $A_{2,2,2}$ results in $K_{1,3}$. If each spanning tree of a graph $G$ is a caterpillar, we call $G$ to be \emph{caterpillar-pure}. Jamison \emph{et al.} \cite{Jamison-McMorris-Mulder-2003} have shown: A connected graph is caterpillar-pure if and only if it does not contain any aster $A_{2,2,2}$ as a (not necessarily induced) subgraph.

We present a new extremal labelling, called felicitous-sum labelling (see an example in Fig.\ref{fig:extremum-felicitous}), as follow:

\begin{defn}\label{defn:felicitous-sum-labelling}
$^*$ Let $f:V(G)\rightarrow [0,q]$ be a labelling of a $(p,q)$-graph $G$, and let $$F_{um}(G,f)=\sum_{uv\in E(G)}[f(u)+f(v)]~(\bmod ~q+1),$$ we call $f$ a \emph{felicitous-sum labelling}. Find two extremum $\max_f F_{um}(G,f)$ and $\min_f F_{um}(G,f)$ over all felicitous-sum labellings of $G$. (See examples shown in Fig.\ref{fig:extremum-felicitous})\qqed
\end{defn}

\begin{figure}[h]
\centering
\includegraphics[height=5.2cm]{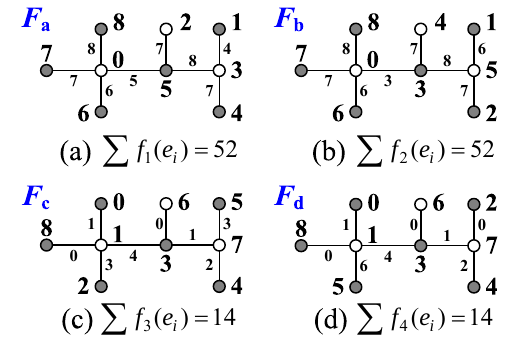}\\
\caption{\label{fig:extremum-felicitous} {\small As understanding Definition \ref{defn:felicitous-sum-labelling}, a tree admits four felicitous-sum labellings with the maximum sum 52, sum 52, the minimum sum 14 and the minimum sum 14.}}
\end{figure}

Observe Fig.\ref{fig:extremum-1} and Fig.\ref{fig:extremum-felicitous}, we can get two particular matching graphs $A=\odot_{12}\langle G_a,G_d\rangle $ with $G_a,G_d$ shown in Fig.\ref{fig:extremum-1} and $B=\odot_{9}\langle F_b,F_c\rangle $ with $F_b,F_c$ shown in Fig.\ref{fig:extremum-felicitous}. These two particular matching graphs enable us to define two new concepts in the following  Definition \ref{defn:difference-sum-matching} and Definition \ref{defn:felicitous-sum-matching}, respectively.

\begin{defn}\label{defn:difference-sum-matching}
$^*$ Suppose that $G_M$ and $G_m$ are two copies of a $(p,q)$-graph $G$, and $G_M$ admits a difference-sum labelling $f_M$ holding $S_{um}(G_M,f_M)=\max_f S_{um}(G,f)$ true, $G_m$ admits another difference-sum labelling $f_m$ holding $S_{um}(G_m,f_m)=\min_f S_{um}(G,f)$ true. The identifying graph $H=\odot\langle G_M,G_m\rangle $ is called a \emph{Max-min difference-sum matching partition}, moreover if $H=\odot_p\langle G_M,G_m\rangle $ holds $E(G_M)\cap E(G_m)=\emptyset$ true, we call $H$ a \emph{perfect Max-min difference-sum matching partition}. (See a perfect Max-min difference-sum matching partition shown in \ref{fig:one-conjecture}(a))\qqed
\end{defn}

\begin{defn}\label{defn:felicitous-sum-matching}
$^*$ Suppose that $H_M$ and $H_m$ are two copies of a $(p,q)$-graph $H$, and $H_M$ admits a felicitous-sum labelling $h_M$ holding $F_{um}(H_M,h_M)=\max_f F_{um}(G,f)$ true, $H_m$ admits another felicitous-sum labelling $h_m$ holding $F_{um}(H_m,h_m)=\min_f F_{um}(G,f)$ true. The identifying graph $G=\odot\langle H_M,H_m\rangle $ is called a \emph{Max-min felicitous-sum matching partition}, and furthermore we call $G$ a \emph{perfect Max-min felicitous-sum matching partition} if $G=\odot_p\langle H_M,H_m\rangle $ holds $E(H_M)\cap E(H_m)=\emptyset$ true. (see a perfect Max-min felicitous-sum matching partition shown in \ref{fig:one-conjecture}(b))\qqed
\end{defn}

We guess: \emph{Each tree  induces a perfect Max-min difference-sum matching partition and a perfect Max-min felicitous-sum matching partition}.

\begin{figure}[h]
\centering
\includegraphics[height=3.8cm]{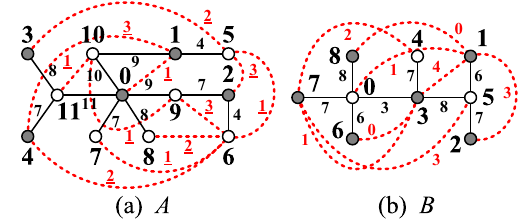}\\
\caption{\label{fig:one-conjecture} {\small (a) A perfect Max-min difference-sum matching partition $A=\odot_{12}\langle G_a,G_d\rangle $, where $G_a,G_d$ are shown in Fig.\ref{fig:extremum-1}; (b) a perfect Max-min felicitous-sum matching partition $B=\odot_{9}\langle F_b,F_c\rangle $, where $F_b,F_c$ are shown in Fig.\ref{fig:extremum-felicitous}.}}
\end{figure}

\section{Matchings from proper total colorings}

Graph coloring is an important branch of graph theory, since there are many unsolved problems and conjectures in various graph colorings. For example, the famous \emph{four color conjecture} on planar graphs was proved in 1976 by Kenneth Appel and Wolfgang Haken, but only after many false proofs and counterexamples. It was the first major theorem to be proved using a computer. A  simpler proof using the same ideas and still relying on computers was published in 1997 by Robertson, Sanders, Seymour, and Thomas. Additionally, in 2005, the theorem was proved by Georges Gonthier with general-purpose theorem-proving software (Ref. Wikipedia). Unfortunately, there is no mathematical proof of the four color theorem up to now. Another example is the \emph{total coloring conjecture} proposed by Behzad in 1965 and Vizing in 1964.

\subsection{Edge-magic and equitable proper total colorings}

Coloring each of vertices and edges of a graph $G$ with a number in $[1,k]$ makes no two adjacent vertices/edges or incident edge/vertex having the same label, we call this coloring a \emph{proper total coloring}, the minimum number of $k$ for which $G$ admits a proper total $k$-colorings is denoted as $\chi ''(G)$. Let $f:V(G)\cup E(G)\rightarrow [1,\chi ''(G)]$ be a proper total coloring of $G$ with $\chi ''(G)=|\{f(x):x\in V(G)\cup E(G)\}|$, and we call such coloring $f$ to be \emph{total chromatic number pure} (tcn-pure). Let $d_f(uv)=f(u)+f(uv)+f(v)$ for each edge $uv\in E(G)$, and set
$$B_{tol}(G,f)=\max_{uv \in E(G)}\{d_f(uv)\}-\min_{uv \in E(G)}\{d_f(uv)\}.$$
Determine a new parameter $\min_fB_{tol}(G,f)$ over all tcn-pure colorings of $G$, and call $\min_fB_{tol}(G,f)$ by the \emph{pan-bandwidth total chromatic number}. Especially, a coloring $h$ is called an \emph{edge-magic proper total coloring} if $B_{tol}(G,h)=0$, or an \emph{equitably proper total coloring} if $B_{tol}(G,h)=1$.

\vskip 0.4cm

In Fig.\ref{fig:new-total-coloring}, three complete bipartite graphs $K_{2,3}$, $K_{2,4}$ and $K_{3,3}$ admit six tcn-pure colorings: each $h_i$ is the \emph{dually total coloring} (also, \emph{matching total coloring}) of $f_i$ with $i\in [1,3]$. Moreover, $B_{tol}(K_{2,3},f_1)=2$, $B_{tol}(K_{2,4},f_2)=3$ and  $B_{tol}(K_{3,3},f_3)=4$; $B_{tol}(K_{2,3},h_1)=2$,  $B_{tol}(K_{2,4},h_2)=3$, $B_{tol}(K_{3,3},h_3)=3$.

\begin{figure}[h]
\centering
\includegraphics[height=5cm]{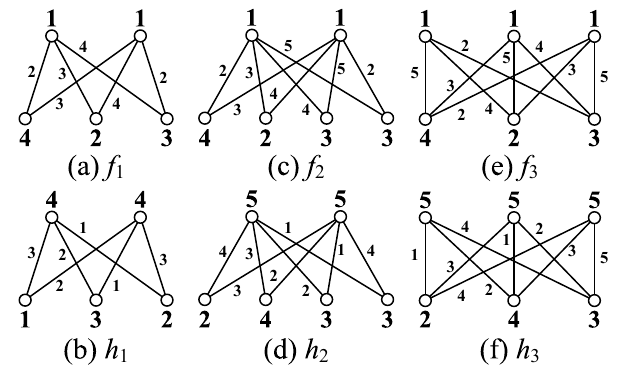}\\
\caption{\label{fig:new-total-coloring} {\small Six tcn-pure colorings of three complete bipartite graphs, in which $h_i$ is the matching coloring of the tcn-pure coloring $f_i$ with $i=1,2,3$.}}
\end{figure}

There are some obvious facts:
\begin{asparaenum}[Tot-1. ]
\item Any path $P_n$ of $n$ vertices admits an edge-magic proper total coloring, that is, $\min_fB_{tol}(P_n,f)=0$.

\item A cycle $C_n$ of $n$ vertices admits an edge-magic proper total coloring if $n\equiv 0~(\bmod~3)$, and admits an equitably total coloring, otherwise.

\item If $h$ is the \emph{matching proper total coloring} of a tcn-pure coloring $f$ defined as: $h(w)=[\max f(V(G)\cup E(G))+\min f(V(G)\cup E(G))]-f(w)$ for $w\in V(G)\cup E(G)$, then,
$${
\begin{split}
d_h(uv)&=h(u)+h(uv)+h(v)\\
&=3(\max f+\min f)-[f(u)+f(uv)+f(v)]\\
&=3(\max f+\min f)-d_f(uv),
\end{split}}$$ thus,
$${
\begin{split}
B_{tol}(G,h)&=\max_{uv \in E(G)}\{d_h(uv)\}-\min_{uv \in E(G)}\{d_h(uv)\}\\
&=\max_{uv \in E(G)}\{d_f(uv)\}-\min_{uv \in E(G)}\{d_f(uv)\}\\
&=B_{tol}(G,f).
\end{split}}$$
\end{asparaenum}

\vskip 0.4cm

It seems to be not easy to determine the exact value of the pan-bandwidth total chromatic number $\min_fB_{tol}(G,f)$ for a given simple graph $G$, since $\min_fB_{tol}(G,f)$ is related with the total chromatic number $\chi ''(G)$, which lets us recall the long standing total coloring conjecture: ``\emph{For any simple graph $G$, the elements of $V (G)\cup E(G)$ can be colored with at most $\Delta (G)+2$ colors so that no two adjacent or incident elements receive the same color, also $\chi ''(G)\leq \Delta (G)+2$.}'' Unfortunately, there are less significant results on the total labelling conjecture so far. The exact values of some special families of graphs have been obtained, such as $K_n$, complete bipartite graph $K_{m,n}$, complete $k$-partite graph $K_{m_1,m_2,\dots ,m_k}$, and join graph $G+H$ and so on. We, as exercise, verify a simple result.

\begin{lem} \label{thm:tree-total-number}
Let $T$ be a tree, then the total chromatic number $\chi\,''(T)=\Delta(T)+1$.
\end{lem}
\begin{proof} The assertion is true if $T$ is a star $K_{1,n-1}$ or a double star $S_{s,t}$, so assume $T\neq K_{1,n-1}$ and $T\neq S_{s,t}$. The following part of the proof is by induction on orders of trees and constrained by diameter $D(T)\geq 4$.

\emph{Case 1.} If there is a vertex $u$ of $T$ such that $\textrm{deg}_T(u)=2$, which is adjacent to a leaf $v$ with  $\textrm{deg}_T(v)=1$ in $T$. Let $T_1=T-v$, then $\Delta(T_1)=\Delta(T)$. We take a proper $\chi\,''(T_1)$-total coloring $f$ of $T_1$, hence, set $f\,'(v)=f(uw)$, $f\,'(uv)=f(w)$ and $f\,'(x)=f(x)$ for $x\in [E(T)\cup V(T)]\setminus N(u)$. It is clear that the coloring $f\,'$ is a proper $\chi\,''(T_1)$-total coloring of $T$ such that
$$\chi\,''(T)=\chi\,''(T_1)=\Delta(T_1)+1=\Delta(T)+1.$$

\emph{Case 2.} For the case $\textrm{deg}_T(u)\geq 3$, we take a vertex $u$ of $T$ with its neighbor set $N(u)=\{u_1,u_2,\dots ,u_k,u_0\}$, where $k=\textrm{deg}_T(u)-1$, $u_i$ is a leaf of $T$ with $i\in [1, k]$ and $\textrm{deg} (u_0)\geq 2$, and let $T_2=T-u_1$.

Now, we take a proper $\chi\,''(T_2)$-total coloring $g$ of $T_2$ having $\chi\,''(T_2)=\Delta(T_2)+1$, and then build a proper $\chi\,''(T)$-total coloring $g'$ of $T$ as: Let $C=[1,\chi\,''(T_2)]$ be the color set of $T_2$ under $g$ and $C(g,w)=\{g(xw): xw\in E(T_2)\}\cup \{g(w)\}$ for $w\in (T_2)$.

\emph{Case 2.1.} If $\Delta(T_2)=\Delta(T)$, then $|C(g,u)|\leq \Delta(T_2)$, so there is a color $\alpha \in C$ but $\alpha \notin C(g,u)$. We have that $g'(v)=g(u_0)$, $g'(uv)=\alpha $ and $g'(x)=g(x)$ for $x\in [E(T)\cup V(T)]\setminus \{v,uv\}$. Hence,
$$\chi\,''(T)=\chi\,''(T_2)=\Delta(T_2)+1=\Delta(T)+1.$$

\emph{Case 2.2.} If $\Delta(T_2)=\Delta(T)-1$, so $\chi\,''(T_2)=\Delta(T_2)+1=\Delta(T)$, and then we set $g'(v)=g(u_0)$, $g'(uv)=\chi\,''(T_2)+1$ and $g'(x)=g(x)$ for $x\in (E(T)\cup V(T))\setminus \{v,uv\}$. Therefore,
\begin{equation}\label{eqa:c3xxxxx}
\chi\,''(T)=\chi\,''(T_2)+1=\Delta(T)+1.
\end{equation}

The proof of this lemma is finished.
\end{proof}

\begin{lem} \label{thm:star-total-edge-magic}
A star also is a complete bipartite graph $K_{1,n}$. Then this star $K_{1,n}$ admits an edge-magic proper total coloring for even $n$, that is, $\min_fB_{tol}(K_{1,n},f)=0$, or an equitably proper total coloring for odd $n$.
\end{lem}
\begin{proof} Clearly, $\Delta(K_{1,n})=n$. Let this star $K_{1,n}$ have its vertex set $V(K_{1,n})=\{x_0,x_1, x_2,\dots ,x_n\}$ and edge set $E(K_{1,n})=\{x_ix_0:i\in [1,n]\}$. We define a proper total coloring $f$ of $K_{1,n}$ as: If $n$ is even, we set $f(x_0)=1$, $f(x_0x_i)=1+i$ and $f(x_i)=n+2-i$ with $i\in [1,n]$. We can compute $\chi''(K_{1,n})=\Delta(K_{1,n})+1=n+1$, and
$${
\begin{split}
d_f(u_0u_i)&=f(x_0)+f(x_0x_i)+f(x_i)\\
&=1+1+i+n+2-i\\
&=n+3,
\end{split}}
$$ which shows $B_{tol}(K_{1,n},f)=0$.

\quad If $n$ is odd, we label $f(x_0)=1$, $f(x_0x_i)=1+i$ with $i\in [1,n]$, $f(x_i)=n+2-i$ with $i\in [1,n]$ and $i\neq (n+1)/2$, and $f(x_{(n+1)/2})=n+1-(n+1)/2$. Clearly, $\chi''(K_{1,n})=\Delta(K_{1,n})+1=n+1$, $d_f(u_0u_i)=f(x_0)+f(x_0x_i)+f(x_i)=1+1+i+n-i=n+2$ with $i\in [1,n-1]$ and $i\neq (n+1)/2$, and
$${
\begin{split}
&\quad d_f(u_0u_{(n+1)/2})=f(x_0)+f(x_0x_{(n+1)/2})+f(x_{(n+1)/2})\\
&=1+[n+2-(n+1)/2]+[n+1-(n+1)/2]\\
&=n+3.
\end{split}}$$

Thereby, we claim $B_{tol}(K_{1,n},f)=1$, in other words, $K_{1,n}$ admits an equitably proper total coloring.
\end{proof}

\begin{lem} \label{thm:bi-star-total-edge-magic}
Each bi-star $S_{m,n}$ admits an equitably proper total coloring.
\end{lem}
\begin{proof} By the definition of a bi-star $S_{m,n}$, it has its own vertex set $V(S_{m,n})=\{x_0,x_1, x_2,\dots ,x_m\}\cup \{y_0,y_1, y_2,\dots ,y_n\}$ and edge set $E(S_{m,n})=\{x_ix_0:i\in [1,m]\}\cup \{y_jy_0:j\in [1,n]\}\cup \{x_0y_0\}$. Without loss of generality, assume that $m\geq n$. We consider the case $m=2a$ and $n=2b$ and set a proper total coloring $g$ of $S_{m,n}$ as follows: $g(x_0)=1$, $g(x_0x_i)=1+i$ with $i\in [1,2a]$, $g(x_0y_0)=2a+2$, $g(y_0)=2$, $g(x_i)=2a+3-i$ with $i\in [1,2a]$ and $i\neq a+1$, $g(x_{a+1})=g(x_{a+2})$; if $a\neq b$, $g(y_0y_j)=2+j$ with $j\in [1,2b]$, $g(y_j)=2a+1-j$ with $i\in [1,2b]$; if $a=b$, $g(y_0y_j)=2+i$ with $j\in [1,2b-1]$, $g(y_0y_{2b})=1$, $g(y_{2b})=2a+2$, $g(y_j)=2a+1-j$ with $j\in [1,2b-2]$, $g(y_{2b-1})=3$. It is not hard to verify $B_{tol}(S_{2a,2b},g)=1$. Notice that
$$g:V(S_{2a,2b})\rightarrow [1,2a+2]=[1,\Delta(S_{2a,2b})+1].$$

\quad For other three cases $m=2a$ and $n=2b-1$, $m=2a-1$ and $n=2b-1$, $m=2a-1$ and $n=2b$, the proof ways are similar to that of the case $m=2a$ and $n=2b$. We claim that $S_{m,n}$ admits an equitably proper total coloring.
\end{proof}

\begin{thm} \label{thm:spider-edge-magic-total-coloring}
There exist infinite trees admitting edge-magic proper total colorings.
\end{thm}
\begin{proof} A spider $S_{m_1,m_2,\dots, m_n}$ has $n$ paths (called \emph{legs}) $P_i=u_{i,1}u_{i,2}\cdots u_{i,m_i}$ with $m_i\geq 1$ for $i\in [1,n]$, and a body $u_0$ is joined with $u_{i,1}$ by an edge $u_0u_{i,1}$ with $i\in [1,n]$. As $n$ is even, the star $K_{1,n}$ having its vertex set $\{u_0,u_{1,1},u_{2,1},\dots ,u_{n,1}\}$ and edge set $\{u_0u_{i,1}:i\in [1,n]\}$ admits an edge-magic proper total coloring $f$ by Lemma \ref{thm:star-total-edge-magic}. So, $f(u_0)=1$, $f(u_0u_{i,1})=1+i$ and $f(u_{i,1})=n+2-i$ with $i\in [1,n]$. Notice that
\begin{equation}\label{eqa:spider-formula}
{
\begin{split}
d_f(u_0u_{i,1})&=f(u_0)+f(u_0u_{i,1})+f(u_{i,1})\\
&=1+n+2-i+1+i\\
&=n+3
\end{split}}
\end{equation} with $i\in [1,n]$. Clearly, $\chi''(K_{1,n})=n+1=\Delta(K_{1,n})+1$.

We expend the coloring $f$ to the spider $S_{m_1,m_2,\dots, m_n}$ by coloring a leg $P_i=u_{i,1}u_{i,2}\cdots u_{i,m_i}$ as: $f(u_{i,1}u_{i,2})=1$ and $f(u_{i,2})=f(u_0u_{i,1})=1+i$, which shows $d_f(u_{i,1}u_{i,2})$ obeys (\ref{eqa:spider-formula}); next, $f(u_{i,2}u_{i,3})=f(u_{i,1})=n+2-i$ and $f(u_{i,3})=f(u_0)=1$, so $d_f(u_{i,2}u_{i,3})$ obeys (\ref{eqa:spider-formula}) too. Go on in this way, we have shown that the spider $S_{m_1,m_2,\dots, m_n}$ admits $f$ as its a proper total coloring such that
$$\chi''(S_{m_1,m_2,\dots, m_n})=n+1=\Delta(S_{m_1,m_2,\dots, m_n})+1$$
and
$$B_{tol}(S_{m_1,m_2,\dots, m_n},f)=0.$$
The proof of the theorem is complete.
\end{proof}

An example is shown in Fig.\ref{fig:spider-1} for understanding the proof of Theorem \ref{thm:spider-edge-magic-total-coloring}. Motivated from the technique in the proof of Theorem \ref{thm:spider-edge-magic-total-coloring}, we have:

\begin{lem} \label{thm:add-leaves-edge-magic-total}
For any tree $T$, we can add a leaf $u$ to $T$ with $u\not\in V(T)$ such that the tree $T+uv$ with $v \in V(T)$ holds $\min_fB_{tol}(T,f)=\min_gB_{tol}(T+uv,g)$ true.
\end{lem}

See a generalized spider admitting an edge-magic proper total coloring shown in Fig.\ref{fig:spider-2} for understanding Lemma \ref{thm:add-leaves-edge-magic-total}, which can help us to design random Topsnut-gpws or \emph{rooted Topsnut-gpw}.

\begin{figure}[h]
\centering
\includegraphics[height=2.4cm]{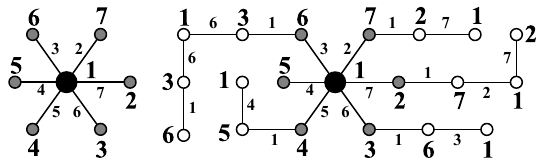}\\
\caption{\label{fig:spider-1} {\small Left is $K_{1,6}$, Right is a spider $S_{1,3,3,4,3,5}$.}}
\end{figure}

\begin{figure}[h]
\centering
\includegraphics[height=4.8cm]{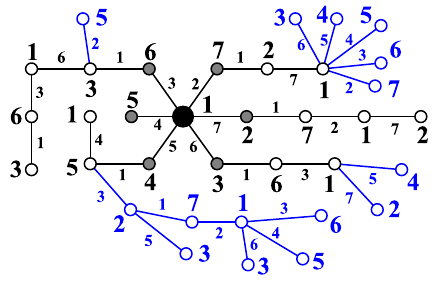}\\
\caption{\label{fig:spider-2} {\small A generalized spider $H$ admits an edge-magic proper total coloring, where $K_{1,6}$ is the root, so $H$ is called a rooted Topsnut-gpw.}}
\end{figure}

\begin{thm} \label{thm:tree-equitable-total-coloring}
Each tree admits an equitable proper total coloring.
\end{thm}
\begin{proof} We consider a tree $T$ having diameter $D_{ia}(T)\geq 4$ since $T$ with $D_{ia}(T)\leq 3$ follows Lemma \ref{thm:star-total-edge-magic} and Lemma \ref{thm:bi-star-total-edge-magic}. Select a vertex $u$ of $T$ such that the neighbor set $N(u)=\{u_1,u_2,\dots ,u_k,u_0\}$, where $k=\textrm{deg}_T(u)-1$, $\textrm{deg}_T(u_i)=1$ with $i\in [1, k]$ and $\textrm{deg} (u_0)\geq 2$. Without loss of generality, assume that $k+2\leq \Delta(T)$, since we use the induction proof on numbers of vertices of trees.

Then, $T$ admits a proper total coloring $f$ holding $B_{tol}(T,f)=\min_gB_{tol}(T,g)=1$ true. Notice that $\chi''(T)=\Delta(T)+1$, there exists a vertex $w$ of $T$ such that $\textrm{deg}_T(w)=\Delta(T)$. By the principle of induction, we add a new vertex $x$ to $T$ and join it with $u$ by an edge $ux$. There are the following cases.

\emph{Case 1.} If $f(u_0)\neq f(u_iu)$ with $i\in [1,k]$, then we set $f(ux)=f(u_0)$, $f(x)=f(u_0u)$. Obviously, $|d_f(u_0u)-d_f(ux)|=0$.

\emph{Case 2.} If $f(u_j)\neq f(u_iu)$ for some $i\neq j$, then we set $f(ux)=f(u_j)$, $f(x)=f(u_ju)$, thus, $|d_f(u_ju)-d_f(ux)|=0$.

\emph{Case 3.} Without loss of generality, we assume $f(u_0)= f(u_ju)$ for some $j$, and $f(u_i)=f(u_{i'}u)$ for each $i\in [1,k]$. We have a color set $$S=\{f(u_0u), f(u_1u), f(u_2u), \dots, f(u_ku),f(u)\}$$ such that $|S|=\Delta(T)$. We can arrange $$f(u_0u)< f(u_1u)< f(u_2u)< \dots <f(u_ku).$$

By $\chi''(T)=\Delta(T)+1$, the cardinality of the color set holds $|f(V(T)\cup E(T))|=\Delta(T)+1$ true, so there exists one color $c\in f(V(T)\cup E(T))\setminus S$. So, $f(u_ju)< c <f(u_{j+1}u)$.

\emph{Case 3.1.} If $f(u)<f(u_ju)$ or $f(u_{j+1}u)<f(u)$, immediately, we have $c-f(u_ju)=1$ and $f(u_{j+1}u)-c=1$, $f(u_{j+1}u)-f(u_ju)=2$. We set $f(ux)=c$, $f(x)=f(u_ju)$, and furthermore

\begin{equation}\label{eqa:c3xxxxx}
{
\begin{split}
&\quad |d_f(u_{j+1}u)-d_f(ux)|\\
&=|[f(u_{j+1})+f(u_{j+1}u)+f(u)]\\
&\quad -[f(x)+f(ux)+f(u)]|\\
&=|f(u_{j+1})-f(u_ju)+f(u_{j+1}u)-c|\\
&=|c+1-2-c|=1.
\end{split}}
\end{equation}

\emph{Case 3.2.} If $f(u_ju)<f(u)<c<f(u_{j+1}u)$, according to $\chi''(T)=\Delta(T)+1$, we get $f(u)-f(u_ju)=1$, $c-f(u)=1$ and $f(u_{j+1}u)-c=1$. We let $f(ux)=c$, $f(x)=f(u_{j+1}u)$, and moreover we can compute

\begin{equation}\label{eqa:c3xxxxx}
{
\begin{split}
&\quad |d_f(u_{j+1}u)-d_f(ux)|\\
&=|[f(u_{j+1})+f(u_{j+1}u)+f(u)]\\
&\quad -[f(x)+f(ux)+f(u)]|\\
&=|f(u_{j+1})+f(u_{j+1}u)-f(u_{j+1}u)-c|\\
&=|c+1-c|=1.
\end{split}}
\end{equation}
Summarizing the above cases, we have shown
$$B_{tol}(T+ux,f)=B_{tol}(T,f)=1,
$$ that is, $f$ is an equitable proper total coloring of the tree $T+ux$. The proof of this theorem is complete.
\end{proof}

According to Lemma \ref{thm:star-total-edge-magic} and Lemma \ref{thm:bi-star-total-edge-magic}, we claim that a tree $T$ admits an edge-magic proper total coloring if its maximum degree $\Delta(T)$ is even, and $T$ admits an equitable proper total coloring otherwise. Furthermore, we can make random rooted Topsnut-gpws by adding vertices and edges based on edge-magic/equitable proper total colorings. About rooted Topsnut-gpws $G$ with edge-magic proper total colorings, we have the following problems: (1) The distance between two maximum degree vertices of $G$ is at least 3. (2) For a fixed $|V(G)|$, how large is $|E(G)|$? and what is the number of maximum degree vertices of $G$? (3) The girth (length of the smallest cycle) of $G$ is at least $6$, and so on.

In  Fig.\ref{fig:rooted-total-00}, we can get a text-based password
$${
\begin{split}
D_{vev}=&5231631721721635272635235145163\\
&5217217235235217213613.
\end{split}}$$
Clearly, no way by $D_{vev}$ to reconstruct the original non-planar rooted Topsnut-gpw $G$ shown in Fig.\ref{fig:rooted-total-00}.

\begin{figure}[h]
\centering
\includegraphics[height=5.4cm]{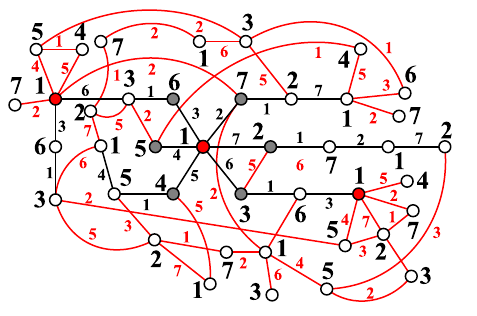}\\
\caption{\label{fig:rooted-total-00} {\small A non-planar rooted Topsnut-gpw $G$ with maximum degree $6$ and an edge-magic proper total coloring.}}
\end{figure}

\subsection{Consecutive integer sets of complete graphs with proper total colorings}

A complete graph $K_m$ admits a proper total coloring $f$ with $\chi ''(K_m)=m+1$ for even $m$, and $\chi ''(K_m)=m$ for odd $m$. Let $$f^*(E(K_m))=\{f(u)+f(uv)+f(v):uv \in E(K_m)\}.$$ Hence, we guess: $f^*(E(K_m))$ is a consecutive integer set $[a,b]$ (see Fig.\ref{fig:total-continus}). Four proper total colorings $h$ shown in Fig.\ref{fig:total-continus-dual} match with four proper total colorings $f$ shown in Fig.\ref{fig:total-continus}, we can see $f^*(E(K_3))=h^*(E(K_3))$ and $f^*(E(K_5))=h^*(E(K_5))$, so we say the proper total colorings $f$ of $K_3$ and $K_5$ are \emph{self-matching about consecutive integer sets}. Does any complete graph $K_{2n+1}$ admit a  proper total coloring $f$ being self-matching about consecutive integer set?

Since any tree $T$ admits an edge-magic proper total labelling $f$ or an equitable  proper total labelling $f$, so $$f^*(E(T))=\{f(u)+f(uv)+f(v):uv \in E(T)\}$$ is a consecutive integer set.

\begin{figure}[h]
\centering
\includegraphics[height=5.8cm]{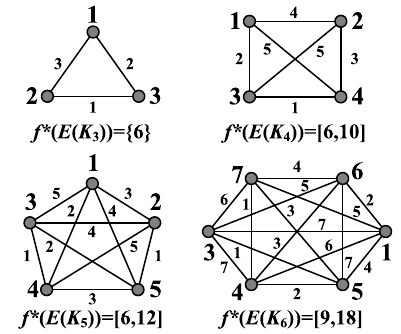}\\
\caption{\label{fig:total-continus} {\small Four examples for illustrating the consecutive integer sets of proper total colorings.}}
\end{figure}

\begin{figure}[h]
\centering
\includegraphics[height=5.8cm]{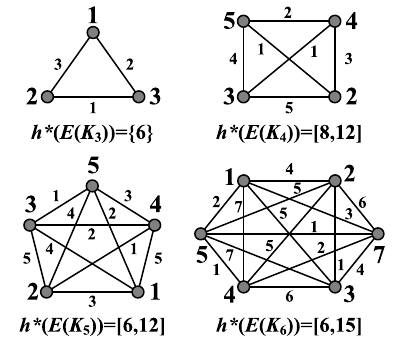}\\
\caption{\label{fig:total-continus-dual} {\small Four proper total colorings of four complete graphs match with four total colorings of four complete graphs shown in Fig.\ref{fig:total-continus}.}}
\end{figure}

In Fig.\ref{fig:new-total-coloring}, we have $f^*_1(E(K_{2,3}))=\{f_1(u)+f_1(uv)+f_1(v):uv \in E(K_{2,3})\}=[6,8]$, $f^*_2(E(K_{2,4}))=\{f_2(u)+f_2(uv)+f_2(v):uv \in E(K_{2,4})\}=[6,9]$ and $f^*_3(E(K_{3,3}))=\{f_3(u)+f_3(uv)+f_3(v):uv \in E(K_{3,3})\}=[6,10]$, and furthermore the matching colorings $h_1,h_2,h_3$ distribute us

$h^*_1(E(K_{2,3}))=\{h_1(u)+h_1(uv)+h_1(v):uv \in E(K_{2,3})\}$ $=[7,9]$,

$h^*_2(E(K_{2,4}))=\{h_2(u)+h_2(uv)+h_2(v):uv \in E(K_{2,4})\}$ $=[9,12]$ and

$h^*_3(E(K_{3,3}))=\{h_3(u)+h_3(uv)+h_3(v):uv \in E(K_{3,3})\}$ $=[8,12]$.

It shows that each complete bipartite graph $K_{m,n}$ admits a proper total coloring
$$
g:V(K_{m,n})\cup E(K_{m,n})\rightarrow [1,\chi''(K_{m,n})]
$$ such that
$$g^*(E(K_{m,n}))=\{g(u)+g(uv)+g(v):uv \in E(K_{m,n})\}$$ is a  consecutive integer set $[a,b]$.

\subsection{New parameters on proper vertex colorings}

We define two parameters on proper vertex colorings of graph theory in this subsection. Let $f:V(G)\rightarrow [1,k]$ be a proper vertex coloring of a graph $G$ with $k=\chi (G)$. We define a parameter
\begin{equation}\label{eqa:c3xxxxx}
B_{sub}(G,f)=\sum _{xy\in E(G)}|f(x)-f(y)|,
\end{equation}
and try to determine $\min_fB_{sub}(G,f)$ and $\max_fB_{sub}(G,f)$. Clearly, if $G$ is a bipartite graph, then $$\min_fB_{sub}(G,f)=\max_fB_{sub}(G,f)=|E(G)|.$$

If a connected graph $G$ holds
\begin{equation}\label{eqa:c3xxxxx}
\min_fB_{sub}(G,f)<M<\max_fB_{sub}(G,f)
\end{equation}
true for each $M$, there exists a proper vertex coloring $f_M$ of $G$ such that $B_{sub}(G,f_M)=M$, we say that $G$ has \emph{a group of consecutive difference proper vertex colorings} (see an example shown in Fig.\ref{fig:new-topic-v-coloring-1}).

\begin{figure}[h]
\centering
\includegraphics[height=2.5cm]{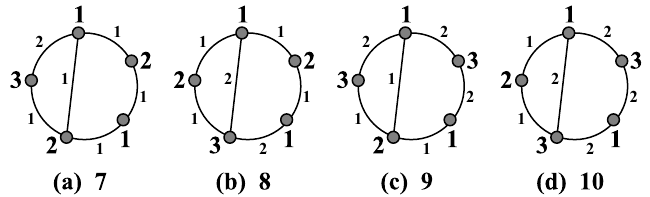}\\
\caption{\label{fig:new-topic-v-coloring-1} {\small The graph $C+uv$ admits a group of consecutive difference proper vertex colorings, which forms a \emph{Topsnut-matching chain}.}}
\end{figure}

Let $g:V(G)\rightarrow [1,k]$ be a proper vertex coloring of a graph $G$ with $k=\chi(G)$. We define another parameter \begin{equation}\label{eqa:c3xxxxx}
B_{sum}(G,g)=\sum _{xy\in E(G)}[f(x)+f(y)]
\end{equation}
and try to compute two extremal values $\min_gB_{sum}(G,g)$ and $\max_hB_{sum}(G,h)$. If $G$ has a proper vertex coloring $g_Q$ for each $Q$ satisfying
\begin{equation}\label{eqa:c3xxxxx}
\min_gB_{sum}(G,g)<Q<\max_hB_{sum}(G,h)
\end{equation}
such that $B_{sum}(G,g_Q)=Q$, we say that $G$ has \emph{a group of consecutive sum proper vertex colorings}. We give an example shown in Fig.\ref{fig:new-topic-v-coloring-2}.

Observed that the labeled graph (a) in Fig.\ref{fig:new-topic-v-coloring-1} is equal to the labeled graph (a) in Fig.\ref{fig:new-topic-v-coloring-2}, so we ask for a problem: If a proper vertex coloring $h^*$ holds $$B_{sub}(G,h^*)=\min_fB_{sub}(G,f)$$ true, then $h^*$ holds $$B_{sum}(G,h^*)=\min_fB_{sum}(G,f)$$ true too?

\begin{figure}[h]
\centering
\includegraphics[height=4.8cm]{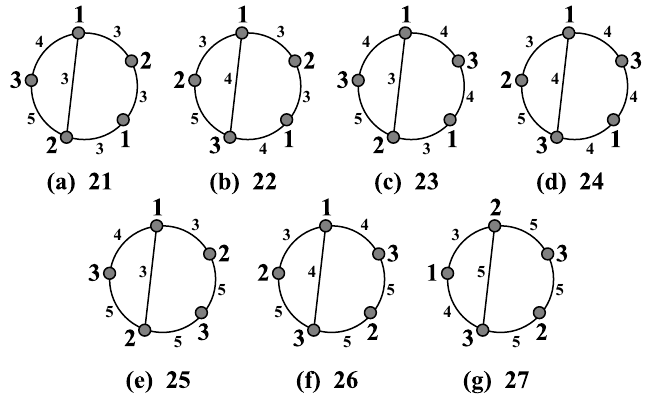}\\
\caption{\label{fig:new-topic-v-coloring-2} {\small The graph $C+uv$ admits a group of consecutive sum proper vertex colorings, which forms a \emph{Topsnut-matching chain}.}}
\end{figure}

For making Topsnut-gpws and Topsnut-matchings more complex, we can use distinguishing edge-colorings and distinguishing total colorings, since they match with the following open problems:

(i) Zhang \emph{et al.} \cite{Zhang-Liu-Wang-2002} show a famous conjecture: For every graph $G$ with no $K_2$ or $C_5$ component, then the adjacent strong edge chromatic number $$\chi'_{as}(G) \leq \Delta(G)+2.$$

(ii) Zhang \emph{et al.} \cite{Zhang-Chen-Li-Yao-Lu-Wang-China-2005} introduced a concept of adjacent vertex distinguishing total coloring (AVDTC), and show a conjecture: Let $G$ be a simple graph with order $n \geq  2$; then $G$ has its AVDTC chromatic number $$\chi''_{as}(G) \leq \Delta(G)+2.$$

(iii) Yang \emph{et al.} propose the following conjectures (\cite{Yang-Ren-Yao-Ars-2016, Yang-Yao-Ren-Information-2016}):  Every simple graph $G $ having at least a $4$-avdtc holds $$\chi''_{4as}(G) \leq \Delta(G)+4.$$

\subsection{New proper total colorings}

We use the successful experience of graph labelling definitions to propose new proper total colorings. Let $f:V(G)\cup E(G)\rightarrow [1,k]$ be a proper total coloring of a graph $G$. If each edge $uv$ holds $f(uv)=|f(u)-f(v)|$ true, we call $f$ a \emph{ve-matching difference total $k$-coloring} of $G$, we denote the smallest number of $k$ over all ve-matching difference total $k$-colorings by $\chi''_{ved}(G)$, called the \emph{ve-matching difference total chromatic number} of $G$. See two ve-matching difference total colorings shown in Fig.\ref{fig:induce-total-coloring}(a) and (b). We can see $\chi''(G)\leq \chi''_{ved}(G)$, in general.

For a proper total coloring $h:V(G)\cup E(G)\rightarrow [1,m]$, if each edge $uv$ holds $h(uv)=h(u)+h(v)$ true, we call $h$ a \emph{ve-matching sum total $m$-coloring}. The minimal number of $m$ over all ve-matching sum total $m$-colorings is denoted as $\chi''_{ves}(G)$, called the \emph{ve-matching sum total chromatic number} of $G$. Two ve-matching sum total colorings are shown in Fig.\ref{fig:induce-total-coloring}(c) and (d). Clearly, $\chi''(G)\leq \chi''_{ves}(G)$.

\begin{figure}[h]
\centering
\includegraphics[height=5.2cm]{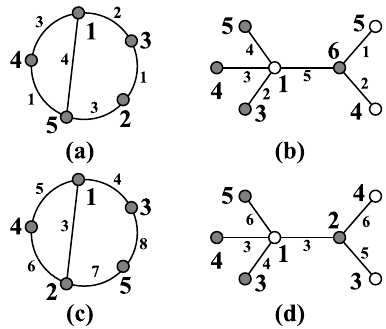}\\
\caption{\label{fig:induce-total-coloring} {\small Four ve-matching difference/sum total colorings.}}
\end{figure}

\begin{lem} \label{thm:ved-ves-trees-lemma}
Suppose that $H$ is a connected subgraph of a connected graph $G$, then
\begin{equation}\label{eqa:c3xxxxx}
\chi''_{ved}(H)\leq \chi''_{ved}(G),~\chi''_{ves}(H)\leq \chi''_{ves}(G).
\end{equation}
\end{lem}

\section{Graph labelling analysis}

Analyzing a graph labelling in detail and depth is necessary and important, since graph labellings are applied to design graphical ciphers serving to information security. We hope that analysis of graph labellings will be helpful for applying graph labellings towards information networks, and become a new subbranch of graph theory.

\subsection{Pan-labellings}

A pan-labelling is constituted of a topological structure and some operations based on numbers, letters, Topsnut-gpws, sets, groups etc. In general, a pan-labelling may be a traditional graph labelling/coloring of graph theory, or others introduced here. A coloring can be admitted by each simple graph, but a traditional labelling is admitted only by part of simple graphs.

For a $(p,q)$-graph $G$, a \emph{pan-labelling} $f$ is defined on a \emph{domain} $S\subseteq V(G)\cup E(G)$, and yields the \emph{main-range} $f(S)=\{f(x):x\in S\}$. Furthermore, $f$ gives a mapping $f'$ based on the edge-domain $E(G)$ and produces the \emph{derivative range}
$$
f'(E(G))=\{f'(uv)=F(f(u),f(v)):uv\in E(G)\}
$$ where $F(f(u),f(v))$ is a \emph{function} of two variables. The main range or the derivative range is one of number sets, graph sets, group sets, Topsnut-sets, and so on.

Definition \ref{defn:sequence-labelling} is one of generalized labelling definitions for connecting more well-defined graph labellings.

\begin{defn}\label{defn:sequence-labelling}
\cite{Yao-Mu-Sun-Zhang-Wang-Su-2018} Let $G$ be a $(p,q)$-graph, and let $A_M=\{a_i\}_1^M$ and $B_q=\{b_j\}_1^q$ be two monotonic increasing sequences of numbers with $M\geq p$. There are the following restrict conditions:
\begin{asparaenum}[Seq-1. ]
\item \label{condi:vertex-mapping} A vertex mapping $f:V(G)\rightarrow A_M$ such that $f(u)\neq f(v)$ for distinct vertices $u,v\in V(G)$.
\item \label{condi:total-mapping} A total mapping $g:V(G)\cup E(G) \rightarrow A_M\cup B_q$ such that $g(x)\neq g(y)$ for distinct elements $x,y\in V(G)\cup E(G)$.
\item \label{condi:induced-edge-label} An induced edge label $f(uv)=O(f(u),f(v))$.
\item \label{condi:total-equation} An $F$-equation $F(g(u),g(uv),g(v))=0$ holds true.
\item \label{condi:total-equation-edge} An $E$-equation $E(f(u),f(uv),f(v))=0$ holds true.
\item \label{condi:not-full-1} $f(E(G))\subseteq B_q$.
\item \label{condi:not-full-2} $g(V(G)\cup E(G))\subseteq A_M\cup B_q$.
\item \label{condi:half-full} $f(V(G))\subseteq A_M$ and $f(E(G))=B_q$.
\item \label{condi:all-full} $f(V(G))=A_M$ and $f(E(G))=B_q$.
\end{asparaenum}
\quad We call $f$:
\begin{asparaenum}[(1) ]
\item  a \emph{sequence-$(A_M,B_q)$ labelling} if Seq-\ref{condi:vertex-mapping} and Seq-\ref{condi:induced-edge-label} hold true;

\item  a \emph{sequence-$(A_M,B_q)$ total labelling} if Seq-\ref{condi:total-mapping}, Seq-\ref{condi:induced-edge-label} and Seq-\ref{condi:not-full-2} hold true;

\item  a \emph{full sequence-$(A_M,B_q)$ labelling} if Seq-\ref{condi:total-mapping}, Seq-\ref{condi:total-equation} and Seq-\ref{condi:half-full} hold true;

\item  a \emph{graceful sequence-$(A_M,B_q)$ labelling} if Seq-\ref{condi:vertex-mapping}, Seq-\ref{condi:induced-edge-label} and Seq-\ref{condi:all-full} hold true;

\item  a \emph{total sequence-$(A_M,B_q)$ labelling} if Seq-\ref{condi:total-mapping} and Seq-\ref{condi:total-equation} hold true;

\item  a \emph{sequence-$(A_M,B_q)$ $F$-total graceful labelling} if Seq-\ref{condi:total-mapping}, Seq-\ref{condi:total-equation} and Seq-\ref{condi:all-full} hold true;

\item  a \emph{sequence-$(A_M,B_q)$  mixed labelling} if Seq-\ref{condi:induced-edge-label}, Seq-\ref{condi:total-equation-edge} and Seq-\ref{condi:half-full} hold true.\qqed
\end{asparaenum}
\end{defn}

If two sets $A_M=\{a_i\}_1^M$ and $B_q=\{b_j\}_1^q$ defined in Definition \ref{defn:sequence-labelling} correspond a graph labelling admitted by graphs, we say $(A_M, B_q)$ a \emph{graph-realized sequence matching}.

\subsection{Topsnut-matchings labeled by graphs}

Let $\{G_i\}^m_1$ be a set of disjoint graphs $G_1,G_2,\dots ,G_m$, where $G_k$ admits a labelling $g_k$ with $k\in [1,m]$, and let  a $(p,q)$-graph $H$ be a \emph{base}.

\textbf{$\ast$ Addition and subtraction.} Define a mapping $f:V(H)\rightarrow \{G_i\}^m_1$, and an induced edge label $f'(uv)=F(G_i,G_j)$ for each edge $uv\in E(H)$, where $f(u)=G_i$ and $f(v)=G_j$. Hereafter, ``joining $G_i$ with $G_j$'' is defined as an operation of ``joining a vertex of a graph $G_i$ with some vertex of another graph $G_j$ by an edge''.
\begin{asparaenum}[Gr-1. ]
\item If $f'(E(H))=\{|i-j|:uv\in E(H)\}=[1,q]$, we call $H(f,f')$ a \emph{graceful graph-set labelling} (graceful gs-labelling). And, we have a \emph{graceful gs-compound} $G=\langle \{G_i\}^m_1;H(f,f')\rangle$ obtained by joining $G_i=f(u)$ with $G_j=f(v)$ for each edge $uv\in E(H)$. In general, the number of such graphs $G=\langle \{G_i\}^m_1;H(f,f')\rangle$ is not one.
\item If $f'(E(H))=\{|i-j|:uv\in E(H)\}=[1,2q-1]^o$, then $H(f,f')$ is called an \emph{odd-graceful graph-set labelling} (odd-graceful gs-labelling), and we have an \emph{odd-graceful gs-compound} $G=\langle \{G_i\}^m_1;H(f,f')\rangle$.
\item If $f'(E(H))=\{i+j~(\bmod~q):uv\in E(H)\}=[0,q-1]$, we have a \emph{felicitous graph-set labelling} (felicitous gs-labelling) $H(f,f')$, and a \emph{felicitous gs-compound} $G=\langle \{G_i\}^m_1;H(f,f')\rangle$.
\item If $f'(E(H))=\{i+j~(\bmod~2q-1):uv\in E(H)\}=[0,2q-3]^o$, we have an \emph{odd-elegant graph-set labelling} (felicitous gs-labelling) $H(f,f')$, and an \emph{odd-elegant gs-compound} $G=\langle \{G_i\}^m_1;H(f,f')\rangle$.
\end{asparaenum}

\textbf{$\ast$ Magic type.} Define a mapping $f:V(H)\cup E(H)\rightarrow \{G_i\}^m_1$, such that $f(u)=G_i$, $f(v)=G_j$ and $f(uv)=G_k$ for each edge $uv\in E(H)$.
\begin{asparaenum}[Mg-1. ]
\item If there exists a constant $k^*$, such that $$i+k+j=k^*$$ for each edge $uv\in E(H)$, we call $f$ an \emph{edge-magic total graph-set labelling} (edge-magic total gs-labelling), the graph $G=\langle \{G_i\}^m_1;H(f)\rangle$ is obtained by joining $G_i$ with $G_k$ and joining $G_k$ with $G_j$ for each edge $uv\in E(H)$ is called an \emph{edge-magic total gs-compound}.
\item If there exists a constant $k^*$ such that $|i-k+j|=k^*$ for each edge $uv\in E(H)$, we call $f$ an \emph{edge-magic total graceful graph-set labelling} (edge-magic total graceful gs-labelling), the graph $G=\langle \{G_i\}^m_1;H(f)\rangle$ is obtained by joining $G_i$ with $G_k$ and joining $G_k$ with $G_j$ for each edge $uv\in E(H)$ is called an \emph{edge-magic total graceful gs-compound}.
\item If there exists a constant $k^*$ such that $k+|i-j|=k^*$ for each edge $uv\in E(H)$, we call $f$ an \emph{edge-magic graceful total graph-set labelling} (edge-magic graceful total gs-labelling), the graph $G=\langle \{G_i\}^m_1;H(f)\rangle$ is obtained by joining $G_i$ with $G_k$ and joining $G_k$ with $G_j$ for each edge $uv\in E(H)$ is called an \emph{edge-magic graceful total gs-compound}.
\end{asparaenum}

\textbf{$\ast$ $(k,d)$ type.} It is not difficult to imitate those well-defined graph labellings having parameters, so it is interesting to mention them as exercise.

\subsection{Topsnut-matchings produced by graph operations}
\begin{asparaenum}[Op-1. ]
\item \textbf{Odd-graceful/odd-elegant graph matching.} Let $S_{og}(p)$ be a set of odd-graceful/odd-elegant graphs of $m$ vertices. A $(p,q)$-graph $G$ admits a graph labelling $f:V(G)\rightarrow S_{og}(p)$ such that each edge $uv\in E(G)$ labeled as $f(uv)=\odot\langle f(u),f(v)\rangle$ is just a twin odd-graceful/odd-elegant graph of $k_{uv}$ vertices, where $(f(u),f(v))$ is just an \emph{odd-graceful/odd-elegant Topsnut-matching}. If $\{k_{uv}:uv \in E(G)\}=[a,b]$, we say the graph $\langle G\odot S_{og}(p)\rangle$ obtained by joining $f(u)$ with $f(uv)$ and joining $f(uv)$ with $f(v)$ for each edge $uv\in E(G)$ an \emph{$[a,b]$-twin odd-graceful/odd-elegant graph}.

\item \textbf{Euler graph matching.} Let $E_{g}$ be a set of non-eulerian graphs. A $(p,q)$-graph $G$ admits a graph labelling $f:V(G)\rightarrow E_{g}$, and induced edge label $f(uv)=\odot\langle f(u),f(v)\rangle$ is just an Euler graph of $k_{uv}$ vertices, we call $(f(u),f(v))$ an \emph{Euler Topsnut-matching}. The graph $\langle G\odot E_{g}\rangle$ obtained by joining $f(u)$ with $f(uv)$ and joining $f(uv)$ with $f(v)$ for each edge $uv\in E(G)$ an \emph{$[a,b]$-Euler graph}, where $\{k_{uv}:uv \in E(G)\}=[a,b]$.

\item \textbf{Hamilton graph matching.} Let $H_{ag}$ be a set of graphs. A $(p,q)$-graph $G$ admits a graph labelling $f:V(G)\rightarrow H_{ag}$, and induced edge label $f(uv)=\odot\langle f(u),f(v)\rangle$ is just a Hamilton graph of $k_{uv}$ vertices, we call $(f(u),f(v))$ a \emph{Hamilton Topsnut-matching}. The graph $\langle G\odot E_{ag}\rangle$ obtained by joining $f(u)$ with $f(uv)$ and joining $f(uv)$ with $f(v)$ for each edge $uv\in E(G)$ is an \emph{$[a,b]$-Hamilton graph}, where $\{k_{uv}:uv \in E(G)\}=[a,b]$.
\item \textbf{Pan-matching with operation $(\bullet)$.} Let $P_{ag}$ be a set of graphs. A $(p,q)$-graph $G$ admits a graph labelling $F:V(G)\rightarrow H_{ag}$, and induced edge label $F(uv)=F(u)(\bullet)F(v)$ is just a graph having a P-matching, where $(\bullet)$ is an operation. Here, a P-matching may be: a perfect matching of $k_{uv}$ vertices, $k_{uv}$-cycle, $k_{uv}$-connected, $k_{uv}$-edge-connected, $k_{uv}$-colorable, edge $k_{uv}$-colorable, total $k_{uv}$-colorable, $k_{uv}$-regular, $k_{uv}$-girth, $k_{uv}$-maximum degree, $k_{uv}$-clique, $\{a;b\}_{uv}$-factor, a maximal planar graph of $k_{uv}$ vertices, and so on. We call the graph $\langle G(\bullet) P_{ag}\rangle$ obtained by joining $F(u)$ with $F(uv)$ and joining $F(uv)$ with $F(v)$ for each edge $uv\in E(G)$ a \emph{P-matching $\{k_{uv}\}$-graph}.
\end{asparaenum}

An $\{a;b\}$-factor is a spanning subgraph $H$ of a graph $G$ such that each vertex $x$ of $H$ has one of degree $a$ and degree $b$. Let $F:V(G)\rightarrow H_{ag}$, where $P_{ag}$ is a set of graphs, and let $F(uv)=\odot\langle F(u),F(v)\rangle$ be a graph having an $\{a;b\}$-factor for each edge $uv\in E(G)$ (see Fig.\ref{fig:graph-label-graph}).

\begin{figure}[h]
\centering
\includegraphics[height=2.6cm]{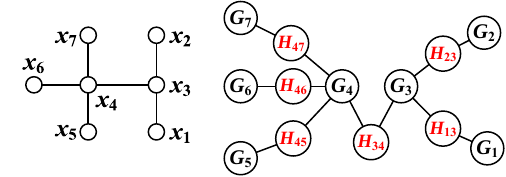}\\
\caption{\label{fig:graph-label-graph-0} {\small Left is a tree $T$, Right is a collection $T^*$ of graph-labelling graphs  based on $T$ and the intersect operation.}}
\end{figure}

\begin{figure}[h]
\centering
\includegraphics[height=5.8cm]{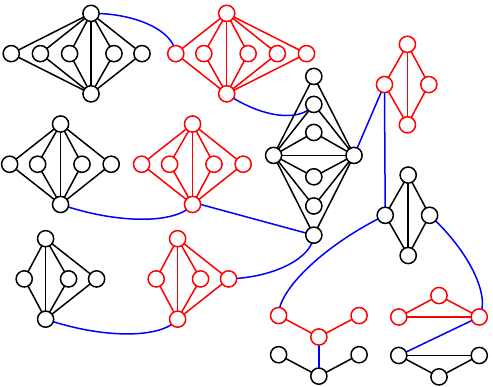}\\
\caption{\label{fig:graph-label-graph} {\small A graph $H$ is one of the collection $T^*$ shown in Fig.\ref{fig:graph-label-graph-0}.}}
\end{figure}

In Fig.\ref{fig:graph-label-graph-0} and Fig.\ref{fig:graph-label-graph}, a tree $T$ admits a graph labelling $F:V(T)\rightarrow F(G)$, where $F(x_i)=G_i$ with $i\in [1,7]$, and $F(x_ix_j)=F(x_i)\cap F(x_j)=G_i\cap G_j=H_{ij}$. The graph $H$ shown in Fig.\ref{fig:graph-label-graph} is one of the collection $T^*$ of graph-labelling graphs, since there are many ways to join two graphs by an edge. Here, $H_{13}=K_1+\overline{K}_2$ is a $\{2;1\}$-factor, $H_{23}=K_2+\overline{K}_2$ is a $\{2;2\}$-factor, $H_{34}=K_2+\overline{K}_2$ is a $\{2;3\}$-factor, $H_{45}=K_2+\overline{K}_3$ is a $\{2;4\}$-factor, $H_{46}=K_2+\overline{K}_4$ is a $\{2;5\}$-factor, and $H_{47}=K_2+\overline{K}_5$ is a $\{2;6\}$-factor. So, The edges of the graph $H$ form a  graceful sequence of $\{2;k\}$-factors with $k\in [1,6]$.

\begin{thm} \label{thm:a-b-factor-graph-labelling}
Each caterpillar with $q$ edges admits an $\{a;b\}$-factor graph labelling, where $\{a;b\}$ is a non-decreasing sequence-pair $\{a^*_i; b^*_i\}^q_1$.
\end{thm}
\begin{proof} A caterpillar $T$ shown in Fig.\ref{fig:caterpillar-general} contains a path $P=u_1u_2\cdots u_n$, and each set of leaves $v_{i,j}$ adjacent to a vertex $u_i$ is denoted as $L(u_i)=\{v_{i,j}:j\in [1,m_i]\}$ with $m_i\geq 0$ and $i\in [1,n]$. We take a non-decreasing sequence-pair $\{a_{i,j}\}$ and $\{b_{i,j}\}$ for $j\in [1,m_i]$ and $i\in [1,n]$. Each complete bipartite graph $K_{a_{i,j},b_{i,j}}$ is written as $K(a_{i,j},b_{i,j})$ for convenient statement. We define a graph labelling $F$ on $T$ as follows:

For $i:=1$, we set
$$F(u_1)=K(c_{1,2}+A(m_1),d_{1,2}+B(m_1))$$ with $A(m_1)=\sum^{m_1}_{j=1}a_{1,j}$ and $B(m_1)=\sum^{m_1}_{j=1}b_{1,j}$, and set $F(v_{1,j})=K(a_{1,j},b_{1,j})$ and $F(u_1v_{1,j})=F(u_1)\cap F(v_{1,j})$ with $j\in [1,m_1]$.

For $i:=i+1$, we set
$${
\begin{split}
F(u_{i+1})=&K\biggr (\sum^{i+1}_{k=1}A(m_k)+\sum ^{i}_{k=1}c_{k,k+1},\\
&\sum^{i+1}_{k=1}B(m_k)+\sum ^{i}_{k=1}d_{k,k+1}\biggr  )
\end{split}}$$
with $A(m_k)=\sum^{m_k}_{j=1}a_{k,j}$ and $B(m_k)=\sum^{m_k}_{j=1}b_{k,j}$, and set $F(v_{i+1,j})=K(a_{i+1,j},b_{i+1,j})$ and $F(u_1v_{i+1,j})=F(u_{i+1})\cap F(v_{i+1,j})$ with $j\in [1,m_{i+1}]$.

At the last, we let
$${
\begin{split}
F(u_{n})=&K\biggr (1+\sum^{n}_{k=1}A(m_k)+\sum ^{n-1}_{k=1}c_{k,k+1},\\
&\sum^{n}_{k=1}B(m_k)+\sum ^{n-1}_{k=1}d_{k,k+1}\biggr )
\end{split}}$$ and set $F(v_{n,j})=K(a_{n,j},b_{n,j})$ and $F(u_nv_{n,j})=F(u_{n})\cap F(v_{n,j})$ with $j\in [1,m_{n}]$.

Let $\alpha(r)=\sum^r_{k=1}m_k$. We write $a^*_{j}=a_{1,j}$ and $b^*_{1,j}=b_{1,j}$ with $j\in [1,m_1]$; $c_{12}=a^*_{1+\alpha(1)}$ and $d_{12}=b^*_{1+\alpha(1)}$; $a^*_{1+\alpha(1)+j}=a_{2,j}$ and $b^*_{1+\alpha(1)+j}=b_{2,j}$ with $j\in [1,m_2]$; $c_{2,3}=a^*_{2+\alpha(2)}$ and $d_{2,3}=b^*_{2+\alpha(2)}$; $\cdots $; $c_{n-1,n}=a^*_{n-1+\alpha(n-1)}$ and $d_{n-1,n}=b^*_{n-1+\alpha(n-1)}$; $a^*_{n-1+\alpha(n-1)+j}=a_{n,j}$ and $b^*_{n-1+\alpha(n-1)+j}=b_{n,j}$ with $j\in [1,m_n]$. Furthermore, we let $a^*_{k}\leq a^*_{k+1}$ and $b^*_{k}\leq b^*_{k+1}$ with $k\in [1, n-1+\alpha(n)]$.

Thereby, we have shown the result of the theorem.
\end{proof}

We give the values of $\{a^*_i\}^q_1$ and $\{b^*_i\}^q_1$ in Theorem \ref{thm:a-b-factor-graph-labelling} in the following:

(i) $\{a^*_i\}^q_1=\{a\}$ and $\{b^*_i\}^q_1=[1,q]$, so each caterpillar with $q$ edges admits a \emph{graceful $\{a;b\}$-factor graph labelling};

(ii) $\{a^*_i\}^q_1=\{a\}$ and $\{b^*_i\}^q_1=[1,2q-1]^o$, then every caterpillar with $q$ edges admits an \emph{odd-graceful $\{a;b\}$-factor graph labelling};

(iii) $\{a^*_i\}^q_1=[1,q]$ and $\{b^*_i\}^q_1=[1,q]$, then each caterpillar with $q$ edges admits an \emph{bi-graceful $\{a;b\}$-factor graph labelling};

(iv) $\{a^*_i\}^q_1=[1,2q-1]^o$ and $\{b^*_i\}^q_1=[1,2q-1]^o$, then each caterpillar with $q$ edges admits an \emph{bi-odd-graceful $\{a;b\}$-factor graph labelling}.

Notice that there are many graphs containing $\{a;b\}$-factors, in general. We can take well-known sequences (such as Fibonacci sequence, arithmetic progression, geometric progression, etc.) to replace $\{a^*_i\}^q_1$ and $\{b^*_i\}^q_1$ for getting more  interesting Topsnut-gpws. Moreover, it is not difficult to prove: \emph{Each lobster with $q$ edges admits an $\{a;b\}$-factor graph labelling, where $\{a,b\}$ is some non-decreasing sequence-pair $\{a^*_i; b^*_i\}^q_1$}.

\subsection{Properties of Labellings}

The previous subsections show the labellings with the following properties. Other properties of labellings can be found in \cite{Gallian2016}. Let $k$ be a constant, and let $G$ be a $(p,q)$-graph admitting a labelling $f$. We have:
\begin{asparaenum}[C-1. ]
\item (e-magic-graceful) Each edge $uv$ matches with another edge $xy$ such that $f(uv)+|f(x)-f(y)|=k$.
\item (e-magic) Each edge $uv$ matches with another edge $xy$ such that $f(x)+f(uv)+f(y)=k$.
\item (ee-graceful) Each edge $uv$ matches with another edge $xy$ holding $|f(x)+f(y)-f(uv)|=k$ true.
\item (ee-difference) Each edge $uv$ matches with another edge $xy$ holding $f(uv)=|f(x)-f(y)|$ true, or $f(uv)=M-|f(x)-f(y)|$.
\item (ee-sum) Each edge $uv$ matches with another edge $xy$ holding $f(uv)=f(x)+f(y)~(\bmod~B)$ true,  such that the resulting edge labels are distinct and nonzero.
\item (ep-matching) Each matching edge $uv\in M$ holds $f(u)+f(v)=k$ true, where $M$ is a perfect matching of $G$, and $k$ is some constant.
\item (ee-bandwiden) Each edge $uv$ matches with another edge $u'v'$ holding $s(uv)+s(u'v')=0$ true, where $s(xy)=|f(x)-f(y)|-f(xy)$.
\item (ve-matching) Each edge $uv$ matches with one vertex $w$ such that $f(uv)+f(w)=k'$, and vice versa, except the \emph{singularity}.
\item (EV-ordered) There two orders:

\quad (i) $f_{\max}(V(G))<f_{\min}(E(G))$, or $f_{\min}(V(G))>f_{\max}(E(G))$;

\quad (ii) $ f(V(G))\subseteq f(E(G))$, $ f(E(G))\subseteq f(V(G))$.

\item (set-ordered) $f_{\max}(X)<f_{\min}(Y)$, or $f_{\min}(X)>f_{\max}(Y)$ if $G$ is bipartite with its partition $(X,Y)$ of $V(G)$.
\item (reciprocal-inverse) $h(V(G)\cup E(G))=[1, p+q]$,
$$
h(V(G))\setminus \{a_0\}=f(E(G)),~h(E(G))=f(V(G))\setminus \{a_0\}
$$ where $a_0=\lfloor (p+q+1)/2\rfloor $ is the singularity of two labellings $f$ and $h$.
\item (odd-even separable) $f(V(G)\cup E(G))=[1, p+q]$, and $f(V(G))$ is an odd-set containing only odd numbers, as well as $f(E(G))$ is an even-set containing only even numbers.
\end{asparaenum}

\subsection{Some indices for analyzing graph labellings}

We design parameters for theoretically metricizing Topsnut-gpws, such as:
\begin{asparaenum}[Deg-1. ]
\item \emph{Difficulty}. A labelling $f$ holds $m$ conditions, we say $f$ to be $m$-rank difficulty.
\item \emph{Complexity.} A labelling $f$ holds $m$ conditions, each condition has a complex rank, summarizing them together forms the whole complex rank.
\item \emph{Constructibility and non-constructibility.} It includes configuration construction (with no polynomial algorithm in general) and structural construction  (with polynomial algorithm), constructive labelling. Conversely, it includes non-structural construction, non-constructive labelling.
\item \emph{Computationally unbreakable.} Consider giant spaces, no-constructive algorithms, non-mathematical interventions (physics, chemistry, biology, music, national language).
\item \emph{Matching.} Twin odd-graceful labelling, reciprocal-inverse labellings, other matchings mentioned here, and so on.
\item \emph{Combinatorics.} Twin type of labellings, such as twin odd-graceful and twin  odd-elegant labellings. Various combinatorics induce many labellings, such as 6C-labellings.
\item \emph{Closure to property and configuration.} Labellings are closed to particular graphs, or graph properties, or labelling properties, and so on.
\item \emph{Connections with others.} There are: (i) canonical mathematical operations, such as addition, subtraction, multiplication and division; (ii) graph operations, such as union, intersection, split, subdivision, and so on; (iii) advanced algebraic operations, such as group, ring and field; (iv) text-based passwords; (v) between labellings, such as equivalence, transformation \emph{etc.}
\item \emph{Compound.} Graphs are labeled by Topsnut-gpws and graphic groups \emph{etc.}
\item \emph{Transformation.} For example, $f$ is set-ordered on $(X,Y)$, so we have an affine transformation $g$ defined by $g(x)=af(x)+b$ for $x\in X$, $g(y)=cf(y)+d$ for $y\in Y$.
\item \emph{Generalization and diversity.} What is a hyperlabelling? What is a network labelling? What is an random labelling? What is a functional (chemistry, physical, biological) labelling?
\end{asparaenum}

\section{Algebraic group/set matching partitions}

Many problems of Topsnut-gpws can be transformed into algebraic problems, such as set problems and algebraic group problems, etc. However, the research of algebraic group/set problems differs greatly from that of Topsnut-gpws. On the other hands, investigating the problems proposed in this subsection does not need knowledge of graph theory, only basic mathematical knowledge.

\subsection{Set matching partitions}

Set matching partition is a natural phenomenon in mathematics, such as an integer set $[1,10]$ contains two subsets $[1,10]^o=\{1,3,5,7,9\}$ and $[1,10]^e=\{2,4,6,8,10\}$. Clearly, $[1,10]^o\cup [1,10]^e=[1,10]$, we say $[1,10]^o$ and $[1,10]^e$ are matching to each other, they are a \emph{set matching partition} of $[1,10]$. Many graph labellings are related with set problems. For two integers $p\geq 2$ and $q\geq 1$, the previous labelling definitions enable us to obtain the following set problems:

\begin{asparaenum}[Set-1. ]
\item From Definition \ref{defn:relaxed-Emt-labelling}: Partition $[1, p+q]$ into two disjoint subsets $V$ and $E$ with $V\cup E=[1, p+q]$ such that: (i) each $c\in E$ corresponds to distinct $a,b\in V$ holding $c=|a-b|$ true; (ii) there exists a constant $k$, each $c'\in E$ matches with distinc $a',b'\in V$ holding $a'+c'+b'=k$ true. We call $(V,E)$ a \emph{relaxed edge-magic total matching partition} of $[1, p+q]$.

\qquad $\ast$ Find all possible relaxed edge-magic total matching partitions $(V,E)$ of $[1, p+q]$.

\item From Definition \ref{defn:Oemt-labelling}: Select two subsets $V,E\subset [1,2q-1]$ with $E=[1,2p-1]^o$ such that there is a constant $k$, each $c\in E$ corresponds two distinct $a,b\in V$ holding $a+b+c=k$ true. We call $(V,E)$ an \emph{odd-edge-magic matching partition} of $[1,2q-1]$.

 \qquad $\ast$ Find all possible odd-edge-magic matching partitions of $[1,2q-1]$.

\item From Definition \ref{defn:relaxed-Oemt-labelling}: Selecting two subsets $V,E\subset [1,2q-1]$ with $E=[1,2p-1]^o$ holds true: (i) each $c\in E$ corresponds $a,b\in V$ to form an \emph{ev-matching} $(acb)$; (ii) each $c\in E$ corresponds to $z\in E$ with the ev-matching $(xzy)$ such that $c=|x-y|$; (iii) each $c\in E$ with the ev-matching $(acb)$ corresponds to $c'\in E$ with the ev-matching $(a'c'b')$ such that $$(|a-b|-c)+(|a'-b'|-c')=0.$$ We call $(V,E)$ an \emph{ee-difference odd-edge-magic matching partition} of $[1,2q-1]$.

 \qquad $\ast$ Find all possible ee-difference odd-edge-magic matching partitions $(V,E)$ of $[1,2q-1]$. Here, each ev-matching $(acb)$ corresponds an edge $c$ of a graph, where the edge $c$ has two ends $a,b$.

\item From Definition \ref{defn:6C-labelling}: Partition $[1,p+q]$ into two subsets $V,E$ satisfies: (i) each $c\in E$ corresponds $a,b\in V$ to form an \emph{ev-matching} $(acb)$; (ii) (e-magic) each $c\in E$ with the ev-matching $(acb)$ hold $c+|a-b|=k$ true; (iii) (ee-difference) each $c\in E$ corresponds to $z\in E$ with the ev-matching $(xzy)$ such that $c=|x-y|$; (iv) (ee-bandwiden) each $c\in E$ with the ev-matching $(acb)$ corresponds to $c'\in E$ with the ev-matching $(a'c'b')$ such that $(|a-b|-c)+(|a'-b'|-c')=0$; (iv) (EV-ordered) $\max V<\min E$ (or $\max V>\min E$); (v) (ve-matching) each ev-matching $(acb)$ matches with another ev-matching $(uwv)$ such that $a+w=k'$ or $b+w=k'$, $k'$ is a constant, except the \emph{singularity} $\lfloor \frac{p+q+1}{2}\rfloor $. We call $(V,E)$ a \emph{6C-partition} of $[1,p+q]$.

 \qquad For a given 6C-partition $(V,E)$ of $[1,p+q]$, if there exists another 6C-partition $(V',E')$ of $[1,p+q]$ such that $$V\setminus (V\cap V')=E',~E=V'\setminus (V\cap V')$$ for $V\cap V'=\{\lfloor \frac{p+q+1}{2}\rfloor \}$, we get a partition $(V\cup E', E\cup V)$, and call it a \emph{6C-complementary matching partition} of $[1,p+q]$.

 \qquad $\ast$ Find all possible 6C-partitions $(V,E)$ of $[1,p+q]$, and all possible 6C-complementary matching partitions $(V\cup E', E\cup V)$.

\item From Definition \ref{defn:Dgemm-labelling}: Partition $[0,p+q-1]$ into two subsets $V,E$ with $E=[1,q]$ and $V\subseteq [0,q-1]$ satisfies: (i) each $c\in E$ corresponds $a,b\in V$ to form an \emph{ev-matching} $(acb)$; (ii) (ee-difference) each $c\in E$ corresponds to $z\in E$ with the ev-matching $(xzy)$ such that $c=|x-y|$; (iii) (ee-bandwiden) each $c\in E$ with the ev-matching $(acb)$ corresponds to $c'\in E$ with the ev-matching $(a'c'b')$ such that $(|a-b|-c)+(|a'-b'|-c')=|a-b|+|a'-b'|-(c+c')=0$; (iv) there exists a constant $k$ such that each $c\in E$ with its ev-matching $(acb)$ holds $c+|a-b|=k$ true; (v) each $c\in E$ corresponds another $c'\in E$ with its ev-matching $(a'c'b')$ such that $c+a'=p$ or $c+b'=p$. We call $(V,E)$ an \emph{ee-difference graceful-magic matching partition} of $[0,p+q-1]$.

 \qquad $\ast$ Find all possible ee-difference graceful-magic matching partitions of $[0,p+q-1]$.

\item From Definition \ref{defn:ve-exchanged-labelling}: Partition $[1, p+q]$ into two disjoint subsets $V$ and $E$ with $V\cup E=[1, p+q]$ such that each $c\in E$ corresponds to distinct $a,b\in V$ holding $c=|a-b|$ true, and there exists a constant $k$ satisfying $a+c+b=k$ for each $c\in E$ which corresponds to distinct $a,b\in V$. We call $(V,E)$ an \emph{edge-magic graceful matching partition} of $[1, p+q]$. If $(E, V)$ is another edge-magic graceful matching partition of $[1, p+q]$, we say $(V,E)$ (resp. $(E, V)$) to be a \emph{ve-exchanged matching partition} of $[1, p+q]$.

 \qquad $\ast$ Find all possible edge-magic graceful matching partitions of $[1, p+q]$, and all possible ve-exchanged matching partitions.
\item If there are two sets $V\subseteq [0,q]^2$ (or $[0,2q-1]^2$) and $E\subseteq [1,q]$ (or $[1,2q-1]$) such that each $c\in E$ with its ev-matching $(acb)$ holds $c=|a-b|$ true, where $a\in A\in V$ and $b\in B\in V$ with $A\cap B=\emptyset$, then we call $(V,E)$ a \emph{v-set e-proper graceful (or odd-graceful) matching partition}.

\item From the twin odd-graceful/odd-elegant labellings: Partition $[0,2q]$ into two subsets $S_1,S_2$ such that $S_1\subset [0,2q-1]$, $S_2\subset [1,2q]$, $|S_1\cap S_2|=1$ and $S_1\cup S_2=[0,2q]$. For $E_1=E_2=[1,2q-1]^o$, each $c_i\in E_i$ corresponds two numbers $a_i,b_i\in S_i$ holding $c_i=|a_i-b_i|$ true (or $c_i=a_i+b_i~(\bmod~2q)$) with $i=1,2$. We call $(S_1,S_2)$ a \emph{twin odd-graceful (or odd-elegant) matching partition} of $[0,2q]$.

 \qquad $\ast$ Characterize twin odd-graceful (or odd-elegant) matching partitions, and find them.

\item From Definitions \ref{defn:difference-sum-labelling} and \ref{defn:felicitous-sum-labelling}: Select a subset $E\subset [0,p-1]$ such that each $c\in E$ corresponds two distinct $a,b\in V=[0,p-1]$ to hold $c=|a-b|$ true (or $c=a+b~(\bmod~|E|)$), we call $f_E=(V,E)$ a \emph{graph matching partition}, and call $S_{um}(G,f_E)=\sum_{c\in E}|a-b|$ a \emph{difference-sum number} (or $F_{um}(G,f_E)=\sum_{c\in E}(a+b)~(\bmod~|E|)$ is a \emph{felicitous-sum number}).

 \qquad $\ast$ Determine $\max_{f_E} S_{um}(G,f_E)$ (profit) and $\min_{f_E} S_{um}(G,f_E)$ (cost) over all difference-sum matching partitions $f_E=(V,E)$ of $[0,p-1]$. Find $\max_{f_E} F_{um}(G,f_E)$ and $\min_{f_E} F_{um}(G,f_E)$ over all felicitous-sum matching partitions $f_E=(V,E)$ of $[0,p-1]$.
\end{asparaenum}

It may be interesting to consider such algebraic groups on the above set partition problems.

\subsection{Matching partitions of algebraic matrices}

We introduce an algebraic expression of a Topsnut-gpw $G$ being a $(p,q)$-graph as follows:
\begin{defn}\label{defn:Topsnut-matrix}
A \emph{Topsnut-matrix} $A_{vev}(G)$ of  a Topsnut-gpw $G$ being a $(p,q)$-graph is defined as
\begin{equation}\label{eqa:a-formula}
\centering
A_{vev}(G)= \left(
\begin{array}{ccccc}
x_{1} & x_{2} & \cdots & x_{q}\\
e_{1} & e_{2} & \cdots & e_{q}\\
y_{1} & y_{2} & \cdots & y_{q}
\end{array}
\right)=(X~W~Y)^{T}
\end{equation}\\
where
\begin{equation}\label{eqa:two-vectors}
{
\begin{split}
&X=(x_1 ~ x_2 ~ \cdots ~x_q), W=(e_1 ~ e_2 ~ \cdots ~e_q)\\
&Y=(y_1 ~ y_2 ~\cdots ~ y_q),
\end{split}}
\end{equation}
and $G$ has another \emph{Topsnut-matrix} $A_{vv}(G)$ defined as $A_{vv}(G)=(X,Y)^{T}$, where $X,Y$ are called \emph{vertex-vectors}, $W$ is called \emph{edge-vector}, such that $e_i=x_iy_i$ is an edge of $G$ for $i\in [1,q]$.\qqed
\end{defn}

So, $A_{vev}(G)$ is a Topsnut-matrix, and $A_{vv}(G)$ is a Topsnut-matrix. Clearly, such Topsnut-matrices are easily saved in computer, and produce quickly text-based passwords. For example, we have the following four Topsnut-matrices $A_1,A_2,B_1$ and $B_2$ from Fig.\ref{fig:matrix-A1} to Fig.\ref{fig:matrix-B2}:

\begin{figure}[h]
\centering
\includegraphics[height=1.8cm]{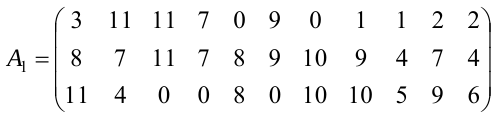}\\
\caption{\label{fig:matrix-A1} {\small A Topsnut-matrix $A_1$ for $G_a$ in the perfect Max-min difference-sum matching partition $A$ shown in Fig.\ref{fig:one-conjecture}.}}
\end{figure}

\begin{figure}[h]
\centering
\includegraphics[height=1.8cm]{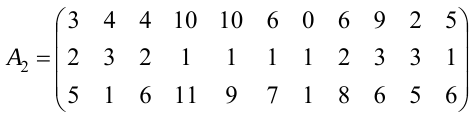}\\
\caption{\label{fig:matrix-A2} {\small A Topsnut-matrix $A_2$ for $G_d$ in the perfect Max-min difference-sum matching partition $A$ shown in Fig.\ref{fig:one-conjecture}.}}
\end{figure}

\begin{figure}[h]
\centering
\includegraphics[height=2cm]{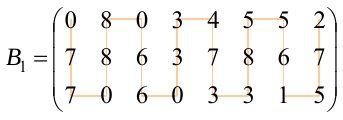}\\
\caption{\label{fig:matrix-B1} {\small A Topsnut-matrix $B_1$ for $F_b$ in the perfect Max-min felicitous-sum matching partition $B$ shown in Fig.\ref{fig:one-conjecture}.}}
\end{figure}

\begin{figure}[h]
\centering
\includegraphics[height=2cm]{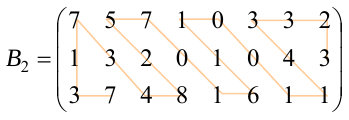}\\
\caption{\label{fig:matrix-B2} {\small A Topsnut-matrix $B_1$ for $F_c$ in the perfect Max-min felicitous-sum matching partition $B$ shown in Fig.\ref{fig:one-conjecture}.}}
\end{figure}

We point out: (i) A Topsnut-matrix $A_{vev}(G)$ is not unique for expressing a Topsnut-gpw $G$, in other words, a Topsnut-gpw $G$ many have two or more Topsnut-matrices; (ii) Topsnut-matrices differ from popular algebraic matrices, since Topsnut-matrices are only the expression of labelled vertices joined by labelled edges. Clearly, we need some new algebraic operations on Topsnut-matrices.

Let $D(A)$ be the matrix of a perfect Max-min difference-sum matching partition $A=\odot_{12}\langle G_a,G_d\rangle $ shown in Fig.\ref{fig:one-conjecture}. So, $D(A)$ is a Topsnut-matrix, denoted directly as $D(A)=\odot_{12}\langle A_1,A_2\rangle$ (see Fig.\ref{fig:matrix-A1} and Fig.\ref{fig:matrix-A2}), called a \emph{matrix matching partition}. Similarly, the matrix $D(B)$ of the perfect Max-min felicitous-sum matching partition $B=\odot_{9}\langle F_b,F_c\rangle $ shown in Fig.\ref{fig:one-conjecture} is a Topsnut-matrix, and we have another \emph{matrix matching partition} $D(B)=\odot_{8}\langle B_1,B_2\rangle$ (see Fig.\ref{fig:matrix-B1} and Fig.\ref{fig:matrix-B2}).

Along the orange line in the matrix $B_1$, we can get a text-based password
$$T_{ext}(B_1)=077088066033473385561572$$ and another text-based password
$$T_{ext}(B_2)=731734825701611001143323$$ obtained along the orange line in the matrix $B_2$. Obviously, it is not easy to reconstruct the perfect Max-min felicitous-sum matching partition $B$ shown in Fig.\ref{fig:one-conjecture} from $T_{ext}(B_1)$ and $T_{ext}(B_2)$, even it is impossible if Topsnut-gpws with large numbers of vertices and edges.

By the vertex-split and vertex-identifying operations, as well as the edge-split and edge-identifying operations, we can define algebraic operations on Topsnut-matrices of $(p,q)$-graphs that are topological structures of Topsnut-gpws, such as $D(A)=\odot_{12}\langle A_1,A_2\rangle$ and $D(B)=\odot_{8}\langle B_1,B_2\rangle$ obtained by the \emph{vertex-identifying operation of $(3\times q)$-matrices}.

\subsection{Topsnut-matchings made by graphic groups}

Let $T^{odd}_{group}$ be a set of odd-graceful Topsnut-groups. We define a labelling $f:V(G)\rightarrow T^{odd}_{group}$ for a $(p,q)$-graph $G$, and set $f(uv)=\odot\langle f(u),f(v)\rangle$ to be a matching of two odd-graceful Topsnut-groups $G^{odd}_i$ and $G^{odd}_j$, here, each $T_i\in G^{odd}_i$ matches with $T_j\in G^{odd}_j$ such that $\odot \langle T_i,T_j\rangle$ is just an odd-graceful Topsnut-matching, and vice versa.

For encrypting a network by graphic groups we show a simple example in Fig.\ref{fig:path-gpw-12}, Fig.\ref{fig:network-group-0} and Fig.\ref{fig:network-group-1}. We have an operation defined by
\begin{equation}\label{eqa:additive-operation-abelian}
[f_i(x)+f_j(x)-f_k(x)]~(\bmod~13)=f_{\lambda}(x)
\end{equation}
for each element $x\in V(G)\cup E(G)$ shown in Fig.\ref{fig:path-gpw-12}, where $\lambda=i+j-k~(\bmod~13)$, and call (\ref{eqa:additive-operation-abelian}) ``\emph{additive operation}''. We can see that there are many ways to realize a network encrypted by a graphic group, since there are two or more ways to join $G_i$ with $G_j$ by an edge (allow by two or more edges). Thereby, we have obtained many encrypted networks.

\begin{figure}[h]
\centering
\includegraphics[width=6cm]{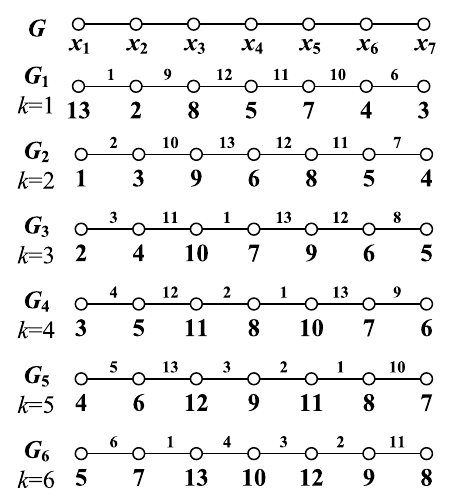}\\
\includegraphics[width=6cm]{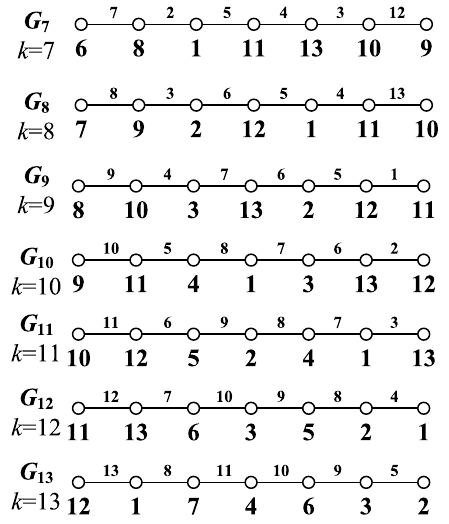}\\
\caption{\label{fig:path-gpw-12} {\small A graphic group based on a path $G$ and the edge-magic graceful labelling, each $G_i$ admits a pan-edge-magic graceful labelling $f_i$ under modulo 13.}}
\end{figure}

\subsection{Algebraic groups from Topsnut-gpws, Topsnut-matrices and text-based passwords}

We have known that a Topsnut-gpw $G$ has its Topsnut-matrix $A(G)$ which induces a text-based password $D(G)$. So this  Topsnut-gpw $G$ and its Topsnut-matrix $A(G)$, as well as the text-based password $D(G)$ can produce three Abelian additive groups by the additive operation shown in (\ref{eqa:additive-operation-abelian}), we call them \emph{Topsnut-group, Topsnut-matrix group and Text-pw group}, respectively. If a Topsnut-gpw $G$ matches with another Topsnut-gpw $H$, so two Topsnut-groups induced by $G$ and $H$ match with each other. More results on such groups can be found in \cite{Sun-Zhang-Zhao-Yao-2017} and \cite{Yao-Mu-Sun-Zhang-Wang-Su-Ma-2018}.

\begin{figure}[h]
\centering
\includegraphics[height=2.8cm]{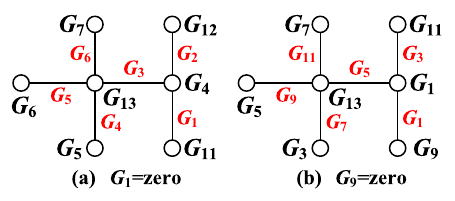}\\
\caption{\label{fig:network-group-0} {\small A tree admits: (a) a graceful group labelling based on the zero $G_1$ shown in Fig.\ref{fig:path-gpw-12}; (b) an odd-graceful group labelling based on the zero $G_9$ shown in Fig.\ref{fig:path-gpw-12}.}}
\end{figure}

\begin{figure}[h]
\centering
\includegraphics[height=7.4cm]{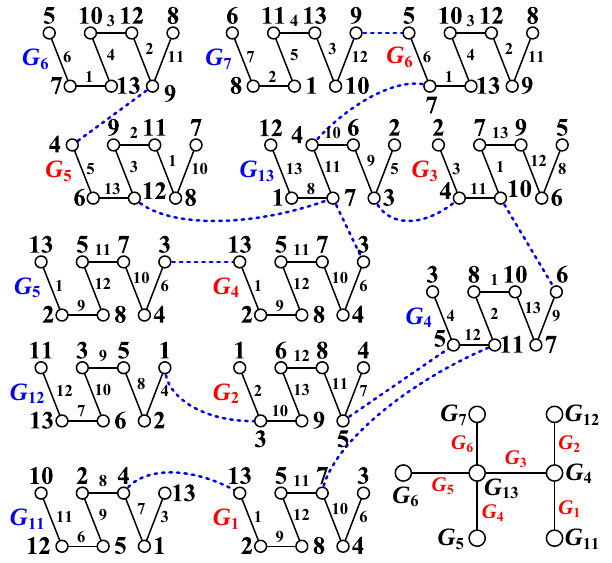}\\
\caption{\label{fig:network-group-1} {\small A network encrypted by a graphic group shown in Fig.\ref{fig:path-gpw-12}.}}
\end{figure}

\section{Researching problems}

For further researching Topsnut-matchings we propose the following problems:
\begin{asparaenum}[\textbf{Pro}-1. ]
\item (A complete graph obtained from labeled trees) Given disjoint trees $T_1,T_2, \dots, T_m$ with
$$\sum^m_{i=1}|E(T_i)|\geq \frac{1}{2}n(n-1).$$ Can we find a labelling $f_i$ for each tree $T_i$ such that $f_i:V(T_i)\rightarrow [0,n-1]$ and $$\{|f_i(u)-f_i(v)|:uv\in E(T_i)\}=[1,|E(T_i)|],$$ and identify the vertices of $\bigcup^m_{i=1}V(T_i)$ having the same labels into one, the resulting graph is just $K_n=\odot\langle T_i\rangle^m_1$, or $K_n=\ominus(T_i)^m_1$, or $K_n=\bigcup^m_{i=1} T_i$?

\item For any odd-graceful graph $H$ of $p$ vertices, does there exists a pan-odd-graceful Topsnut-matching team $\odot_1\langle H,T_i\rangle ^p_1$? Or consider this $H$ as a lobster first. Find conditions for the perfect odd-gracefully Topsnut-matching team $\odot_1\langle H,T_i\rangle ^p_1$ with $T_i\cong T_j$ for $i\neq j$. Find nontree graphs which induce pan-odd-graceful Topsnut-matching teams.

\item  Plant the concept of pan-odd-graceful Topsnut-matching team $\odot_1\langle H,T_i\rangle^p_1$ on other graph labellings.

\item For a given lobster $T$, find another lobster $T'$ such that $\odot  \langle T,T'\rangle$ admits a twin odd-graceful labelling (or a twin odd-elegant labelling) (\cite{Wang-Xu-Yao-2017-Twin, Wang-Xu-Yao-2017}).

\item For a given $(p,q)$-tree $G$ admitting a 6C-labelling $f$, find all possible $(p,q)$-tree $H$ admits a 6C-labelling $g$ such that $\odot\langle G,H \rangle $ are 6C-complementary matchings.
\item Find conditions for a connected graph $G$ to be a multiple-tree matching partition $G=\oplus_F\langle T_i\rangle ^m_1$ with $m\geq 2$.
\item Find all possible odd-graceful Topsnut-matchings $\odot \langle G,H \rangle $ for a given $(p,q)$-graph $G$ admitting odd-graceful labellings.
\item Determine v-set e-proper graceful/odd-graceful labellings of Euler graphs.

\item Determine the conditions for $A_M$ and $B_q$ and any $a\in A_M$ corresponds two numbers $a^*\in A_M$ $b^*\in B_q$ such that $b^*=|a-a^*|$. Then determine such sequence pair $(A_M,B_q)$ defined in Definition \ref{defn:sequence-labelling} such that the sequence type of labellings defined in Definition \ref{defn:sequence-labelling} hold true on graphs. Clearly, if any $b\in B_q$ corresponds two  numbers $a',a''\in A_M$ such that $b=|a'-a''|$, then there exists at least a forest $T$ admitting a graceful sequence-$(A_M,B_q)$ labelling defined in Definition \ref{defn:sequence-labelling}.

\item Find connected graphs $G$ such that for any integer $M$ holding
\begin{equation}\label{eqa:c3xxxxx}
\min_f S_{um}(G,f)<M<\max_f S_{um}(G,f)
\end{equation} true, then there exists a difference-sum labelling $h$ of $G$ with $M=S_{um}(G,h)$.

\item Find connected graphs $G$ such that for any integer $M$ holding
\begin{equation}\label{eqa:c3xxxxx}
\min_f F_{um}(G,f)<M<\max_f F_{um}(G,f)
\end{equation} true, then there exists a felicitous-sum labelling $h$ of $G$ with $M=F_{um}(G,h)$.

\item For particular graphs $G$, compute the exact values of $\min_f S_{um}(G,f)$ and $\max_f S_{um}(G,f)$.

\item For a given graph $G$, find all graphs $H$ for forming set-ordered matching graphs $G\ominus H$ with $H\not\cong G$.

\item For an odd-graceful graph $G$, find all matching graphs $H$ such that $\odot \langle G,H\rangle $ admitting twin odd-graceful labellings.

\item Consider other v-set e-proper $\varepsilon$-labellings of a complete graph $K_n$, where $\varepsilon \in \{$edge-magic total labelling, odd-elegant labelling, harmonious labelling, the labellings defined in this paper$\}$. For example:

 \quad (i) A v-set e-proper felicitous labelling $(F,f)$ of a $(p,q)$-graph $G$ is defined as: $F:V(G)\rightarrow [0,q-1]^2$ with $F(x)\cap F(y)=\emptyset$ for distinct $x,y\in V(G)$, and $f:E(G)\rightarrow [0,q-1]$ holding $f(E(G))=[0,q-1]$ and $f(uv)=a_u+a_v~(\bmod ~q)$ true with $a_u\in F(u)$ and $a_v(v)\in F(v)$.

 \quad Does $K_n$ admits a v-set e-proper felicitous labelling?

 \quad (ii) A v-set e-proper edge-magic total labelling $(F,f)$ of a $(p,q)$-graph $G$ is defined by $F:V(G)\rightarrow [1,M]^2$ with $p+q\leq M$ and $F(x)\cap F(y)=\emptyset$ for distinct $x,y\in V(G)$, and $f:E(G)\rightarrow [1,M]$ with $f(uv)\neq f(xy)$ for any two edges $uv,xy\in E(G)$, and there exists a constant $k$ such that $$a_u+f(uv)+a_v=k$$ for any edge $uv\in E(G)$ with $a_u\in F(u)$ and $a_v\in F(v)$.

 \quad Does $K_n$ admits a v-set e-proper edge-magic total labelling? Find the parameter $E_{mt}(G)=\min_{(F,f)}\{M\}$ over all v-set e-proper edge-magic total labellings of $G$.

\item If we can split a connected graph admitting a v-set e-proper graceful labelling into a tree, then characterize this graph and its possible v-set e-proper graceful labellings.
\item Find conditions for a connected graph $G$ that can be split into caterpillars, or lobsters, such that $G$ admits a v-set e-proper $X$-labelling, where $X$  is a graph labelling admitted by caterpillars, or lobsters (see Theorem \ref{thm:caterpillar-lobster-v-set-e-proper}).

\item For each $p\geq 2$, find a $(p,q)$-graph $G=\odot_f\langle G_i\rangle ^m_1$ defined in definition \ref{defn:multiple-graph-matching}, such that $q$ is the largest edge number on such $(p,q)$-graphs when $p$ is fixed. We can add other restrictions: (i) each $G_i$ is a spanning subgraph of $G$; (ii) $E(G_i)\cap E(G_j)=E^*\subset E(G)$ for $i\neq j$ and $E^*\neq \emptyset$; (iii) each $G_i$ is a tree, or an Euler graph, or a bipartite graph, and so on.

\item A $(p,q)$-graph $G$ and a $(q,p)$-graph $H$ admit two edge-magic graceful labellings $f$ and $g$, respectively, and $f$ and $g$ are reciprocal inverse because $f(E(G))=g(V(H))\setminus X^*$ and $f(V(G))\setminus X^*=g(E(H))$ for $X^*=f(V(G))\cap g(V(H))$. Find such pairs of graphs $G$ and $H$, and characterize them.
\item Find reciprocal complementary (reciprocal-inverse matching) $G=\odot \langle T,G\rangle $ for a fixed graph $T$, where $T$ and $G$ admit reciprocal-inverse labellings $f$ and $g$, respectively, such that $$f(E(T))=g(V(G))\setminus X^*\textrm{ and }f(V(T))\setminus X^*=g(E(G))$$ for $X^*=f(V(T))\cap g(V(G))$.
\item If a total coloring $g$ of a graph $G$ arrives at $B_{tol}(G,g)=\min_fB_{tol}(G,f)$, is there $\chi ''(G)=|\{g(x):x\in V(G)\cup E(G)\}|$?
\item For any connected subgraph $H$ of a connected graph $G$, does there exists $\max_g S_{um}(H,g)\leq \max_f S_{um}(G,f)$?
\item Determine connected graphs having a group of consecutive difference proper vertex colorings, or a group of consecutive sum proper vertex colorings.
\item Find connected graphs admitting one of the edge-magic proper total coloring and the equitably proper total coloring.
\item It is not difficult to verify $\chi''_{ved}(K_n)\leq 2n$ and $\chi''_{ves}(K_n)\leq 2n-1$. Does any tree $T$ hold $$\chi''_{ved}(T)\leq \Delta(T)+4$$ and $$\chi''_{ves}(T)\leq 2\Delta(T)+1$$ true?
\item Let $f: V(G)\cup E(G)\rightarrow [1,\chi''(G)]$ be a proper total coloring of a graph $G$, and let $$f^*(E(G))=\{f(u)+f(uv)+f(v):uv \in E(G)\}.$$ Characterize $G$ if $f^*(E(G))$ is a consecutive integer set $[a,b]$.
\item  In \cite{YAO-SUN-WANG-SU-XU2018arXiv}, the authors defined: ``Let $\eta$-labeling be a given graph labelling, and let a connected graph $G$ admit an $\eta$-labeling. If every connected proper subgraph of $G$ also admits a labelling like  $\eta$-labeling, then we call $G$ a \emph{perfect $\eta$-labeling graph}.'' Caterpillars are perfect $\eta$-labeling graphs if these $\eta$-labelings are listed in this article, and each lobster is a perfect (odd-)graceful labeling graph. They ask for: \emph{If every connected proper subgraph of a connected graph $G$ admits an $\eta$-labelling, then does $G$ admits this $\eta$-labelling too}? Clearly, a perfect $\eta$-labeling graph (like an elder generation) can be used to produce a crowd of Topsnut-GPWs (like son generations).
\end{asparaenum}

\section{Conclusion}

We have known that Topsnut-matching is a larger topic in researching Topsnut-gpws, \emph{nature-inspired passwords}. Results and techniques of graph theory are proven to be powerful in designing and researching Topsnut-gpws, since there are no polynomial algorithms for many of these results and techniques. Many of the graph labellings introduced here match with mathematical conjectures, so they may provide computationally unbreakable for our Topsnut-gpws. It is hopeful to let more people use Topsnut-gpws and \emph{pan-Topsnut-gpws} (allow label vertices and edges with non-mathematical elements) for protecting their information and profits in networks ( \cite{Yao-Mu-Sun-Zhang-Wang-Su-2018}, \cite{Hongyu-Wang-2018-Doctor-thesis}). There are over 200 graph labellings introduced in \cite{Gallian2016}, and more new graph labellings emerge everyday. It is time to consider \emph{Graph Labelling Analysis} as a subbranch of graph theory. So, we try doing some exploring work here, although we have two hands empty on this topic.

Matching can help us to design Topsnut-gpws for one public key vs one private key, one public key vs two or more private keys, and more public keys vs more private keys. Matching opens a window for us to understand something new in cryptography. It is very important that matching is just one of mathematical principles. Almost mathematical operations have their own matching operations. The graph labellings first defined or introduced here match with other existing graph labelllings, and can be shown to be related with mathematical conjectures, or open problems.

Researching Topsnut-matching can derive two interesting topics: one is \emph{set matching partition} to number theory, and another is about labeled graphs for constructing large scale of graphs with labellings, which differs from finding labellings to unlabeled graphs. We have listed possible researching problems for further studying works on Topsnut-gpws, and hope to find more something new and to do more theoretical works on Topsnut-gpws. We try to use Topsnut-groups to build up so-called network passwords for encrypting a network with thousand and thousand nodes (vertices). So we have investigated one of Topsnut-groups, called \emph{Abelian additive graphic group} (\emph{graphic group} for short). This type of graphic groups based on addition operation processes a particular property: ``Every element in a graphic group can be regarded as ``zero'' of the graphic group, so we can call it an \emph{every-zero graphic group}''. Unfortunately, we do not discover graphic group based on multiplication operation. It may be a way to find more graph labellings of a graph from connection between two or more graphic groups.

Several new colorings and new parameters on proper total colorings have been introduced and investigated. We have found that the difference-sum labelling (\emph{extremal labelling}) can be admitted by every graph, so then it breaks down the case of no labelling admitted by each graph. Thereby, we are motivated from the difference-sum labelling and know that there are many\emph{ extremal labellings} like the difference-sum labelling, which mean that we may touch a new subbranch of graph labellings.

The above research works on two different areas motivate us to think of the biological combination of human being and AI machine in current development of the world, rather than AI machines only that will take a long time to success. An application project supported by mathematics like passwords depends on mathematics going deep into and continuous improvement, how long will it support last, how far can the project go.

\section*{Acknowledgment}

The author, \emph{Bing Yao}, is delight for supported by the National Natural Science Foundation of China under grants 61163054, 61363060 and 61662066. Bing Yao, cordially, thanks every member of \emph{Topological Graphic Passwords Symposium} in the second semester of 2017-2018 academic year for their effective discussion and constructive suggestions, and part of members of the Symposium were supported by Scientific research project of Gansu University under grants 2016A-067, 2017A-047 and 2017A-254.



%


\end{document}